\documentclass{article}
\usepackage{iclr2027_conference,times}

\usepackage{amsmath}
\usepackage{amssymb}
\usepackage{amsthm}
\usepackage{algorithm}
\usepackage[noend]{algpseudocode}
\usepackage{booktabs}
\usepackage{longtable}
\usepackage{array,tabularx}
\usepackage{needspace}
\usepackage{graphicx}
\usepackage{multirow}
\usepackage{placeins}
\usepackage[hidelinks]{hyperref}
\usepackage{url}

\newtheorem{proposition}{Proposition}

\newtheorem{theorem}{Theorem}

\title{Beyond the Beam: Constructive Repair\\
and Candidate Completion for\\
Generative Recommendation}

\author{%
\begin{minipage}[t]{0.97\textwidth}
\centering
\normalsize
Zijun Zhao\textsuperscript{1}\quad
Peng Zhang\textsuperscript{2}\quad
Gang Zhang\textsuperscript{1}\\[2pt]
Yuanchi Ma\textsuperscript{3}\quad
Hui He\textsuperscript{4}\quad
Zhendong Niu\textsuperscript{1}\thanks{Corresponding author.}\\[5pt]
\normalfont\small
\textsuperscript{1}\,Beijing Institute of Technology\quad
\textsuperscript{2}\,China Meteorological Administration\\[2pt]
\textsuperscript{3}\,Tsinghua University\quad
\textsuperscript{4}\,Singapore Management University\\[4pt]
\href{mailto:zhaozijun@bit.edu.cn}{\texttt{zhaozijun@bit.edu.cn}}\quad
\href{mailto:zniu@bit.edu.cn}{\texttt{zniu@bit.edu.cn}}
\end{minipage}%
}
\hypersetup{pdfauthor={Zijun Zhao, Peng Zhang, Gang Zhang, Yuanchi Ma, Hui He, Zhendong Niu}}

\iclrfinalcopy
\hypersetup{
  pdftitle={Beyond the Beam: Constructive Repair and Candidate Completion for Generative Recommendation},
  pdfkeywords={recommender systems, generative recommendation, semantic identifiers}
}

\begin{document}

\maketitle
% Override the conference publication header set by the original style.
\lhead{Preprint}

\begin{abstract}
Generative recommenders retrieve items by generating identifiers, but a valid
identifier can remain outside the beam after catalog expansion. This raises
two connected questions: which failures can identifier assignment repair,
and how should retrieval proceed beyond the initial beam? We characterize
assignment repair with a fixed generator and retained old identifiers.
Output-invariance certificates identify failures shared by all admissible
assignments. Under a common effective prefix, coupled support and ranking
constraints give the exact feasible interval of new-item counts for target
recovery. Building on this characterization, Beyond the Beam (BB) obtains
minimum-replacement repairs through an integral flow formulation, selects
a shared map and adapts the generator. At inference, generative likelihood and collaborative evidence
define one score for ranking, candidate priority and stopping. Retained
prefix bounds guide candidate completion and certify its global Top-$K$
when the stopping condition is met. Exhaustive finite-catalog evaluation
confirms construction in every feasible case. Across three Amazon Reviews
categories and three random seeds, the full T5 procedure improves mean
Recall@10 by 15.5--46.3\% and NDCG@10 by 15.2--44.4\% over the best-performing
evaluated generative baseline for each dataset and metric. Matched controls show
that shared construction and adaptation improve new-target ranking and
certification efficiency on Beauty and Toys. Combined scoring and candidate
completion improve NDCG@10 across all three datasets with both T5 and
decoder-only LC-Rec.

\end{abstract}

\section{Introduction}
\label{sec:introduction}

Generative recommenders retrieve items by generating compact identifiers
\citep{rajput2023tiger}. After catalog expansion, a new item can have a valid
identifier and strong collaborative support yet be omitted by beam search.
Reranking cannot return an item outside its candidate pool. This failure
motivates two questions: when can changing identifiers repair retrieval
while preserving old codes, and how can inference evaluate promising items
beyond the initial beam? Appendix~\ref{app:worked-example} illustrates both
questions with a concrete example.

Existing approaches update the index, tokenizer or model.
Joint tree/model learning optimizes item-to-leaf matching
\citep{zhu2019joint}; DREAM uses multi-context evidence to choose candidate
identifiers \citep{guan2026dream}; DACT incorporates collaborative information
into tokenization \citep{feng2026dact}. These procedures choose an assignment and measure
its performance. The complementary question concerns the entire admissible
family: which failures persist across every assignment, and which admit a
repair while retaining old identifiers and recommendation eligibility?

A target's identifier score alone does not determine repair feasibility.
An occupied leaf activates prefixes that compete for a finite beam.
Supporting the target's path can therefore displace other paths, while
filling the catalog can introduce higher-ranked terminal competitors.
The challenge is to characterize which catalogs jointly satisfy these
support and ranking requirements while returning the target.

After a common effective prefix, we characterize exactly which catalogs
recover a specified target. Prefix support sets a lower bound on new-item
counts; ranking competition and catalog capacity set an upper bound.
Every feasible count admits an explicit repair. A complementary
output-invariance certificate identifies queries whose ordered predictions
remain unchanged across admissible assignments, yielding bounds on population
improvement.

Building on this feasible family, we develop Beyond the Beam (BB), connecting
constructive repair with bound-guided candidate completion. For each feasible state
and target, minimum-replacement repair retains as much of the current leaf
occupancy as the constraints allow. Disjoint support groups give an integral
minimum-cost-flow solution. We turn local repairs into shared-map proposals,
compare them across re-encoded training histories, and adapt the generator
to the selected map.

The shared map is learned across training contexts, while candidate access
remains query-dependent. At inference, collaborative evidence can favor items
outside the generator's initial beam. We therefore combine generator likelihood
with an item-level correction and evaluate additional candidates under this
common score. Adding the known correction to retained prefix scores yields
bounds for candidate priority and global Top-$K$ certification.

Our contributions connect repair feasibility and minimum leaf replacement
with candidate completion under the combined score. Exhaustive finite catalogs
and direct model decoding test the constructive guarantees. Three-seed T5
controls show that shared construction and adaptation improve new-target
ranking and certification efficiency on Beauty and Toys. Combined scoring
and candidate completion improve NDCG@10 across all three datasets with
both T5 and decoder-only LC-Rec.

\section{Problem Formulation}
\label{sec:update-setting}

Given a user's history $h$, a generative recommender predicts an item through
its semantic identifier, a length-$L$ path in
$\mathcal Z=\Sigma_1\times\cdots\times\Sigma_L$.
An injective assignment $c:\mathcal I\to\mathcal Z$ maps the current items
to the occupied catalog $A=c(\mathcal I)$. Beam search maintains $W$ live prefixes and
returns a ranked item list. We study when assigning identifiers to new items
can recover a target in this list while preserving existing item identifiers.

\paragraph{Assignment repair.}
Let $c_O$ be the retained assignment of old items $\mathcal I_{\mathrm{old}}$,
with occupied paths $O=c_O(\mathcal I_{\mathrm{old}})$.
A legal expansion preserves their identities and recommendation eligibility:
\begin{equation}
 O\subseteq A\subseteq\mathcal Z,\qquad c(i)=c_O(i)\quad(i\in\mathcal I_{\mathrm{old}}).
 \label{eq:admissible-catalog}
\end{equation}
The $m$ new items $\mathcal I_{\mathrm{new}}=\mathcal I\setminus\mathcal I_{\mathrm{old}}$
receive distinct paths in $A\setminus O$.
A repair is an admissible assignment under which a specified target appears
in Top-$K$. The theoretical question is whether such an assignment exists.
For recommendation, we use repairs from training contexts to propose a shared
map and evaluate it over held-out queries. The map used for a test query
does not depend on its target.

To characterize assignment repair, fix the generator, encoded history,
vocabulary, beam width and termination rule. Old-only histories retain the
same encoding as new assignments vary; mixed histories require re-encoding
when evaluating the shared map. At depth $d$, the effective decoder state
$S_{d-1}$ contains its live prefixes, accumulated scores, inputs and caches.
Before masking invalid catalog children, its full-vocabulary expansions are
\begin{equation}
 \mathcal E_d(S_{d-1})=\{(p,a):p\in S_{d-1}^{\rm live},a\in\Sigma_d\},\qquad
 q(p,a)=s(p)+r_\theta(h,S_{d-1};p,a).
 \label{eq:full-expansions}
\end{equation}
Legal children keep these scores. Deterministic state transitions and a fixed
tie order make comparisons across catalogs well defined. The assignment-family
results assume query-separable execution and verify relevant numerical gaps
(Appendix~\ref{app:native-closure}).

\paragraph{Scoring and candidate completion.}
After the shared map and adapted model are fixed, let
$\operatorname{enc}_c(h)$ denote the encoded history. A standard autoregressive
generator assigns complete-code score
$s_\theta(i\mid h)=\sum_{t=1}^{L}\log p_\theta(c_t(i)\mid\operatorname{enc}_c(h),c_{<t}(i))$.
We augment this score with an item-level predictor. The inference task is to
retrieve the Top-$K$ under the combined score. Collaborative evidence may
favor items outside the initial beam, so candidate completion evaluates
additional items while retaining the chosen map. Its inputs are the history,
current catalog and learned scores; the target is used only for evaluation.

We measure Recall and NDCG on the full test population. A target that has not
yet entered the current catalog contributes a miss. The candidate-completion
budget controls the number of additional complete-item likelihoods evaluated,
separately from the initial beam width and the cost of the item-level predictor.

\section{Assignment Limits and Constructive Guarantees}
\label{sec:repair-guarantees}

We first establish when identifier assignment cannot change retrieval, then
characterize when it can recover a target. Both results hold the generator
and encoded input fixed, isolating the effect of catalog occupancy.
We study Top-$K$ retrieval within a width-$W$ beam, with $K\le W$.

\subsection{Limits of assignment-only repair}

Suppose the unmasked Top-$W$ extensions at every layer of a reference run are
finite, live, old-supported prefixes with strict score gaps. Under the scoring
and termination premises of Section~\ref{sec:update-setting}, every admissible
catalog retains these extensions and their order. The effective state is
therefore preserved layer by layer, and all assignments return the same
ordered old-item list (Proposition~\ref{prop:assignment-closure}). A failed
check leaves repair unresolved.

This certificate bounds population metrics normalized to the interval $[0,1]$.
If $C$ is a
certified subset of an $N$-query population and query $u\in C$ has metric
upper bound $c_u$, then, over legal assignments $c$,
$\sup_c N^{-1}\sum_u m_u(c)\le(\sum_{u\in C}c_u+N-|C|)/N$.
Unresolved queries receive the optimistic metric value one. An observed
update exceeding this bound on the same population cannot be matched by any
legal assignment of the fixed parent model. The bound also applies to a
single shared assignment, even when different queries require incompatible
repairs (Appendix~\ref{sec:closure}).

\subsection{Exact feasibility of assignment repair}

Assume all admissible catalogs share a valid effective beam of $W$
old-supported prefixes at depth $L-2$.
Successful old-prefix checks establish this condition. At the next depth,
old-supported prefixes are mandatory, while optional prefixes can be enabled
by occupying a descendant. Enumerate the ordered Top-$W$ states $S$ induced
by the optional prefixes that can reach the mandatory cutoff.
For a state $S$, exclude descendants of unselected optional prefixes above
its cutoff, leaving maximal catalog $Z_S$. Each selected nonold prefix needs
one occupied descendant; these disjoint support groups form $\mathcal G_S$.
Write $U_S=Z_S\setminus O$, $C_S=|U_S|$ and $r_S=|\mathcal G_S|$.

\begin{theorem}[Catalog feasibility and constructive repair]
\label{thm:main-repair}
A catalog reaches $S$ exactly when
$O\subseteq A\subseteq Z_S$ and $A\cap G\ne\varnothing$ for every
$G\in\mathcal G_S$. For a new terminal child $y$ of $S$, let $R_S(y)$ be its full
one-based rank. Let $a_S(y)$ count its old outrankers and $g_S(y)$ count the
support groups whose every leaf precedes $y$. Let $\delta_S(y)$ be zero if $y$ belongs to a support group
and one otherwise. An $m$-item repair returning $y$ in Top-$K$ exists exactly when
\begin{equation}
 a_S(y)+g_S(y)<K,\qquad
 r_S+\delta_S(y)\le m\le C_S-\max\{0,R_S(y)-K\}.
 \label{eq:main-repair-interval}
\end{equation}
A satisfying case yields an explicit legal assignment.
\end{theorem}

The interval exposes two competing requirements. Enough new leaves must be
occupied to support the selected prefixes and the target. Enough
higher-ranked leaves must remain unoccupied to keep the target in Top-$K$.
Thus catalog capacity and ranking competition impose an upper count bound
alongside the lower bound from prefix support.

The characterization is constructive: select the target, support uncovered
groups with non-outranking leaves whenever possible, and fill remaining
capacity with non-outrankers before using higher-ranked leaves. Enumerating
the candidate states and target leaves decides repair feasibility within
the common-prefix setting. Section~\ref{sec:method-completion} uses this
feasible family as the constraint set for minimizing leaf replacement
while recovering a specified target in a chosen state.
Appendix~\ref{sec:closure} gives the proofs, critical-state search and the
worst-case lower bound for complete-state score-oracle access.

\section{Constructive Repair and Bound-Guided Candidate Completion}
\label{sec:method-completion}

BB connects catalog updating with retrieval for each query
(Figure~\ref{fig:method-overview}). We solve the repair constraints of Theorem~\ref{thm:main-repair}
with minimum leaf replacement, then evaluate shared-map proposals on re-encoded
histories and adapt the generator. With that map fixed, the combined score
guides candidate completion and stopping.
\begin{figure}[t]
\centering
\includegraphics[width=\textwidth]{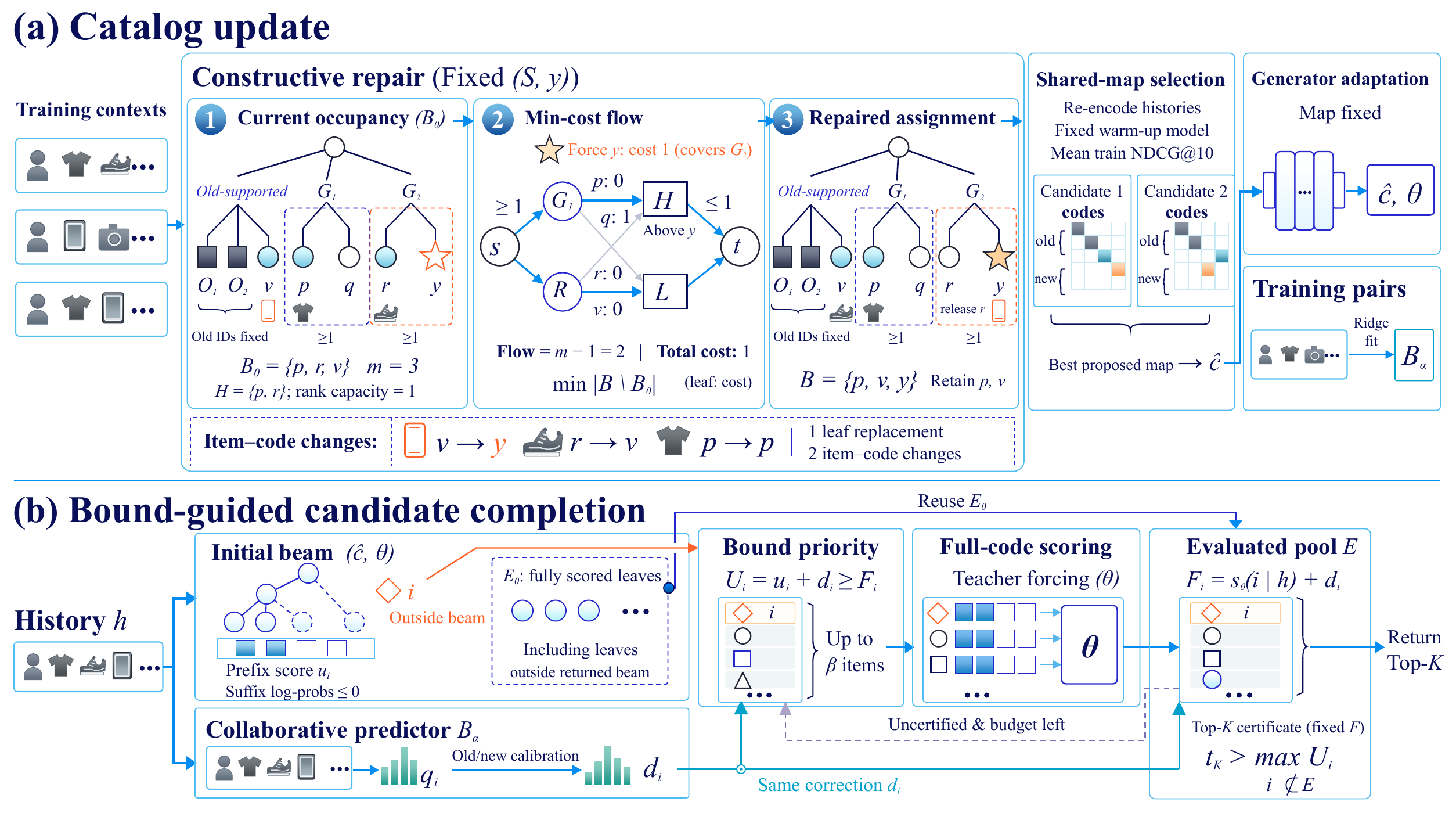}
\caption{Overview of BB. (a) Minimum-cost flow repairs fixed $(S,y)$ while
preserving old identifiers; trees omit unchanged branches and additional leaves.
Replacing $r$ with $y$ changes one occupied leaf and two item codes.
Re-encoded proposals are selected by training NDCG@10 under the fixed warm-up
model; the generator then adapts with the map fixed. (b) The correction $d_i$
enters both $U_i$ and $F_i$. Filled/outlined tokens denote scored prefixes/unevaluated
suffixes. Reuse includes all scored leaves $E_0$, even outside the returned
beam; the highlighted item initially lies outside $E_0$. Priority evaluates
only unscored items. With $|E|\ge K$, \eqref{eq:completion-stop} certifies
global Top-$K$ under fixed $F$; budget exhaustion can leave it uncertified.}

\label{fig:method-overview}
\end{figure}

\subsection{Constructive assignment and generator adaptation}

Repair jointly constrains prefix support, ranking competition and catalog
size. Among feasible occupancies, the objective preserves as much of the
current occupancy as possible. Let $B_0$ be the current
$m$-leaf new-item occupancy and fix a feasible state--target pair $(S,y)$.
Let $H_S(y)$ contain optional leaves ranked above $y$ and $a_S(y)$ count its
mandatory old outrankers. Within the compatible universe $U_S$ of
Theorem~\ref{thm:main-repair}, we solve
\begin{equation}
\begin{aligned}
 \min_{B\subseteq U_S}\quad &|B\setminus B_0|\\
 \mathrm{s.t.}\quad &|B|=m,\quad y\in B,\quad B\cap G\ne\varnothing\ (G\in\mathcal G_S),\\
 &|B\cap H_S(y)|\le K-1-a_S(y).
\end{aligned}
 \label{eq:method-minreplace}
\end{equation}
The universe excludes leaves that would displace the selected beam state.
The count constraint preserves catalog size, the support constraints activate
every required prefix, and the final constraint reserves a Top-$K$ position
for the target. The solution specifies occupied leaves; the next step assigns
item identities and evaluates the resulting shared map.

\paragraph{Reduction to an integral flow.}
Disjoint support groups let us aggregate leaves by support membership,
rank relative to $y$, and retention cost. The resulting minimum-cost-flow
network enforces all three constraints jointly and minimizes leaf replacement
for the specified state and target, avoiding enumeration of the $m$-leaf
subsets of $U_S$. Its integral solution and
reconstruction into a repair are detailed in Appendix~\ref{app:tuned-construction}.

\paragraph{From local solutions to a shared map.}
A local solution specifies leaf occupancy for one decoder state and target.
To obtain item mappings, we pin each proposal's target identity and retain
compatible original identities. We combine compatible proposals
into a finite set $\mathcal C_{\mathrm{prop}}$ of modified candidate maps $c$.
Each map re-encodes histories, so we assess its effect across training contexts
under the fixed warm-up generator. Let $r_c(h,i)$ be the target's one-based
rank in beam output after re-encoding history $h$ under $c$.
We evaluate candidates on a subset $\mathcal D_{\mathrm{sel}}$ of update-training
examples disjoint from the proposal-generation contexts. For a nonempty
candidate set, we select the shared assignment by mean NDCG@10:
\begin{equation}
 \widehat c\in\operatorname*{arg\,max}_{c\in\mathcal C_{\mathrm{prop}}}
 \frac{1}{|\mathcal D_{\mathrm{sel}}|}
 \sum_{(h,i)\in\mathcal D_{\mathrm{sel}}}
 \frac{\mathbf 1\{r_c(h,i)\le 10\}}{\log_2(1+r_c(h,i))}.
 \label{eq:shared-map-selection}
\end{equation}
Unreturned targets have rank $\infty$ and contribute zero. Ties use Recall@10,
fewer changed identities, then deterministic mapping order. If the candidate
set is empty, we retain the original map. Candidate budgets and composition
details are given in Appendix~\ref{app:tuned-construction}.

\paragraph{Learning under the constructed assignment.}
The selected map $\widehat c$ fixes history encodings and target identifiers
throughout generator adaptation. For update-training pairs
$(h,i)\in\mathcal D_{\mathrm{upd}}$, write $z=\widehat c(i)$ and minimize
\begin{equation}
 \mathcal L_{\mathrm{adapt}}(\theta;\widehat c)
 =-\mathbb E_{(h,i)\sim\mathcal D_{\mathrm{upd}}}
 \left[\frac{w(h,i)}{L}\sum_{\ell=1}^{L}
 \log p_\theta\!\left(z_\ell\mid
 \operatorname{enc}_{\widehat c}(h),z_{<\ell}\right)\right].
 \label{eq:generator-adaptation}
\end{equation}
Here $w(h,i)$ is the example weight, normalized to mean one. Adaptation starts
from the common warm-up model; Appendices~\ref{app:tuned-construction}
and~\ref{app:lcrec} specify the weights, trainable parameters and validation
checkpoint selection.

\needspace{7\baselineskip}
\subsection{Catalog-wide scoring for candidate completion}

With the map and adapted generator fixed, the candidate beam depends on the
query history. Collaborative evidence can favor items outside this
likelihood-based beam. A catalog-wide predictor provides item scores before
identifier evaluation, supplying the correction used for ranking and
candidate completion.

We fit a supervised ridge predictor $B_\alpha$
from normalized, recency-weighted history vectors to one-hot next items,
using only available training examples. With these inputs and targets
stacked as $X$ and $Y$, the fit minimizes
$\|Y-XB\|_F^2+\alpha\|B\|_F^2$ over $B$.
Linear item-to-item prediction is an established component of
recommendation~\citep{steck2019ease}; the closed-form solution and parameter
grids are given in Appendix~\ref{app:completion}. For history $h$, write
$q_i=\max\{(x_hB_\alpha)_i,10^{-8}\}$.

The amount of training evidence differs between old and newly admitted items.
We therefore calibrate the new-item contribution around the uniform reference
$\mu=-\log|\mathcal I|$. For an item with identifier $z=\widehat c(i)$,
the correction and final score are
\begin{align}
 d_i&=\lambda\log q_i+\mathbf 1\{i\in\mathcal I_{\mathrm{new}}\}
 \bigl[\lambda(\gamma-1)(\log q_i-\mu)+b\bigr],
 \label{eq:calibrated-correction}\\
 F(i\mid h)&=\sum_{\ell=1}^{L}\log p_\theta\!\left(z_\ell\mid
 \operatorname{enc}_{\widehat c}(h),z_{<\ell}\right)+d_i.
 \label{eq:completion-score}
\end{align}
The summed term is the adapted generator score $s_\theta(i\mid h)$.
Here $\lambda$ weights collaborative evidence, $\gamma$ rescales variation
among new-item corrections around the reference, and $b$ shifts new-item
scores relative to old items.
The choices $(\gamma,b)=(1,0)$ and $(0,0)$ recover the global correction
and the uniform new-item reference, respectively. Validation selects the
parameters for each dataset and retriever run. Because $d_i$ is known before
full identifier evaluation, it can also enter an upper bound on $F(i\mid h)$.
Candidate completion uses this bound to prioritize evaluations under the
same score used for final ranking.

\subsection{Bound-guided candidate completion}

During catalog-constrained beam search, retain each item's longest evaluated
prefix, with length $\ell_i$ and cumulative log probability $u_i$.
For $z=\widehat c(i)$, full-vocabulary log probabilities without length
normalization give
\begin{equation}
 \begin{aligned}
 U_i&:=u_i+d_i
 =\sum_{j=1}^{\ell_i}\log p_\theta\!\left(z_j\mid
 \operatorname{enc}_{\widehat c}(h),z_{<j}\right)+d_i,\\
 U_i-F(i\mid h)&=-\sum_{j=\ell_i+1}^{L}\log p_\theta\!\left(z_j\mid
 \operatorname{enc}_{\widehat c}(h),z_{<j}\right)\ge0.
 \end{aligned}
 \label{eq:item-upper}
\end{equation}
The unevaluated suffix contributes only nonpositive terms. Thus $U_i$
bounds the final score even after the path leaves the beam, allowing pruned
identifiers to remain eligible for completion. An empty prefix has
$\ell_i=0$ and $u_i=0$.

\paragraph{Budgeted candidate selection.}
The initial evaluated set $E_0$ contains every complete catalog leaf scored
in the terminal beam expansion, including leaves outside the returned beam.
At round $t$, let $n_t$ count additional full evaluations, with $n_0=0$.
For batch size $\beta$ and total budget $B_{\max}$, set
$b_t=\min\{\beta,B_{\max}-n_t,|\mathcal I\setminus E_t|\}$.
We select the batch with the largest total upper bound, evaluate its complete
identifier likelihoods by teacher forcing, and update
\begin{equation}
 \begin{aligned}
 Q_t&\in\operatorname*{arg\,max}_{\substack{Q\subseteq\mathcal I\setminus E_t\\|Q|=b_t}}
          \sum_{i\in Q}U_i,\\
 E_{t+1}&=E_t\cup Q_t,\qquad n_{t+1}=n_t+|Q_t|.
 \end{aligned}
 \label{eq:completion-update}
\end{equation}
Sorting the remaining bounds solves this selection problem, with item order
breaking ties.

\needspace{12\baselineskip}
\paragraph{Certified stopping.}
For an evaluated set $E$ containing at least $K$ items, let $t_K$ be
its $K$th largest score. Completion stops with a global Top-$K$ certificate when
\begin{equation}
 t_K>\max_{i\notin E}U_i.
 \label{eq:completion-stop}
\end{equation}
Every unevaluated item then scores below the $K$th evaluated item; the maximum
over an empty remaining set is $-\infty$. If the budget ends first, the
procedure returns the highest-scoring evaluated items without that certificate.
The rule applies optimal-stopping principles for monotone prefix
scores~\citep{huang2017finish,meister2020bestfirst} to the corrected catalog
score. Appendix~\ref{app:completion-proof} gives the certificate and the
nonincreasing certification-gap argument.
Algorithm~\ref{alg:completion} summarizes training and inference.

\begin{algorithm}[htbp]
\caption{Minimum-replacement repair and bound-guided completion}
\label{alg:completion}
\begin{algorithmic}[1]
\Require Training pairs, catalog, warm-up model; width $W$, cutoff $K$, batch $\beta$, budget $B_{\max}$
\State Solve~\eqref{eq:method-minreplace} for sampled state--target proposals by integral flow
\State Pin target identities and compose compatible proposals into shared maps
\State Select $\widehat c$ by~\eqref{eq:shared-map-selection}, retaining the original map if $\mathcal C_{\mathrm{prop}}=\varnothing$
\State Adapt the generator with $\widehat c$ fixed using~\eqref{eq:generator-adaptation}
\State Fit the next-item predictor; select scoring parameters on validation
\For{each recommendation history $h$}
 \State Run catalog beam search; retain prefix bounds $u_i$ and complete leaves $E_0$
 \State Compute $d_i$, $U_i$ for all items and $F_i$ for $i\in E_0$; set $n_0\gets0$, $t\gets0$
 \While{\eqref{eq:completion-stop} fails for $E_t$ and $n_t<B_{\max}$}
  \State Select $Q_t$ by~\eqref{eq:completion-update} and evaluate $F_i$ for $i\in Q_t$
  \State Update $E_{t+1},n_{t+1}$ by~\eqref{eq:completion-update}; set $t\gets t+1$
 \EndWhile
 \State Return the $K$ highest-scoring items in $E_t$ and whether~\eqref{eq:completion-stop} holds
\EndFor
\end{algorithmic}
\end{algorithm}

\section{Experiments}
\label{sec:evaluation}
\label{sec:results}

\subsection{Experimental setup}
\label{sec:setup}

\paragraph{Data and task.}
We use Amazon 2014 Beauty, Tools, and Toys and Games~\citep{he2016ups}.
Beauty and Toys expand their catalogs from the 60th to the 80th event-time
percentile; Tools uses DACT's public period-0.8 sequences.
All methods within a domain share the examples and available catalog.
Beauty/Tools/Toys provide 7,335/5,360/6,266 test queries.
Future-catalog targets remain misses;
\emph{Primary} denotes current-catalog new targets with entirely old-item
histories. Appendix~\ref{app:protocol} gives preprocessing and training coverage.

\paragraph{Comparisons.}
We adapt JTM matching~\citep{zhu2019joint} and DREAM voting~\citep{guan2026dream}
to the shared T5/candidate framework. DACT~\citep{feng2026dact},
Reformer~\citep{shi2025incremental} and five tokenizer/generative recommendation
model (GRM) update strategies share a tokenizer-compatible parent.
All methods are validation-tuned from original-paper or official settings;
BB selects adaptation and scoring settings per dataset and seed.
Appendix~\ref{app:protocol} specifies initialization and training budgets.
We also evaluate decoder-only LC-Rec~\citep{zheng2024lcrec,feng2026dact}
on the same datasets (Appendix~\ref{app:lcrec}).

\paragraph{Evaluation and selection.}
T5 uses four-token identifiers and beam width 40. Completion reuses fully
scored leaves, returns 20 items and evaluates at most 80 additional items
in batches of 20. Validation Recall and NDCG at 10 and 20 jointly guide
tuning, with one configuration per run for all four metrics.
The main comparison reports means and sample standard deviations over
seeds 17, 42 and 2027 for every method on all three datasets.
Matched ablations inherit each BB
run's settings. \emph{Without construction} uses the paired static map and
adapted checkpoint, distinct from the separately tuned \emph{Static assignment}
baseline. The linear predictor is shared across retriever seeds.
Table~\ref{tab:current-scoring} lists the nine main-result configurations.

% Keep compact experimental floats from stretching their inter-float gaps.
\raggedbottom
\subsection{Validation of constructive repair and guarantees}
\label{sec:count-validation}

Across 96 scorers and 221,184 exhaustively decoded catalogs, full construction
solves all 1,816 feasible count/cutoff conditions, versus 1,708 for
reference-state construction. High-score allocation solves all 882
prefix-local cases but only 733 of 934 complete-state cases.
Removing support or either count bound produces false feasibility decisions;
the full criterion makes none. Thus each constraint is necessary in this
finite design, and searching beyond the reference state recovers additional cases.

Direct Beauty decoding verifies 27 repairs, including 22 requiring new-prefix
support; simpler allocations match these counts. Population-bound experiments
establish two Beauty separations across three parents; all six Tools
comparisons remain unresolved. Appendix~\ref{app:assignment-diagnostics}
gives the full feasibility and population accounting.

\subsection{Main recommendation comparison}
\label{sec:shared-results}

With T5 and decoder-only LC-Rec, BB leads Recall@10 and NDCG@10 on all
three datasets (Table~\ref{tab:trained-comparison}; Appendix~\ref{app:lcrec}).
On T5, changing the tokenizer alone reduces accuracy; adapting the retriever
recovers much of it. Among standard strategies,
retraining is strongest on Beauty; retaining the tokenizer and fine-tuning
the retriever is strongest on Tools and Toys.

\begin{table}[t]
\centering\scriptsize
\setlength{\tabcolsep}{1.5pt}
\caption{T5 recommendation under catalog updates: Recall and NDCG at 10 (\%),
mean $\pm$ sample s.d. over seeds 17/42/2027 for all methods;
bold marks the highest mean. In tokenizer/GRM pairs,
FT denotes fine-tuning and RT retraining.}

\label{tab:trained-comparison}
\begin{tabular}{lrrrrrr}
\toprule
& \multicolumn{2}{c}{Beauty} & \multicolumn{2}{c}{Tools} & \multicolumn{2}{c}{Toys} \\
Method & R@10 & N@10 & R@10 & N@10 & R@10 & N@10 \\
\midrule
Frozen model & $1.636\!\pm\!0.098$ & $0.775\!\pm\!0.084$ & $2.077\!\pm\!0.075$ & $1.074\!\pm\!0.043$ & $1.059\!\pm\!0.066$ & $0.532\!\pm\!0.003$ \\
Static assignment & $2.122\!\pm\!0.127$ & $1.006\!\pm\!0.105$ & $3.464\!\pm\!0.127$ & $1.913\!\pm\!0.065$ & $1.612\!\pm\!0.188$ & $0.747\!\pm\!0.090$ \\
\midrule
Frozen/Frozen & $1.445\!\pm\!0.024$ & $0.734\!\pm\!0.038$ & $2.475\!\pm\!0.106$ & $1.390\!\pm\!0.056$ & $1.107\!\pm\!0.176$ & $0.546\!\pm\!0.068$ \\
FT/Frozen & $0.345\!\pm\!0.008$ & $0.156\!\pm\!0.004$ & $0.336\!\pm\!0.000$ & $0.166\!\pm\!0.021$ & $0.681\!\pm\!0.082$ & $0.373\!\pm\!0.035$ \\
Frozen/FT & $1.822\!\pm\!0.048$ & $0.914\!\pm\!0.021$ & $4.291\!\pm\!0.226$ & $2.325\!\pm\!0.128$ & $1.468\!\pm\!0.120$ & $0.743\!\pm\!0.059$ \\
FT/FT & $2.086\!\pm\!0.076$ & $0.998\!\pm\!0.067$ & $4.080\!\pm\!0.118$ & $2.163\!\pm\!0.085$ & $1.426\!\pm\!0.088$ & $0.704\!\pm\!0.052$ \\
FT/RT & $2.200\!\pm\!0.180$ & $1.058\!\pm\!0.062$ & $2.730\!\pm\!0.324$ & $1.467\!\pm\!0.127$ & $1.468\!\pm\!0.100$ & $0.714\!\pm\!0.064$ \\
\midrule
JTM matching (adapted) & $2.059\!\pm\!0.165$ & $0.984\!\pm\!0.116$ & $3.302\!\pm\!0.202$ & $1.799\!\pm\!0.190$ & $1.623\!\pm\!0.060$ & $0.735\!\pm\!0.024$ \\
DREAM voting (adapted) & $2.154\!\pm\!0.165$ & $1.022\!\pm\!0.122$ & $3.470\!\pm\!0.135$ & $1.915\!\pm\!0.067$ & $1.325\!\pm\!0.032$ & $0.614\!\pm\!0.012$ \\
DACT~\citeyearpar{feng2026dact} & $2.095\!\pm\!0.200$ & $1.063\!\pm\!0.068$ & $4.142\!\pm\!0.391$ & $2.260\!\pm\!0.214$ & $1.505\!\pm\!0.066$ & $0.762\!\pm\!0.047$ \\
Reformer~\citeyearpar{shi2025incremental} & $2.063\!\pm\!0.143$ & $1.015\!\pm\!0.072$ & $4.210\!\pm\!0.159$ & $2.436\!\pm\!0.140$ & $1.623\!\pm\!0.202$ & $0.809\!\pm\!0.109$ \\
\midrule
\textbf{BB (Ours)} & $\mathbf{3.217\!\pm\!0.072}$ & $\mathbf{1.535\!\pm\!0.047}$ & $\mathbf{4.956\!\pm\!0.138}$ & $\mathbf{2.806\!\pm\!0.053}$ & $\mathbf{2.154\!\pm\!0.064}$ & $\mathbf{1.027\!\pm\!0.063}$ \\
\bottomrule
\end{tabular}
\end{table}

\subsection{Component effects on ranking and certification}
\label{sec:ablation-results}

Table~\ref{tab:completion-all-at10} and Figure~\ref{fig:main-ablation-profiles}
report matched ranking, cohort and certification effects.

\begin{table}[t]
\centering\scriptsize
\setlength{\tabcolsep}{3pt}
\caption{Matched T5 component controls at cutoff 10: Recall and NDCG (\%),
mean $\pm$ sample s.d. over seeds 17/42/2027. Each removal inherits its
corresponding full-method settings; BB uses the runs in Table~\ref{tab:trained-comparison}.
Without completion reranks the returned beam; initial-pool reranking uses
all fully scored leaves $E_0$ with no additional evaluations.}
\label{tab:completion-all-at10}
\begin{tabular}{lrrrrrr}
\toprule
& \multicolumn{2}{c}{Beauty} & \multicolumn{2}{c}{Tools} & \multicolumn{2}{c}{Toys} \\
Variant & R@10 & N@10 & R@10 & N@10 & R@10 & N@10 \\
\midrule
Generator only & $2.336\!\pm\!0.202$ & $1.115\!\pm\!0.105$ & $3.408\!\pm\!0.151$ & $1.850\!\pm\!0.129$ & $1.649\!\pm\!0.103$ & $0.790\!\pm\!0.072$ \\
Without collaborative score & $2.336\!\pm\!0.202$ & $1.115\!\pm\!0.105$ & $3.414\!\pm\!0.153$ & $1.852\!\pm\!0.129$ & $1.649\!\pm\!0.103$ & $0.790\!\pm\!0.071$ \\
Without completion & $2.872\!\pm\!0.163$ & $1.407\!\pm\!0.076$ & $4.291\!\pm\!0.067$ & $2.558\!\pm\!0.016$ & $2.016\!\pm\!0.040$ & $0.975\!\pm\!0.049$ \\
Initial-pool reranking & $2.881\!\pm\!0.142$ & $1.416\!\pm\!0.070$ & $4.291\!\pm\!0.067$ & $2.558\!\pm\!0.016$ & $2.080\!\pm\!0.024$ & $1.003\!\pm\!0.032$ \\
Without construction & $3.231\!\pm\!0.189$ & $1.518\!\pm\!0.081$ & $4.963\!\pm\!0.131$ & $2.810\!\pm\!0.054$ & $2.096\!\pm\!0.109$ & $1.023\!\pm\!0.083$ \\
Item correction only & $2.449\!\pm\!0.042$ & $1.156\!\pm\!0.026$ & $3.197\!\pm\!0.298$ & $1.790\!\pm\!0.281$ & $1.282\!\pm\!0.009$ & $0.639\!\pm\!0.003$ \\
Collaborative priority & $3.222\!\pm\!0.075$ & $1.537\!\pm\!0.048$ & $4.820\!\pm\!0.163$ & $2.760\!\pm\!0.079$ & $2.197\!\pm\!0.079$ & $1.042\!\pm\!0.060$ \\
Without calibration & $2.568\!\pm\!0.028$ & $1.206\!\pm\!0.016$ & $4.434\!\pm\!0.152$ & $2.406\!\pm\!0.142$ & $2.032\!\pm\!0.024$ & $0.968\!\pm\!0.062$ \\
BB (Ours) & $3.217\!\pm\!0.072$ & $1.535\!\pm\!0.047$ & $4.956\!\pm\!0.138$ & $2.806\!\pm\!0.053$ & $2.154\!\pm\!0.064$ & $1.027\!\pm\!0.063$ \\
\bottomrule
\end{tabular}
\end{table}

\paragraph{Generative and collaborative evidence.}
Combining generative and collaborative evidence improves NDCG@10 over
either alone on all datasets. Calibration further improves these means by
adjusting collaborative evidence for old and new items
(Figure~\ref{fig:main-ablation-profiles}a).

\paragraph{Shared construction and adaptation.}
Relative to the paired without-construction branch, shared construction and
adaptation improve new-target NDCG@10 by 3.82/9.56\% on Beauty/Toys,
with larger Primary gains of 13.95/19.42\%.
They also raise numerical Top-20 certification by 3.54/9.87 percentage points
and reduce additional evaluations by 2.86/7.90\%
(Figure~\ref{fig:main-ablation-profiles}; Appendix~\ref{app:final-construction-effects}).
Old-target losses partly offset the ranking gains, leaving overall NDCG@10
improvements of 1.13/0.37\%; Beauty Recall@10 is slightly lower.

Tools's 133 new-target training contexts are far fewer than Beauty's 2,993
or Toys' 4,950, limiting shared-map training support. Two maps remain
unchanged; the third changes two identities, with near-zero matched
construction effects (Appendix~\ref{app:final-construction-effects}).

\subsection{Candidate access, certification and inference cost}
\label{sec:completion-results}

\paragraph{Separating candidate reuse and additional scoring.}
Reranking $E_0$ reuses fully scored leaves without additional evaluations,
improving mean NDCG@10 over returned-beam reranking on Beauty/Toys and
tying Tools. Full completion further improves Recall@10 and NDCG@10 in
all nine runs, recovering a net 3.36/6.65/0.74 Top-10 hits per 1,000 queries
(Table~\ref{tab:completion-all-at10}; Appendix~\ref{app:access-decomposition}).

\paragraph{Certification efficiency.}
At the same score, stopping test and 80-item maximum, bound priority reduces
Beauty's additional evaluations from 60.69 to 51.08 (15.83\%) at nearly
identical Recall@10 and NDCG@10. Numerical Top-20 certification reaches
77.64/4.45/61.53\% on Beauty/Tools/Toys, versus 45.40/0.71/31.95\%
under collaborative priority. Tools nearly exhausts both budgets.
Certification follows Appendix~\ref{app:completion-numerics}'s numerical
allowance and requires~\eqref{eq:completion-stop}; budget exhaustion is insufficient.

\begin{figure}[tbp]
\centering
\includegraphics[width=\textwidth]{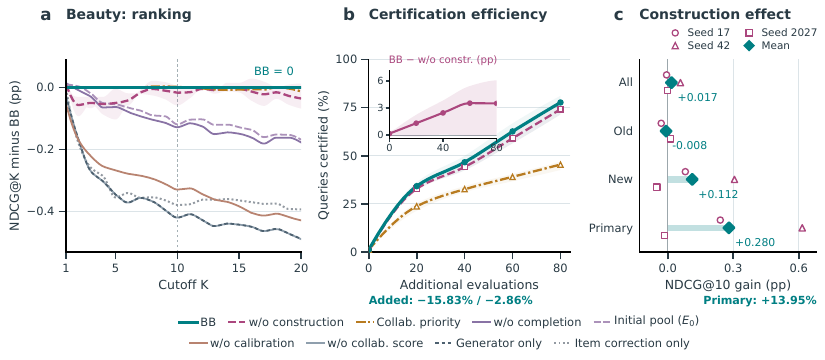}
\caption{Beauty T5 ablations, seeds 17/42/2027.
(a) Paired NDCG@$K$ differences, variant minus BB (pp).
(b) Certification efficiency: numerical Top-20 certification versus
additional evaluations; inset: paired BB-minus-without-construction gap (pp).
The two reductions compare BB with collaborative priority and without
construction, respectively.
(c) NDCG@10 differences, BB minus without construction; labelled open symbols
identify seeds and diamonds denote means. Primary denotes current-catalog new
targets with entirely old-item histories, a subset of New.
All/Old/New/Primary contain 7,335/5,334/1,486/734 queries per seed;
All includes future-target misses. Curves interpolate measured seed means,
with bands showing the range across seeds. The Primary annotation reports
the relative gain.}
\label{fig:main-ablation-profiles}
\end{figure}

\paragraph{Measured inference cost.}
GPU wall time covers beam search, item scoring and completion.
Batched means for bound/collaborative priority are 10.16/10.50, 10.64/10.63 and 10.96/11.98
ms/query on fixed Beauty, Tools and Toys subsets (Table~\ref{tab:completion-cost}).

\FloatBarrier

\section{Related Work}
\label{sec:related-work}

\paragraph{Identifier and model updates.}
JTM matches items to leaves \citep{zhu2019joint}; Reformer, DACT and
SID-Staleness update tokenization \citep{shi2025incremental,feng2026dact,baikalov2026staleness};
DREAM refines assignments and preserves multiple decoding paths \citep{guan2026dream}.
BB characterizes repair feasibility while preserving old identifiers, then
minimizes leaf replacement for a specified state and target.

\paragraph{Reachability and scoring.}
Prefix support, training rollouts and item resolution affect retrieval
\citep{peng2026cold,zhang2026hcgrec,ding2026sidscope}; beam-aware training
addresses pruning during learning \citep{zhuo2020beam,yang2026bear}.
BB connects assignment repair with collaborative scoring and access to
items outside the initial beam.

\paragraph{Candidate access and search.}
SpecGR proposes and verifies candidates \citep{ding2026specgr}.
Building on monotone search bounds \citep{huang2017finish,meister2020bestfirst},
BB uses a known item correction to bound the combined score, including
pruned paths, and certify catalog-wide Top-$K$ when the stopping condition holds.
Appendix~\ref{app:extended-related-work} develops these connections.

\section{Discussion and Conclusion}
\label{sec:discussion}
\label{sec:conclusion}

Shared construction and adaptation improve new-target ranking and certification
with fewer evaluations on Beauty and Toys. Ranking gains are largest for
entirely old-item histories; old-target losses partly offset them in the population mean.

Combined scoring and completion improve NDCG@10 across all three datasets
with both T5 and decoder-only LC-Rec, linking candidate access to the final
ranking objective across these backbone families. Extending repair
beyond late branching after a common effective prefix, optimizing shared
assignments across contexts and tightening completion bounds are natural
next steps.
\label{maintext:end}

% The preprint has no conference page-limit boundary; let statements follow
% naturally so the enlarged author block does not strand the conclusion.
\par\medskip
\flushbottom
\subsection*{AI use statement}
Generative AI tools were used solely to assist with language editing and
polishing during manuscript preparation. All other aspects of this work were
completed by the authors without generative AI assistance. The authors reviewed
and approved the final manuscript and take full responsibility for its content.

\subsection*{Reproducibility statement}
Experimental settings and evaluation protocols are provided in
Appendices~\ref{app:protocol} and~\ref{app:lcrec} for T5 and LC-Rec,
respectively. The selected T5 configurations are listed in
Table~\ref{tab:current-scoring}. These details are provided to facilitate
reproduction of the reported results.
The code is available at \url{https://anonymous.4open.science/r/bb_code-809D/}.

% References and appendices are also outside the main-text limit.
\bibliography{references}
\bibliographystyle{iclr2027_conference}

\appendix
% Supplement-only presentation settings; official ICLR dimensions and fonts remain unchanged.
% Long text cells wrap without forced word spacing. Numeric columns retain their source values.
\newcolumntype{L}[1]{>{\raggedright\arraybackslash}p{#1}}
\newcolumntype{Y}{>{\raggedright\arraybackslash}X}
\raggedbottom
% Permit mixed figure/table pages without flushing half-filled float pages.
\setcounter{topnumber}{4}
\setcounter{bottomnumber}{3}
\setcounter{totalnumber}{6}
\renewcommand{\bottomfraction}{0.85}
\renewcommand{\floatpagefraction}{0.8}
\renewcommand{\arraystretch}{1.15}
\setlength{\tabcolsep}{4pt}
\setlength{\belowcaptionskip}{5pt}
\setlength{\intextsep}{12pt plus 2pt minus 2pt}
\setlength{\textfloatsep}{14pt plus 2pt minus 2pt}
\setlength{\floatsep}{14pt plus 2pt minus 2pt}
% Float-only pages start at the top instead of stretching gaps above and between tables.
\makeatletter
\setlength{\@fptop}{0pt}
\setlength{\@fpsep}{14pt}
\setlength{\@fpbot}{0pt plus 1fil}
\makeatother

\section{Theoretical Results, Proofs and Counterexamples}
\label{sec:closure}
\label{app:theory}

We first certify when assignment leaves the output unchanged and convert this
certificate into a cohort bound. We then characterize exact repair after a
common effective prefix, using the catalogs that realize each candidate state.

\subsection{Certifying output invariance}

Write $z_{d,1},\ldots,z_{d,W+1}$ for the highest-scoring unmasked extensions,
with scores $q_{d,1}\ge\cdots\ge q_{d,W+1}$, and $O_d$ for the depth-$d$ prefixes
of $O$. Define
\begin{equation}
 z_{d,j}\in O_d\ (1\le j\le W),\qquad
 \gamma_d=\min_{1\le j\le W}(q_{d,j}-q_{d,j+1})>0.
 \label{eq:closure-condition}
\end{equation}
All selected hypotheses must be finite and live. Inactive or finished hypotheses
must not displace them before the common termination.
\begin{proposition}[All-assignment output invariance]
\label{prop:assignment-closure}
Under the stated state, scoring, and termination premises, if
\eqref{eq:closure-condition} holds at every layer of one admissible run, every
catalog in~\eqref{eq:admissible-catalog} returns the same ordered $W$ old paths
and the same item identities.
\end{proposition}
The detailed proof is given in Appendix~\ref{app:closure-proof}.

The pass/fail flag is itself allocation-independent: two reference runs share
all states through their first failed layer, including that layer's test.
The checker observes the decoder's unmasked scores and records one flag/margin per
layer; it requires no target labels, candidate pool, or extra model forward pass.
The result is a sufficient certificate: a failed check leaves repair unresolved.
Numerical ties or unverified execution premises return unknown
(Appendix~\ref{app:native-closure}).

\subsection{Bounding improvement over a query cohort}

For a fixed $N$-query cohort $\mathcal Q$, let $C$ be queries with certified
per-query metric upper bounds $m_u(c)\le c_u\in[0,1]$ over legal assignments $c$. Then
\begin{equation}
 \sup_c\frac1N\sum_{u\in\mathcal Q}m_u(c)
 \le \frac{\sum_{u\in C}c_u+N-|C|}{N}=:U_C.
 \label{eq:metric-bound}
\end{equation}
Ordered-output invariance permits $c_u=m_u^0$, its fixed metric value.
For a metric on new targets, any sound no-new Top-$K$ certificate permits
$c_u=0$, even when old outputs differ across states, giving $U_C=1-|C|/N$.
Old-target contributions use their own certified bounds. Unknown and mixed-input
queries contribute the optimistic value one. One shared allocation obeys
\begin{equation}
 \underbrace{\sup_c\frac1N\sum_u m_u(c)}_{\mathrm{OPT}_{\rm shared}}
 \le\underbrace{\frac1N\sum_u\sup_c m_u(c)}_{\mathrm{OPT}_{\rm separate}}
 \le U_C.
 \label{eq:bound-slack}
\end{equation}
An observed update exceeding $U_C$ establishes a separation from every legal
allocation of the parent. A value below $U_C$ leaves the comparison unresolved.
Equation~\eqref{eq:bound-slack} also locates a source of slack: different queries
can require incompatible assignments. Thus individual witnesses and cohort
exclusion answer distinct questions about the same allocation family.

\subsection{Exact catalogs and fixed-count reachability}
\label{sec:exact-repair}

The sufficient certificate requires old prefixes to occupy the whole beam.
We now allow the beam to branch in the final two layers and solve reachability
through the catalogs that induce each state.

\paragraph{Common-prefix setting.}
Assume a common, valid effective beam $b$ of $W$ old-supported prefixes at depth
$L-2$ for every containing catalog, independently of its new-item count.
Successful prefix checks in Section~\ref{sec:closure} establish this premise.
When the premise is known for a single $m$, the result applies at that count;
reuse across counts requires a common state over the count range. Let $K\le W$.
A fixed total order extends scores and is preserved when candidates are deleted.
Decoder evaluations verify strict relevant gaps and leave ambiguous ties unknown.

Let $E$ extend $b$ to depth $L-1$, and let $M=E\cap O_{L-1}$ be its mandatory
candidates. At least $W$ exist. If $\tau$ is their $W$th score, optional candidates
strictly below $\tau$ cannot survive. Enumerating subsets of
$Q=\{e\in E\setminus M:q(e)\ge\tau\}$ and merging identical ordered Top-$W$
indices therefore covers every possible penultimate beam $S$; put $q=|Q|$.

For one such $S$, let $F_S$ be the unselected optional prefixes preceding its last
member. Write $D(e)$ for the full leaf descendants of $e$ and set
\begin{equation}
 Z_S=\mathcal Z\setminus\bigcup_{e\in F_S}D(e),\quad
 U_S=Z_S\setminus O,\quad C_S=|U_S|.
 \label{eq:maximal-catalog}
\end{equation}
The selected nonold prefixes require support groups
$\mathcal G_S=\{D(s):s\in S\setminus O_{L-1}\}$, with $r_S=|\mathcal G_S|$.
Each group needs an occupied descendant to keep its prefix in the beam.
Together, forbidden subtrees and required support determine the complete
catalog family for $S$.

\begin{theorem}[Exact catalog family]
\label{thm:catalog-family}
A containing catalog reaches $S$ if and only if
\begin{equation}
 O\subseteq A\subseteq Z_S,\qquad A\cap G\ne\varnothing\quad(G\in\mathcal G_S).
 \label{eq:catalog-family}
\end{equation}
It does so with exactly $m$ new leaves if and only if $r_S\le m\le C_S$.
\end{theorem}
The detailed proof is given in Appendix~\ref{app:exact-repair}.

For final alphabet size $B_L$, $C_S=|\mathcal Z|-|O|-B_L|F_S|$ is symbolic.
It includes compatible leaves outside the current beam, which can supply padding.

\paragraph{Which admission counts permit a repair?}
After fixing $S$, a repair must both support the beam and keep enough
higher-ranked leaves out of the catalog. Score all terminal children of $S$.
For a new terminal leaf $y$, let $R_S(y)$ be its full
one-based rank. Let $a_S(y)$ count the old leaves preceding it and $g_S(y)$ count
the required groups whose \emph{every} leaf precedes it. Put $\delta_S(y)=0$ if $y$ belongs to a
required group, and $1$ otherwise.
\begin{theorem}[Exact fixed-count reachability]
\label{thm:count-repair}
There exists an admissible allocation of exactly $m$ new items reaching $S$ and
returning $y$ in Top-$K$ if and only if
\begin{equation}
 \boxed{a_S(y)+g_S(y)<K,\qquad
 r_S+\delta_S(y)\le m\le C_S-\max\{0,R_S(y)-K\}.}
 \label{eq:repair-interval}
\end{equation}
A satisfying case constructively yields a catalog and an assignment of a
specified new target to $y$.
\end{theorem}
The detailed proof is given in Appendix~\ref{app:exact-repair}.

Equation~\eqref{eq:repair-interval} separates three requirements. Old outrankers
and unavoidable support must leave room in Top-$K$. The lower endpoint supplies
enough new leaves to support the state and target. The upper endpoint reserves
enough unoccupied paths to remove optional outrankers. These requirements connect
catalog size directly to retrieval, beyond checking whether a target path exists.

\subsection{Constructing repairs and accounting for search}
\label{sec:repair-algorithm}

\paragraph{Exact decision and witnesses.}
Enumerate $S$, discard states with $m\notin[r_S,C_S]$, obtain their complete
terminal scores, and apply~\eqref{eq:repair-interval} to every new leaf.
Within the common-prefix setting, the absence of a feasible pair certifies
that no $m$-admission catalog returns a new item in Top-$K$. A feasible pair gives one query's
witness; different queries need not share that allocation. With unrestricted
admission count, universal no-new retrieval requires every state's full
Top-$K$ to contain only old leaves: if a
new leaf is in a state's full Top-$K$, $Z_S$ realizes that state and outcome.
Such a leaf has a witness using at most $r_S+1\le W+1$ new paths. Computation caps
and numerical ambiguity still mean unknown, not an exact negative.

\paragraph{Algorithm and cost boundary.}
One terminal sort plus cumulative old counts and required-group endpoints gives
all intervals in $O(WB_L\log(WB_L))$ time and $O(WB_L)$ storage after scoring.
There are at most $2^q$ distinct states, each with its own decoder continuation.
Algorithm~\ref{alg:exact-repair} summarizes the construction. Supporting the
selected state and filling with non-outrankers first attains the rank bound.
The cost depends on critical branches and terminal scoring, in addition to
writing the allocation. Appendix~\ref{app:exact-repair} gives the full procedure
and complexity accounting. Unverified execution premises or ambiguous numerical
ties return unknown before an affected decision is made.
\begin{theorem}[Complete-state oracle lower bound]
\label{thm:oracle-bound}
For every $q\ge1$, a positive normalized three-layer score family with
$W=q+1$, $K=1$, and $m=q+1$ has $2^q$ reachable penultimate states and $q$ critical
optional prefixes. Any deterministic exact algorithm deciding universal no-new
Top-1 through arbitrary complete-state terminal-score oracle access requires
$2^q$ state queries in the worst case.
\end{theorem}
A query returns the entire terminal-score vector for one ordered effective state.
The proof places the successful repair in one hidden state
(Appendix~\ref{app:exact-repair}). It identifies the worst-case information cost
of exact decisions when complete-state scores are accessed through this oracle.
Prefix-local scoring can supply additional structure beyond that interface.

\subsection{Proof Details and Diagnostic Counterexamples}
\label{app:closure-proof}

\subsubsection{Effective-state induction}

Proposition~\ref{prop:assignment-closure} concerns the state consumed by the next scoring step,
and the final ordered predictions. It does not require every temporary decoder variable to agree.
Let $S_d^A$ contain the ordered live hypotheses, their accumulated scores, model inputs/positions,
and the cache state after the layer-$d$ transition under catalog $A$. All random choices and
numerical execution conditions are fixed. Query separability means that changing another query's
values, without changing this query's numerical execution shape, does not change its scoring
function or transition. This is a stated numerical premise.

For an arbitrary $A$ satisfying~\eqref{eq:admissible-catalog}, assume
$S_{d-1}^A=S_{d-1}^{A_0}$ for the reference $A_0$. The unmasked score vector is identical.
Every reference top-$W$ expansion belongs to $O_d$, hence to the prefix closure of $A$.
The positive minimum adjacent gap fixes its rank among all codebook extensions.
An admissible mask can delete only lower-ranked competing extensions, so the same ordered
top $W$ is selected. Identical predecessor indices and tokens produce identical next inputs
and cache reorderings through the decoder transition. Thus $S_d^A=S_d^{A_0}$ for every layer
whose state is consumed by a subsequent score computation. At the common terminal layer,
the selected paths and scores agree; all paths lie in $O$ and resolve to the same old items.

Catalog independence of the \emph{certificate} follows from the same induction until the first
failed condition. Before that layer the states coincide, so the failure agrees too. No claim is
made about later flags on an already failed trace. An old-only catalog can therefore serve as
the reference if it meets the decoder's validity assumptions. Same-length new aliases do not
help: any such catalog is still included in~\eqref{eq:admissible-catalog}.

For~\eqref{eq:metric-bound}, fix any legal assignment $c$ and split the cohort
sum into $C$ and $\mathcal Q\setminus C$. By assumption, the first sum is at
most $\sum_{u\in C}c_u$; each remaining contribution is at most one. Dividing
by $N$ and taking the supremum over $c$ proves the bound. Under ordered-output
invariance, the certified contributions equal their fixed reference values.
The supremum need not attain the bound because all queries share one global assignment.
For a new-target metric the certified contribution is zero. For an old-target metric it is
its fixed reference contribution, not necessarily zero.

\subsubsection{Final-only and static-support checks are insufficient}

Consider two constructed scorers, a two-level code space with four symbols per level,
$W=2$, and old paths $O=\{a0,a1,c1,c2,c3\}$. In the reference allocation the
single new item occupies $a2$. Both its atomic symbols and its first prefix occur in old
paths. The two scorers share all second-level probabilities:
\begin{center}
\begin{tabular}{lrrrr}
\toprule
Prefix & 0 & 1 & 2 & 3\\
\midrule
$a$ & .60 & .38 & .01 & .01\\
$b$ & .26 & .25 & .25 & .24\\
$c$ & .97 & .01 & .01 & .01\\
$d$ & .25 & .25 & .25 & .25\\
\bottomrule
\end{tabular}
\end{center}
Their root probabilities over $a,b,c,d$ are $(.75,.04,.20,.01)$ and
$(.49,.30,.20,.01)$, respectively. All values are evaluated as exact rational numbers.

For the first scorer, full-codebook search keeps $a,c$ and then selects old paths
$a0=.45$ and $a1=.285$, ahead of $c0=.194$. Every layer satisfies the strict old-prefix
condition. Enumerating all 11 unoccupied paths as the new item's allocation gives the same
old prediction list and no new hit.

For the second scorer, full-codebook search keeps $a,b$, then also returns old paths
$a0=.294$ and $a1=.1862$. The published catalog permits only root prefixes $a,c$;
new $a2=.0049$ still loses. However, allocating the new item to $c0$ instead gives
$c0=.194>a1=.1862$, so it enters at rank two. All old paths and identities stay fixed.
Full-codebook search is not a monotone envelope of returned paths in its subcatalogs:
branch $b$ removes $c$ early but produces no competitive final leaf.
The two scorers therefore agree on static support and both simple final-only observations,
yet differ in whether allocation repair is possible. The complete example uses 26 exact beam
calls: two scorers, each with the current catalog, full codebook, and all 11 legal allocations.
The construction isolates the effect of intermediate beam competition.

\subsubsection{Why masking and old eligibility matter}

Local renormalization changes the score function with the catalog. Let $W=L=2$,
$O=\{00,01,10\}$, and use the following probabilities over four symbols:
\[
p_{\rm root}=(.59,.40,.005,.005),\qquad
p_0=(.65,.34,.005,.005),\qquad p_1=(.01,.45,.27,.27).
\]
The unrenormalized full-codebook top two are old paths $00=.3835$ and $01=.2006$.
Renormalizing over valid children and then adding new path 11 instead yields
$P(11)=100/253\approx.395257$, ahead of $P(00)=3835/9801\approx.391287$.
The original mask-only condition cannot be transferred to that decoder.

Old lookup validity is also insufficient. With width one and fixed scores .6 for an old path
and .4 for a new path, hiding the old candidate returns the new item even if a historical
resolver still maps the old path correctly. Our $O\subseteq A$ requirement includes eligibility.
Reserved tombstones, versioned identifiers, vocabulary expansion, and different input encodings
must therefore be analyzed under their actual interfaces.

\subsection{Exact-repair proofs and decision procedure}
\label{app:exact-repair}

\subsubsection{Catalog family and fixed-count endpoints}

The common-prefix premise covers every containing catalog
$O\subseteq A\subseteq\mathcal Z$, irrespective of $|A\setminus O|$.
Consequently it also covers the maximal catalog $Z_S$ and all padding choices.
If that premise is established only for a fixed count, the theorem is pointwise
at that count; terminal scores cannot automatically be reused at other counts.
Effective next state is determined by the common prior state and ordered selected
candidate indices, including the corresponding decoder cache update.

All mandatory prefixes are enabled. For an obtainable ordered beam $S$, no
unselected mandatory candidate can precede its final member. Every selected
optional prefix must have an occupied descendant, and every unselected optional
prefix above the cutoff must have none. Conversely those conditions leave all
selected prefixes eligible and all other eligible candidates below the cutoff.
This proves the necessity and sufficiency of~\eqref{eq:catalog-family}. The sets
$D(s)$ for selected optional prefixes are disjoint, nonempty subsets of $U_S$;
selecting one leaf per group realizes $r_S$ admissions. Any remaining member of
$U_S$ can be added without changing the beam, realizing every count through $C_S$.
Forbidden subtrees are also disjoint and contain no old leaf, giving the symbolic
capacity $C_S=|\mathcal Z|-|O|-B_L|F_S|$.

For a specified new leaf $y$, partition $U_S\setminus\{y\}$ into the $h$ optional
terminal outrankers $H$ and the $t=C_S-1-h$ remaining compatible leaves $T$.
The latter include all compatible leaves outside the terminal expansion.
A required group wholly in $H$ forces one selected outranker. Every other group
contains $y$ or a member of $T$. The disjointness of groups makes the minimum
support count $r_S+\delta_S(y)$ exact. Any $m$-element new set containing $y$
requires at least $m-1-t$ leaves from $H$, and at least $g_S(y)$ for support.
Thus its best possible terminal rank is
\[
 1+a_S(y)+\max\{g_S(y),m-1-t\}.
\]
Choose $y$, one $T$ member of each unsupported group whenever possible, and one
$H$ member otherwise. Fill unused $T$ capacity before selecting further $H$
members. This attains the lower bound for every
$r_S+\delta_S(y)\le m\le C_S$. Using $R_S(y)-1=a_S(y)+h$ gives the upper endpoint
$C_S-\max\{0,R_S(y)-K\}$ and the separate obstruction $a_S(y)+g_S(y)<K$.
Equality at either endpoint is permitted. Too few admissions cannot support
both the state and target; near a full catalog, too many can force outrankers
back in. The interval describes the reachability of the specified leaf as $m$ varies.

If the unrestricted full terminal Top-$K$ contains a new leaf, its full rank is
at most $K$ and $Z_S$ realizes that outcome. Supporting each optional selected
prefix and $y$ needs at most $r_S+1\le W+1$ new leaves; deleting the other optional
leaves cannot worsen $y$ while those supports and forbidden-prefix constraints
preserve $S$. Conversely, $K$ mandatory old leaves preceding every new leaf
remain present in every compatible catalog. This gives exact unrestricted no-new certification over the complete
reachable state set. The branch experiment in
Appendix~\ref{app:assignment-diagnostics} evaluates the sufficient no-new
predicate used for its coverage comparisons.

\subsubsection{Constructive decision procedure and complexity}

\begin{algorithm}[t]
\caption{Fixed-count repair after a certified common prefix}
\label{alg:exact-repair}
\label{alg:repair-overview}
\begin{algorithmic}[1]
\Require Fixed input/model/decoder, mandatory old paths $O$, admission count $m$, $K\le W$
\If{common depth-$(L-2)$ effective state is not established}
  \State \Return unknown
\EndIf
\State Obtain its global expansion scores; form mandatory $M$ and critical optional $Q$
\State Enumerate subsets of $Q$ with $M$ retained; merge identical ordered Top-$W$ states
\For{each state $S$}
  \State Form forbidden prefixes $F_S$, support groups $\mathcal G_S$, and capacity $C_S$
  \If{$r_S\le m\le C_S$}
    \State Obtain terminal scores at the full effective state $S$
    \State Sort once; compute ranks, old-before counts, and required-group endpoints
    \For{each new terminal leaf $y$ satisfying \eqref{eq:repair-interval}}
      \State Select $y$ and support every uncovered group using a non-outranker when possible
      \State Pad to $m$ new leaves with compatible non-outrankers before outrankers
      \State Assign the designated target to $y$ and biject the other new items to remaining leaves
      \State \Return reachable and the complete allocation witness
    \EndFor
  \EndIf
\EndFor
\State \Return unreachable
\end{algorithmic}
\end{algorithm}
Every unverified numerical tie, execution premise, or exhausted computation cap
returns unknown; Algorithm~\ref{alg:exact-repair}'s final negative requires a
complete valid enumeration. A fixed total secondary order handles abstract ties.
Applications to trained models retain the tie order of the implemented decoder.

A group's contribution to $g_S(y)$ starts just after its worst-ranked leaf.
After one sort, cumulative old flags and these group endpoints give all count
intervals in linear additional time. If $n\le2^q$ distinct states survive the
count filter, a straightforward implementation reads $WB_{L-1}$ expansion
scores, spends $O(2^q(|M|+q)\log(|M|+q))$ time on subset planning, requests $n$
terminal-state evaluations, and spends $O(nWB_L\log(WB_L))$ on terminal ranking.
Intervals do not require rescoring for each $m$ within the common-state range.
Explicitly writing a catalog of $m$ new identifiers still costs at least $m$
outputs; symbolic capacity does not remove that cost.

\subsubsection{Proof of the state-oracle lower bound}

For each $q\ge1$, use three layers with root alphabet size $2q+1$ and binary
second/final alphabets. Every root has one old leaf at suffix $(0,0)$.
Set $W=q+1$. The anchor root has probability $3/4$. Distribute total mass
$249/1000$ among $q$ lower roots in proportion to
$1+i/(100q)$, $i=1,\ldots,q$. Each of $q$ off-beam roots has mass $1/(1000q)$.
All roots are old-supported, so the ordered beam after layer one is common:
the anchor and all lower roots. The padding roots are always excluded.

Split the anchor into old/optional probabilities
$1-1/(10000q)$ and $1/(10000q)$; split each lower root into $2/5$ old and $3/5$
optional. Lower root masses vary by less than a factor $1.01$, so all lower
optional children precede all lower old children. The anchor old child dominates
them, and the anchor optional child is below the lowest mandatory score. Exactly
$q$ optional children are critical. Every subset survives entirely, with mandatory
lower old children filling unused positions; all $2^q$ states are distinct.

The anchor's old terminal leaf wins under the good oracle, with old/new terminal
probabilities $.9/.1$. Other prefixes can use $.6/.4$; even their entire accumulated
mass is smaller than the winning anchor leaf. For any designated state $S^*$,
change only the anchor terminal distribution there to $.01/.99$. Its new leaf
$y$ then wins. Every probability remains positive and normalized, and all other
state answers agree with the good oracle. Arbitrary dependence on the complete
ordered state permits this family.

Each off-beam root has three optional leaves, providing $3q$ padding leaves that
cannot alter the first beam. A fixed-count witness for any state includes $y$,
one leaf for each selected optional prefix, and enough such padding to reach
$m=q+1$. It introduces no forbidden prefix and reaches the designated state.
Thus every hidden bad state has a witness at the same admission count.

One oracle query supplies the complete terminal vector for one specified ordered
state. On a run receiving only good answers, any unqueried state might be $S^*$.
The good oracle requires a universal-no-new answer, whereas the alternative
requires the opposite answer. Their observed responses are identical, so a
deterministic exact algorithm cannot stop before all $2^q$ state queries.
This lower bound concerns complete-state terminal-score oracle access.
Network weights, symbolic or prefix-local score structure, and richer batched
oracles provide information beyond this model.

\section{Construction, Adaptation and Completion Details}
\label{app:implementation}
\label{app:tuned-construction}
This appendix develops the steps from assignment repair to candidate completion.
A worked example first shows how changing catalog occupancy and evaluating
additional candidates each affect retrieval. We then specify shared-map
construction, generator adaptation and completion, with the final T5
configurations listed in Table~\ref{tab:current-scoring}.

\subsection{A worked example of repair and completion}
\label{app:worked-example}

Figure~\ref{fig:repair-completion-example} follows an old item $a$ and a new
item $x$ under a fixed generator and old-item history, with two-token codes
and $W=K=1$. Their initial codes are P0 and Q0. Empty leaves are masked
without renormalizing token probabilities. The beam selects prefix P and
returns only $a$, giving the initial fully scored catalog set $E_0=\{a\}$.

\begin{figure}[!htbp]
\centering
\includegraphics[width=\linewidth]{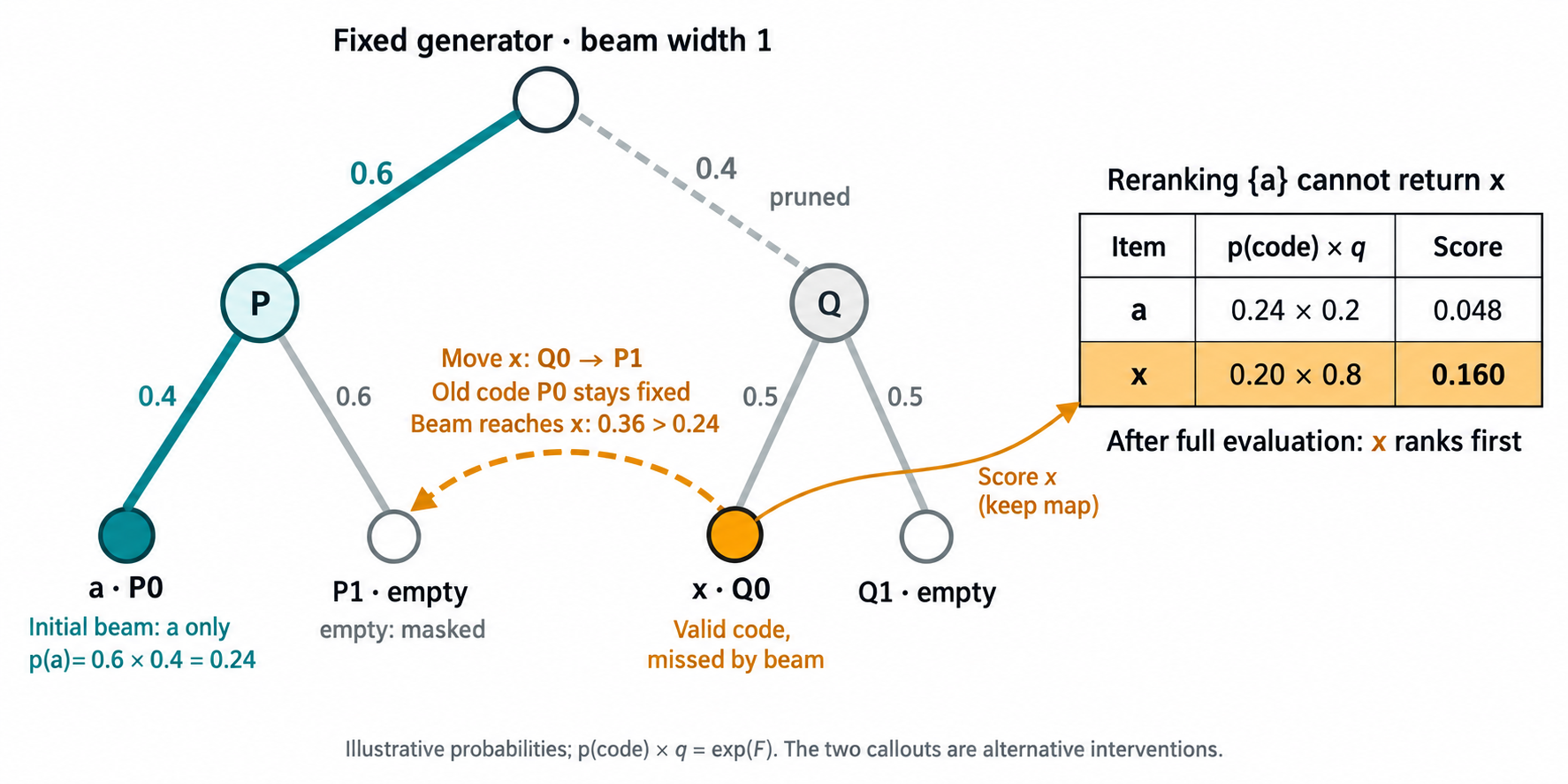}
\caption{Repair and completion in a two-item catalog with illustrative
probabilities. The teal path is the initial beam; hollow leaves are
unoccupied. The dashed arrow changes the assignment, whereas the solid
arrow keeps it fixed and adds a complete likelihood evaluation. These
interventions are applied separately to the initial miss. Suffix
probabilities below the pruned Q branch and the completed score of $x$
are shown retrospectively.}
\label{fig:repair-completion-example}
\end{figure}

The dashed arrow moves $x$ from Q0 to P1 while retaining the old code P0.
The beam then returns $x$, whose code probability $0.36$ exceeds $a$'s
$0.24$. The previously occupied Q0 becomes empty, completing a repair
with one occupied-leaf replacement.

The solid arrow retains the original map. With known collaborative item
scores $q_a=0.2$ and $q_x=0.8$, choose $(\lambda,\gamma,b)=(1,1,0)$ in
\eqref{eq:calibrated-correction}. The figure displays
$\exp(F_i)=p_\theta(c(i)\mid h)q_i$, preserving the ranking of
\eqref{eq:completion-score}. Reranking $E_0$ still cannot return $x$.
Evaluating Q0 expands the scored set to $\{a,x\}$; $x$ then ranks first
with $0.160>0.048$, while both identifiers remain fixed.

In the full method, construction and adaptation determine the shared map
and generator on which candidate completion operates. The following
subsections detail these stages in that order.
\FloatBarrier

\needspace{15\baselineskip}
\subsection{From a repair witness to a shared assignment}

The example repairs one query by replacing a single occupied leaf. General
repairs must also satisfy prefix-support and rank constraints while preserving
catalog size. For a fixed state and target leaf $y$, let $U$ be the admissible
new leaves, $B_0$ the
current $m$-leaf occupancy, and $G_j$ the disjoint support groups.
Let $H$ contain optional leaves ranked above $y$, and let $a$ count mandatory
old outrankers. The exact-state and rank constraints give the optimization
\[
 \min_{B\subseteq U}|B\setminus B_0|
 \quad\text{s.t.}\quad |B|=m,\quad y\in B,\quad
 B\cap G_j\ne\varnothing\ (\forall j),\quad
 |B\cap H|\le K-1-a.
\]
In the experiments, $K=10$ matches the primary recommendation cutoff.
Removing the forced target reduces demand to $m-1$ and discharges its
support group. The remaining leaves are partitioned by support group,
membership in $H$, and membership in $B_0$. Each support group and the
ungrouped remainder has at most four such buckets. A bucket has its integer leaf
capacity and unit cost zero for retained leaves or one for replacements.

This formulation is an integral minimum-cost-flow problem. Source arcs
supply group nodes, with lower bound one for each uncovered support group.
Bucket arcs connect group nodes to high- or low-ranked nodes; the high node
has sink capacity $K-1-a$. Total flow is $m-1$. Integer capacities and
network-flow integrality yield an optimal integral occupancy
\citep[Section~2.3]{bertsekas1991linear}. The forced
target's replacement cost is constant and is added to the objective.
This minimizes leaf replacement for the specified state and target;
identity edits and the search over states are separate decisions.
\paragraph{Equivalence and optimality.}
After removing $y$, let $V_j=G_j\setminus\{y\}$ for each support group,
and let $V_0$ contain the remaining ungrouped leaves of $U\setminus\{y\}$.
These sets partition the residual universe. The source-to-group arc has
upper capacity $|V_j|$ and lower capacity one exactly when $j\ne0$ and
$y\notin G_j$; all other lower capacities are zero. Split $V_j$ into
buckets according to membership in $H$ and $B_0$. A bucket arc has capacity
equal to its size and cost zero or one according to retention, and ends
at the corresponding high- or low-ranked node. Their sink capacities are
$K-1-a$ and $m-1$, respectively; the required total flow is $m-1$.

Every feasible occupancy sends one unit along the bucket route of each
leaf other than $y$. Its support constraints meet the lower capacities,
and its outranker count meets the high-ranked capacity. Conversely,
an integral feasible flow selects its prescribed number of distinct leaves
from each bucket. Bucket capacities and disjointness make this possible;
adding $y$ then satisfies every occupancy constraint. Flow cost plus
$\mathbf 1\{y\notin B_0\}$ equals $|B\setminus B_0|$ in both directions.
Integral capacities, lower bounds and total demand admit an integral
minimum-cost flow whenever feasible, proving equality of the two optima.
Negative high-ranked capacity or an unsupported empty group is infeasible.
The network has at most four bucket arcs per group, so its size is
$O(r_S+1)$ once bucket capacities are available; finding those capacities
and writing the selected leaves remain separate costs.

We pin the target identity, retain compatible original identities, and
assign displaced identities deterministically. Complete mappings are
checked after history re-encoding.

Candidate generation samples up to 256 training contexts and retains at
most 12 proposals. It considers up to eight feasible states per context,
with at most eight critical branches, 256 enumerated states, and four target
leaves per state. Every proposal preserves old identities and changes at
most 16 identities. A disjoint subset of up to 1,024 update-training examples
ranks proposals by re-encoded NDCG@10, Recall@10, and fewer identity changes.
It balances old and new targets where available and draws from the
training examples used in the common warm-up.
Starting from the original map, greedy composition accepts at most four
proposals. Each proposal must change identities disjoint from earlier
accepted proposals, preserve injectivity, and keep the total within 16
identity changes. At each step, we re-encode and evaluate every compatible
merge on the same training subset. We accept the largest strict improvement
in the ordered criteria above, retaining proposal order on ties, and stop
when none improves the current map. Each accepted intermediate map joins
the candidate pool.
One nonoriginal mapping is selected from single proposals and these
compositions; the lexicographically largest SHA-256 digest of the serialized mapping
resolves remaining ties.
When no feasible proposal exists, the mapping remains static.

\subsection{Training and checkpoint selection}

For the selected assignment, both history items and target items are encoded
with the same fixed map. Equation~\eqref{eq:generator-adaptation} averages
cross-entropy over the four identifier tokens using example weights $w(h,i)$
normalized to mean one over the training set. Uniform weighting sets $w(h,i)=1$.
The main T5 results use cross-day weighting for Beauty seed 2027 and
Toys seed 42, and uniform weighting for the other runs. LC-Rec uses uniform weights.
T5 updates all generator parameters. The LC-Rec adaptation uses the same
identifier-token objective with the trainable modules described in
Appendix~\ref{app:lcrec}.

Within each dataset and seed, the constructed and static branches start
from the same five-epoch warm-up. We treat the adaptation setting as a
categorical hyperparameter with three candidate values.
\emph{Constant} uses a fixed learning rate of $3\times10^{-4}$ throughout adaptation.
\emph{Cosine+EMA} uses cosine learning-rate decay from $10^{-4}$ to $10^{-5}$
and an exponential moving average (EMA) of decay 0.99.
\emph{Weighted EMA} adds cross-day training weight two and same-day weight
one half to the cosine+EMA setting.
For this weighted setting, $w(h,i)=\widetilde w(h,i)/\overline w$, where
$\widetilde w$ is the stated raw weight and $\overline w$ is its mean over
the whole training set.
EMA is initialized from the warm-up weights, updated after every optimizer
step, and used for validation and testing. AdamW, batch size 256, weight
decay 0.01, and gradient clipping at one remain fixed. Each run initializes
one optimizer and retains its state throughout adaptation.

Adaptation lasts at least 20 and at most 60 epochs after the common warm-up.
Checkpoint selection uses the 1,040 Beauty and 774 Toys positive-day-gap
validation examples; Tools uses all 3,007 validation examples because timestamps
are unavailable. We evaluate at the warm-up checkpoint and every five adaptation epochs.
Four evaluations without a strict NDCG@10 improvement on this selection subset
trigger early stopping after the minimum duration. Checkpoint ordering uses
NDCG@10 and Recall@10 on that subset, then full-validation NDCG@10 and earlier epoch.

The main comparison selects an adaptation setting for each dataset and seed
using validation results. Table~\ref{tab:current-scoring} gives the selected
setting and checkpoint epochs for both paired branches, counting from the
start of their common warm-up. Epoch 5 therefore denotes the trained warm-up
checkpoint before any further adaptation. The without-construction branch uses
the same adaptation setting and randomization with the original mapping.
JTM and DREAM retain their 20-epoch update schedule; DACT follows
Appendix~\ref{app:dact-comparison}.

\subsection{Training data and parameter selection}
\label{app:completion}

For a history of length $T$, its item at position $t$ receives weight
$\exp[-\rho(T-t)]$; repeated occurrences are summed and the vector is
normalized to unit sum. Training histories form $X$ and their one-hot next
items form $Y$. The predictor minimizes a supervised ridge objective:
\begin{equation}
 B_\alpha=\arg\min_B\|Y-XB\|_F^2+\alpha\|B\|_F^2
       =(X^\top X+\alpha I)^{-1}X^\top Y.
 \label{eq:transition-ridge}
\end{equation}
Its next-item target differs from reconstruction of the history vector.
The linear predictor fits available initial and update training examples,
deduplicated by user, history and target. Beauty, Tools and Toys contain
55,909, 49,661 and 46,964 distinct examples. Histories are normalized recency
weighted item counts; the target is a one-hot next-item vector. Validation
and test examples are excluded from this supervised fit. Repeated items
remain eligible under the shared recommendation protocol.

We evaluate six predictor configurations per domain: the Cartesian product
of decay $\rho\in\{0,0.2\}$ and ridge $\alpha\in\{1,10,100\}$.
Validation NDCG@10 selects a configuration, followed by Recall@10 and fixed
grid order to resolve ties. The selected ridge is 100 in all domains;
decay is 0.2 for Beauty and Tools and zero for Toys. The normal equations
are solved by double-precision Cholesky factorization; sampled-column relative residuals
are below $10^{-8}$. Selected coefficients are stored in single precision.

The collaborative score is clipped at $10^{-8}$ before taking its logarithm.
Scoring parameters $(\lambda,\gamma,b)$ are selected using the validation
procedure in Appendix~\ref{app:selection-measurements}.
Table~\ref{tab:current-scoring} reports the configuration used for each
test run and its matched ablations. Figure~\ref{fig:calibration-test-tradeoffs}
compares calibrated and global corrections under these configurations.

\subsection{Inference and numerical verification}
\label{app:completion-numerics}

The decoder computes full-vocabulary log probabilities before masking
invalid catalog paths. Prefix accumulation uses their sum, without beam
length normalization. The longest scored prefix of each catalog item
provides its bound. All catalog leaves already scored in the terminal beam
expansion enter the initial evaluated set, including leaves outside the
final width-40 output. We record this set size separately from additional
teacher-forced evaluations.

Additional candidates are scored in groups of 20, with an 80-item maximum
per query. The final output contains 20 items. For the numerical stopping
check, the twentieth known score minus $10^{-3}$ must exceed every remaining
bound plus $10^{-3}$. The mathematical guarantee in
Proposition~\ref{prop:completion} concerns exact scores; this allowance and
numerical checks do not constitute a formal floating-point error bound.
A run reaching its item budget can still produce a useful list without
satisfying the stopping check.

We verify prefix bounds and complete Top-20 rankings by exhaustive scoring
of 64-leaf catalogs on three training histories, at weights 0, 0.5 and 2
with both candidate priorities. Complete scores are also checked against
independent teacher forcing. The initial beam
predictions match the unchanged generator outputs. Calibration is additionally
checked in 144 exhaustive small-catalog comparisons covering the grid and
both priorities. All test metrics are independently recomputed and candidate
budgets verified across domains and seeds. Fresh evaluation on the timing
subsets reproduces the saved rankings and added-item counts.

\subsection{Completion guarantee}
\label{app:completion-proof}
\begin{proposition}[Completion certificate]
\label{prop:completion}
For fixed history, model, map and collaborative scores, let $E$ be the set
of items with complete likelihoods, and let $U_i$ be defined by
\eqref{eq:item-upper}. If the $K$th largest score in $E$ satisfies
\eqref{eq:completion-stop}, the returned list is the global Top-$K$ under
$F(\cdot\mid h)$.
\end{proposition}
\begin{proof}
Every suffix contributes a sum of nonpositive log probabilities, so each
unevaluated item's complete score is at most $U_i$. Under the stopping
condition, all such items score below the $K$th evaluated item. Hence none
can enter the returned Top-$K$ set.
\end{proof}

\paragraph{Progress toward certification.}
With at least $K$ evaluated items, define the certification gap
\[
 \Delta(E)=\max_{i\notin E}U_i-t_K(E),
 \qquad \max\varnothing=-\infty.
\]
For fixed prefix bounds and item corrections, the update in
\eqref{eq:completion-update} enlarges $E$, so its $K$th largest score cannot
decrease and the largest remaining bound cannot increase. Hence
$\Delta(E_{t+1})\le\Delta(E_t)$; a negative gap gives the stopping
certificate. This monotonicity does not guarantee certification within
the item budget or a gain in recommendation quality.

\subsection{Ablation definitions and computational accounting}

The \emph{generator only} control retains the constructed map and adapted
model with ordinary beam decoding. \emph{Without completion} applies the
selected combined score to the original 40 candidates. \emph{Without the
collaborative score} sets $d_i=0$ and retains completion. \emph{Without
construction} runs the full inference procedure on the paired static map
and adapted checkpoint; it measures this complete assignment-and-adaptation
choice. \emph{Item correction only} ranks the current catalog by the calibrated
item correction, with ties ordered by item index. \emph{Collaborative
priority} replaces the bound-based candidate order with this same correction,
keeping the final score, bound checks, batch size and budget fixed.
\emph{Without calibration} sets $(\gamma,b)=(1,0)$, retaining the selected
$\lambda$ and applying the global correction to all items.

For each query, the candidate budget counts newly evaluated complete items
outside its initial fully scored set $E_0$. Initial
terminal expansions, encoder work and the linear predictor also incur
cost. Candidate counts and stopping fractions use every test query. Fresh runtime
measurements use the same first 128 test histories per domain and retriever
seed (17/42/2027), with neural inference on one NVIDIA RTX 3090 and two CPU
threads for host-side work. All timed runs use the configurations in Table~\ref{tab:current-scoring}.
Each seed has three timed repetitions, with CUDA synchronization around
each neural stage. The common width-40 beam and linear predictor are
recomputed once per repetition and included in both priority costs.
Each priority separately recomputes its encoder and candidate likelihoods;
no likelihood cache is shared between them or reused from earlier experiments.
Model/input loading and a four-query beam-and-completion warmup for both
priorities are excluded. Beam batches contain 64 histories, completion batches
16, and teacher-forcing calls at most 256 sequences. Priority order alternates
by repetition and seed. We average repetitions within each seed before
computing seed means and sample standard deviations. All timed rankings,
candidate counts and stopping flags match their full-test cached outputs.
Table~\ref{tab:completion-cost} separates full-population search counts from
amortized wall time per query on this fixed, batched runtime sample.

For $I$ items, the dense predictor stores $O(I^2)$ coefficients. A direct
Cholesky fit costs $O(I^3)$; prediction from a sparse length-$T$ history costs
$O(TI)$. Prefix-bound arrays also scale with the current catalog.

\section{Data, Baselines and Evaluation Protocols}
\label{app:protocol}
\label{app:trained-comparison}
\subsection{Task, data, and common training procedure}

We evaluate shared assignments together with ordinary retriever updating on
Amazon 2014 Beauty, Tools, and Toys and Games~\citep{he2016ups}.
Beauty retains the chronological
split described in Appendix~\ref{app:beauty-protocol}; Toys uses the same
preprocessing rule with its own timestamps and content embeddings. Their
initial catalogs contain items observed by the 60th event-time percentile,
and the update catalogs include items observed by the 80th percentile.
Training uses up to the last four eligible next-event prefixes per training
user, with histories truncated to 20 items. Validation users are disjoint
from training users. Test targets are the first subsequent interaction,
with the history ending at the update cutoff. Boundary timestamp ties are
included consistently. Tools uses the public period-0.8 training,
validation, and test sequences supplied with DACT~\citep{feng2026dact}.
Thus all methods within a dataset share the same examples and information;
the two raw-review splits and the supplied Tools split retain their respective
validation conventions.

We use all available update examples: 49,784/1,870 training/validation
examples for Beauty, 4,227/3,007 for Tools, and 41,230/1,559 for Toys.
The assignment objectives draw from all available new-target training contexts,
including histories containing new items: 2,993, 133, and 4,950, respectively.
These cover 532/595 Beauty, 74/158 Tools, and 944/1,025 Toys new items.
The remaining identities stay in the candidate catalog and in evaluation.
The complete test sets contain 7,335, 5,360, and 6,266 queries.
Targets outside the available catalog remain misses in the full-stream
metrics; current-catalog metrics and all denominators are also retained.

The T5 retriever~\citep{rajput2023tiger} architecture and content encodings
are shared within each dataset. The main comparison uses parent seeds
17/42/2027, with checkpoint selection by validation loss.
Beauty and Toys parents train for 40 epochs from random initialization,
using AdamW at $3\times10^{-4}$, weight decay 0.01, batch size 256
and gradient clipping at one. Tools parents start from the public period-0.6
checkpoint and train on period-0.7 data with AdamW at $10^{-4}$,
weight decay 0.01 and batch size 256. Examples targeting items newly admitted
in period 0.7 receive weight two.
Tools parent training stops after 15 epochs without lower validation loss
or at 200 epochs. In both procedures, validation is evaluated every epoch,
including epoch zero, and the earliest minimum-loss checkpoint is retained.

JTM and DREAM use a common five-epoch warm-up before assignment selection,
followed by 15 further training epochs. Training uses AdamW with learning rate $3\times10^{-4}$,
weight decay 0.01, batch size 256, and gradient norm clipping at one.
The optimizer is restarted after assignment. Within each compatible-parent comparison group and seed, methods
receive the same initial weights and minibatch randomization.
BB and its paired without-construction control
use the validation-selected adaptation in Appendix~\ref{app:tuned-construction};
DACT and Reformer use the selected settings specified below.

\subsection{Chronological split and content encoding}
\label{app:beauty-protocol}
Beauty uses Amazon Reviews 2014 5-core reviews and content-only
sentence-T5-base embeddings~\citep{he2016ups,ni2022sentencet5}.
Toys follows the same preprocessing rule with its own timestamps and embeddings.

\paragraph{Time and user boundaries.}
The 60th and 80th percentiles of event time, using the inverted empirical CDF, are Unix
times 1380326400 and 1394150400. All ties at each boundary are included. The old catalog
contains 11,183 items observed by the first cutoff; the current catalog contains 11,778
items observed by the second, admitting 595. At each cutoff, training retains the last
four eligible next-event prefixes per user, with at least three history items and at most
20 inputs. Users whose SHA-256 account-ID integer is zero modulo ten are excluded from
training and supply one last eligible validation prefix. No raw account IDs are released.
This gives 34,836/1,416 parent train/validation examples and 49,784/1,870 update examples.

Each test query contains the last 20 events at or before the second cutoff and predicts
the first event strictly after it. All 7,335 eligible users are retained; 5,732 queries
have pure-old inputs, 1,486 have an admitted new target, and their intersection
contains the 734 Primary queries. Mixed histories remain unknown in an allocation bound.
Future targets are retained and score zero under the current catalog. Metrics decode
predicted paths to eligible item identities: a future item's unallocated placeholder code
must never count as a hit by accidentally matching an eligible item's code.

\paragraph{Fixed content tokenizer and learning.}
We normalize each 768-dimensional content vector and fit three residual KMeans levels,
each with 64 centers, on old items only. Each level uses one initialization, at most 100
iterations, and seed $20260910+d$ for zero-based level $d$. Each item receives its three
center indices and a unique suffix from 0--255. Old items are assigned first, then new
items, each in stable numeric item order; choose the first free suffix. Old paths are
retained exactly. Four position-specific blocks of 256 tokens, with offsets 1/257/513/769,
give a 1,025-token vocabulary including pad/start token zero.

Each parent is a fresh T5 with four encoder and decoder layers, model dimension 128,
feed-forward dimension 1,024, six heads of key/value dimension 64, ReLU and dropout 0.1.

\paragraph{Temporal validation populations.}
Beauty and Toys timestamps have day-level granularity.
Validation predicts the last eligible interaction before the update cutoff;
testing predicts the first interaction after it using pre-cutoff history.
Same-day targets constitute 44.39\% and 50.35\% of their validation sets,
and none of their test sets. The corresponding median prediction gaps
are 7/0 days in validation and 134/150 days in testing.
For adaptation and scoring selection, the positive-day-gap subsets contain
1,040 Beauty and 774 Toys queries. Tools uses all 3,007 validation queries
because timestamps are unavailable.

\subsection{JTM and DREAM assignment objectives}

JTM and DREAM share a score-based candidate bank and retriever procedure.
A seed-controlled sample first selects one new-target context per observed
item in random order, then additional contexts if required, up to 256.
The bank contains the original catalog and up to eight catalogs formed
from the highest-scoring free leaves on these sampled contexts.

For JTM~\citep{zhu2019joint}, we adapt its assignment objective by
maximizing summed training sequence log probability. For DREAM, we use
its UC3 module with temperature-one, entropy-weighted multi-context
probabilities, at most eight candidates per item, minimum support three,
and vote margin 0.05~\citep{guan2026dream}. Its candidate union includes
every item's original code. An item failing either confidence gate
retains that code; rectangular maximum-weight matching resolves competing
assignments among the remaining items.

A shared assignment can change codes in histories as well as targets.
Every complete candidate mapping is therefore re-encoded and decoded on
the full validation set before selection. The original mapping remains
available as a fallback.

\subsection{Parameter selection and evaluation}
\label{app:selection-measurements}

For JTM and DREAM, full-validation NDCG@10 selects
the shared mapping, followed by Recall@10,
fewer changed identities, and candidate order. We evaluate checkpoints
after 5, 10, 15, and 20 total update epochs and select by validation
NDCG@10, Recall@10, and earlier epoch. Test evaluation uses beam width 40,
with Recall and NDCG at
10 and 20. The main comparison uses complete test populations;
Appendix~\ref{app:extended-t5} reports component effects and old/new-target
analyses, with subset definitions and denominators given alongside each result.

\paragraph{Main T5 results and matched ablations.}
We evaluate every method in the main T5 comparison on Beauty, Tools and Toys
using seeds 17, 42 and 2027.
For every seed, validation Recall@10, NDCG@10, Recall@20 and NDCG@20 jointly
guide parameter tuning. We select one configuration per dataset and seed
and use it to report all four test metrics. Table~\ref{tab:current-scoring}
specifies the selected adaptation setting, paired checkpoint epochs, mapping changes
and scoring parameters. Table~\ref{tab:trained-comparison} reports means
and sample standard deviations across these three seeds.
This dispersion includes differences in the adaptation settings and scoring parameters
selected for individual runs.

\paragraph{Initialization and validation populations.}
Each seed has a parent model and a five-epoch update warm-up. Beauty and Toys
use seed-specific T5 parent training on their initial-period data; Tools
uses seed-specific training from a shared public initialization. Constructed
and static branches share the warm-up weights within each dataset and seed.
The selected mapping is fixed before their adaptation trajectories.
For adaptation and scoring selection, Beauty uses 1,040 positive-day-gap
validation queries, Toys uses 774, and Tools uses all 3,007 validation queries.
The training-subset shared-map selector is specified separately in
\eqref{eq:shared-map-selection}.

\paragraph{Adaptation and scoring search.}
The categorical adaptation hyperparameter and checkpoint rules are given in
Appendix~\ref{app:tuned-construction}. Scoring search combines coarse
comparisons and local refinement of $(\lambda,\gamma,b)$.
The coarse Tools comparison uses
$\lambda\in\{0.125,0.25,0.5,1,2\}$,
$\gamma\in\{0,0.25,0.5,1\}$ and $b\in\{0.5,1,1.5,2\}$.
Additional local comparisons use multipliers 0.75/1 for $\lambda$,
increments 0/0.125 for $\gamma$, and $b\in\{1.5,2.25\}$.
Refinement multiplies the reference $\lambda$ by 0.75/1/1.25 and adds
$-0.125/0/0.125$ to $\gamma$ and $-0.25/0/0.25$ to $b$.
Repeated configurations are evaluated once. These comparisons use the same
four validation metrics; the reported parameter values identify the
configuration used for each test run.

\begin{table}[!htbp]
\centering\small
\setlength{\tabcolsep}{2.5pt}
\caption{Configurations for the nine T5 runs in Table~\ref{tab:trained-comparison}.
Adaptation setting is a categorical hyperparameter with three values defined in
Appendix~\ref{app:tuned-construction}. Selected update epochs include the common
five-epoch warm-up and exclude parent training; epoch 5 denotes the trained
warm-up checkpoint. They locate the selected checkpoints, not the stopping
epochs. IDs count changed item codes; leaves count newly occupied paths,
both relative to the paired static mapping. Validation uses the cross-day subset
for Beauty/Toys and the complete split for Tools.}
\label{tab:current-scoring}
\begin{tabular}{lllrrrrrrr}
\toprule
\multirow{2}{*}{Dataset} & \multirow{2}{*}{Seed} & \multirow{2}{*}{Adaptation setting} &
\multicolumn{2}{c}{Selected update epoch} & \multirow{2}{*}{IDs} &
\multirow{2}{*}{Leaves} & \multirow{2}{*}{$\lambda$} &
\multirow{2}{*}{$\gamma$} & \multirow{2}{*}{$b$} \\
\cmidrule(lr){4-5}
& & & BB & W/o construction & & & & & \\
\midrule
Beauty & 17 & Constant & 25 & 15 & 5 & 1 & 1.125 & 1.25 & 2.5 \\
Beauty & 42 & Constant & 15 & 5 & 2 & 0 & 1.5 & 1.25 & 2.5 \\
Beauty & 2027 & Weighted EMA & 10 & 10 & 2 & 0 & 1.125 & 1.125 & 2.5 \\
\addlinespace[2pt]
Tools & 17 & Cosine+EMA & 10 & 10 & 0 & 0 & 0.625 & 0.25 & 1.75 \\
Tools & 42 & Cosine+EMA & 5 & 5 & 2 & 0 & 0.75 & 0.25 & 2.25 \\
Tools & 2027 & Cosine+EMA & 5 & 5 & 0 & 0 & 0.25 & 0.625 & 2.25 \\
\addlinespace[2pt]
Toys & 17 & Constant & 45 & 10 & 2 & 0 & 0.25 & 1 & 1 \\
Toys & 42 & Weighted EMA & 60 & 60 & 1 & 1 & 0.5 & 1 & 1.5 \\
Toys & 2027 & Constant & 25 & 25 & 2 & 0 & 0.25 & 1 & 1 \\
\bottomrule
\end{tabular}
\end{table}

\paragraph{Matched component controls.}
Component removals inherit the corresponding BB scoring parameters.
Without construction uses the static map and paired adaptation setting,
with its selected checkpoint shown in Table~\ref{tab:current-scoring}.
The comparison therefore measures assignment together with subsequent
adaptation. Score and candidate-access controls retain the BB map and
checkpoint and change only the specified inference component.

\subsection{DACT Comparison Across Three Domains}
\label{app:dact-comparison}

\subsubsection{Data and compatible initialization}

DACT~\citep{feng2026dact} is evaluated on the same Beauty, Tools, and Toys
queries and candidate catalogs as Table~\ref{tab:trained-comparison}.
Each domain has a collaboration-aware tokenizer and a retriever trained on
its initial codes. Adaptation therefore starts from a model that has learned
the preceding DACT encoding. Hyperparameters are tuned on validation data
starting from the published configuration. Beauty and Toys use the existing initial-period
training examples. Tools uses the deduplicated training examples from its
first two observed periods: 41,576 plus 3,858 examples, with 2,993 validation
examples at the second origin. This cumulative initialization retains the
earlier interactions available to the other Tools parents.
The update training, validation, and test examples are unchanged.

\subsubsection{Collaborative tokenization}

We use the public DACT encoder, residual quantizer, memory gate, and
differentiated update objective. Item content is represented by the existing
768-dimensional embeddings. The encoder has widths 512, 256, 128, 64, and 32;
the decoder reverses these widths. Three codebooks each contain 256 vectors
of dimension 32. Initial training minimizes reconstruction, quantization,
and collaborative alignment losses, with commitment weight 0.25 and
alignment weight 0.02. It runs for 20,000 AdamW steps at learning rate
$10^{-3}$. DACT updating runs for 5,000 steps at $10^{-4}$, with batch size
1,024, weight decay $10^{-4}$, linear decay after 500 warm-up steps, and
gradient clipping at one. These are optimizer-step budgets, following the
paper's implementation description.

The update uses a 0.3 drift-selection fraction, gate temperature 0.1,
old-branch weight 0.1, latent-stability weight one, assignment-divergence
weight five, and gate penalty 0.001. The public implementation applies the
assignment-divergence term to the first codebook. Code reassignment retains
the entire previous code when its first token remains unchanged. Retained
codes are reserved before assigning collision suffixes to changed and new
items, giving an injective four-token map without altering those retained
identities.

Collaborative features come from a 32-dimensional, three-layer
LightGCN~\citep{he2020lightgcn}. Its graph contains deduplicated user--item
pairs from training histories and targets. The update inherits the initial
embedding parameters; sampled negatives exclude all known training
positives from both periods. Full-batch Bayesian personalized ranking uses
Adam at $10^{-3}$ and regularization $10^{-4}$. Validation uses the mean
history-item embedding to represent each query, including held-out users.
Validation is checked every 20 epochs, with patience five and a
300-epoch cap in every domain. These fixed collaborative-feature preparation
settings are shared across domains; the public DACT release does not specify
a collaborative-feature training configuration for these splits.

\subsubsection{Retriever training and selection}

The retriever follows the public DACT training defaults: AdamW at learning
rate $10^{-4}$, batch size 256, weight decay 0.01, and at most 200 epochs.
Full-validation cross entropy is evaluated every epoch in batches of 128,
using the arithmetic mean of batch losses. Training stops after 15 consecutive
epochs without a strictly lower validation loss; the minimum-loss checkpoint
is selected, with earlier epochs retained on ties. Each phase uses one
optimizer, without a learning-rate scheduler or gradient clipping.

The same fixed recipe trains the encoding-compatible parent from random
initialization and adapts its selected checkpoint to the updated codes,
with a fresh optimizer for adaptation. The public release supplies the
adaptation recipe; we also apply it to parent training on these shared
benchmark splits. The four-layer T5 dimensions follow the public model;
the per-domain vocabulary and data are those of the common benchmark.
The random seed is set before model initialization, and the best checkpoint
is saved separately from the current training state.

All nine selected checkpoints are fixed before testing. Decoding uses beam
width 40 and the complete current catalog; unavailable future targets remain
misses. Three retriever seeds share one tokenizer and collaborative-feature
fit per domain, so the reported standard deviations measure retriever-seed
variation conditional on those fits. DACT's own parent and adaptation budget
differ from the common-warm-up comparisons described in
Appendix~\ref{app:trained-comparison}.

\subsection{Standard Update Strategies and Reformer}
\label{app:continual-updates}

This comparison evaluates how tokenizer and retriever updates affect
recommendation under the same catalog expansion. Table~\ref{tab:trained-comparison}
includes all five standard strategies and Reformer; Tables~\ref{tab:continual-at20}
and~\ref{tab:continual-groups} give cutoff-20 and old/new-target results.
All methods rank the same current catalog for the same test histories,
using a width-40 beam and seeds 17, 42, and 2027.

\paragraph{Standard update strategies.}
Following the tokenizer/GRM comparison in DACT~\citep{feng2026dact},
Frozen/Frozen retains both the initial tokenizer and retriever;
FT/Frozen updates only the tokenizer; Frozen/FT updates only the retriever;
FT/FT updates both; and FT/RT updates the tokenizer and trains a new retriever
from random initialization. These strategies share DACT's initial residual
quantizer, collaborative features, and corresponding initial T5 models.
The frozen tokenizer preserves existing codes and assigns new items using
its unchanged quantizer. Tokenizer fine-tuning re-encodes all items with
the updated quantizer. A fourth token resolves complete-code collisions.
The frozen-model and static-assignment rows in the main table instead use
the initial models of the construction comparison
(Appendix~\ref{app:trained-comparison}) and serve as separate controls.

Plain tokenizer fine-tuning minimizes reconstruction, quantization, and
collaborative-alignment losses, with alignment coefficient 0.02.
It uses 5,000 AdamW steps, learning rate $10^{-4}$, batch size 1,024,
weight decay $10^{-4}$, 500 warmup steps, and gradient clipping at 1.
The final tokenizer is used in each strategy. Retriever fine-tuning uses
update training examples. Retraining uses the deduplicated union of initial
and update training examples: 55,909 for Beauty, 49,661 for Tools,
and 46,964 for Toys. Both use AdamW at $10^{-4}$, batch size 256,
weight decay 0.01, and a 200-epoch maximum. Full update-validation
cross-entropy is evaluated each epoch with batch size 128; 15 consecutive
epochs without a strict improvement stop training, with earlier ties retained.
Frozen retrievers require no additional fitting.

\paragraph{Reformer.}
We evaluate Reformer's incremental quantizer and identifier assignment
\citep{shi2025incremental} with the same initial tokenizer and T5 models
as the standard strategies. Its three semantic codebooks each gain 16
entries, while existing item identifiers remain fixed. New items train
the quantizer with the published reconstruction and quantization losses;
warm items establish the frequency counts. We use the released frequency
exponents $(0.25,0,0)$, power scaling, and Sinkhorn configuration
$(0,0,0.003)$ with 50 iterations. New-item assignment uses unit distance
scales and up to 11 collision-resolution passes; the common fourth-token
suffix resolves remaining catalog collisions. The retriever gains 48
semantic-token embeddings while retaining all existing embedding rows.
Tokenizer updates use the common 5,000-step configuration above, and
retriever adaptation uses the same optimizer and stopping rule as
the standard strategies. Thus the comparison preserves Reformer's update
mechanisms under the common initialization and training protocol.
Hyperparameters are tuned on validation data starting from the published settings.

\subsection{Conventional recommendation references}
\label{app:default-baselines}
\paragraph{Shared information and evaluation.}
SASRec, EASE and ItemKNN use the same test histories, target items and current
catalog as Table~\ref{tab:trained-comparison}. Training combines the available
initial and update examples, removing duplicate user--history--target records.
This yields 55,909, 49,661 and 46,964 training examples for Beauty, Tools and Toys.
Current validation and test examples are not used as supervised training rows.
The original chronological splits and maximum 20-item histories are retained.
All current-catalog items are ranked, including previously observed items;
future-catalog targets remain misses. Equal scores are ordered by item identity.
These baselines fit accumulated training data, whereas the generative comparisons
use the update procedures specified in their respective appendices.

\paragraph{Selected baseline settings.}
We tune hyperparameters on validation data, using the RecBole v1.2.1
model and training defaults as starting points.\footnote{\url{https://github.com/RUCAIBox/RecBole/tree/v1.2.1}}
SASRec uses two layers, two attention heads, hidden width 64, feed-forward width
256, GELU, hidden and attention dropout 0.5, layer-normalization tolerance
$10^{-12}$, and normal initialization with standard deviation 0.02.
Its position table supports 50 positions; actual input histories retain the
common 20-item maximum. The reference cross-entropy variant predicts the next
item from the last observed position. Adam uses learning rate $10^{-3}$,
batch size 2,048 and no weight decay, scheduler or gradient clipping.
Validation MRR@10 is evaluated after every epoch and rounded to four decimal
places, following the reference evaluation precision. Ties replace the current
best checkpoint and reset patience; training stops after more than ten
consecutive strictly worse epochs or at the fixed 300-epoch cap.
All nine runs stop by this rule. Seeds are 17, 42 and 2027.

EASE fits its zero-diagonal linear reconstruction model with regularization
250. ItemKNN uses cosine similarity, 100 neighbors per output item, zero
shrinkage and the reference $10^{-6}$ denominator offset. Both use a binary
user--item matrix formed only from the available training histories and targets.
The matrices contain 132,687, 65,788 and 112,002 observed user--item pairs,
respectively. At prediction, each method multiplies a query's binary observed
history by the fitted item-weight matrix. This history fold-in uses the same
query information as the other recommenders and does not fit evaluation targets.
Each deterministic model is fitted once per domain.

The complete-population and cohort results are reported in
Appendix~\ref{app:extended-t5}.

\section{Feasibility Validation and Certificate Coverage}
\label{app:assignment-diagnostics}
This appendix evaluates assignment feasibility and certificate coverage
under fixed scoring models. The exhaustive construction comparison uses 96
deterministic scorers; the independent Beauty parent/update comparison uses seeds
17, 29 and 43. Recommendation quality for the final T5 method is evaluated
separately in Appendix~\ref{app:extended-t5}.

\subsection{Exhaustive Evaluation of Assignment Invariance}
\label{app:exact-calibration}

The experiment uses $\mathcal Z=\{0,1\}^3$, widths one and two, 16 seeds, and every
nonempty proper old-path subset: the integer masks 1 through 254, inclusive.
For each query/scorer, positive integer weights derived from a deterministic SHA-256 key
are normalized to exact rational probabilities. The prefix family keys on seed and prefix;
the state family additionally keys on the entire ordered beam and its exact accumulated
probabilities. Ties in decoding use lexical prefix order. Certification uses strict adjacent
gaps and therefore rejects relevant ties rather than relying on this tie-breaking order.

We enumerate all $A\supseteq O$, including $A=O$ and the full codebook. Ground-truth
invariance means equality of the complete ordered final list across this family.
At width two, 992 prefix-scoring cases and 987 state-scoring cases never
return a new path. Only 976 and 965 cases, respectively, preserve the
complete ordered output list. A query can have an invariant
new-target contribution without an invariant ordered old list. The certificate establishes
the stronger list-level claim. Relevant top-rank ties occur in one coverage query/scorer
case and nine cohort query/scorer cases; exact arithmetic retains them.

For cohort calibration, use eight seeds, four fixed queries, two new item identities, and
old-catalog sizes two, four, or six. Query target identities are $(0,1,0,1)$. Every ordered
injection of the two new identities into unused paths defines one shared allocation.
The separate optimum allows a different legal injection for each query, whereas the shared
optimum maximizes the four-query hit count with a single injection. The bound treats all
uncertified queries as potentially correct. All three objectives use the same fixed item
labels; changing an occupied path set without tracking identity would not define this test.

An independent implementation shares only the deterministic probability specification
with the primary implementation. It separately computes catalog decoding,
full-codebook certification, ground-truth invariance, and both allocation optima;
all 16,256 coverage rows and 4,032 cohort rows match exactly. The CPU evaluation
takes approximately three seconds for this complete finite design. Experimental
specifications, per-instance results, and explicit examples of missed invariance
and conflicting query optima are included in the supplementary material.

% Derived from independently verified AL01 summary.json; exact counts, descriptive means.
\begin{table}[!htbp]
\caption{Exhaustive calibration of assignment invariance. Each row contains 4,064 single-query cases and 1,008
four-query cohorts. $C/I$ is certified / all-allocation invariant; the final three columns are
cohort means in percent. $U_C\geq\mathrm{OPT}_{separate}\geq\mathrm{OPT}_{shared}$ throughout.
The width-one equality of $C$ and $I$ is an observation in this finite design.}
\label{tab:al01-calibration}
\centering\small
\begin{tabular*}{\linewidth}{@{\extracolsep{\fill}}lrrrrr@{}}
\toprule
Scoring & $W$ & $C/I$ & $U_C$ & Separate optimum & Shared optimum \\
\midrule
Prefix & 1 & 2,032/2,032 & 51.56 & 50.00 & 39.51 \\
Prefix & 2 & 910/976 & 78.94 & 75.40 & 65.30 \\
Beam state & 1 & 2,032/2,032 & 54.69 & 50.00 & 41.10 \\
Beam state & 2 & 866/965 & 80.51 & 75.97 & 62.85 \\
\bottomrule
\end{tabular*}
\end{table}

\subsection{Can allocation alone match a model update?}
\label{sec:cohort-results}

The output-invariance certificate fixes the returned old items on 6,253 of
6,398 VK-LSVD (VK) Primary queries. Thus every legal allocation of that parent has
new-item Recall@20 at most $145/6{,}398=2.27\%$, compared with the update's
$8.82\%$. The gap excludes allocation alone as an explanation for the update's
performance. On the prospective Beauty cohort, two of three updates also
exceed their width-20 bounds (Table~\ref{tab:assignment-bounds}).
% Historical values unchanged; AL02 rows independently verified.
\begin{table}[!htbp]
\caption{Can an allocation alone match the updated retriever? Each row evaluates the complete pure-old-input, new-target cohort. $|C|$ queries are certified to contribute zero under every legal allocation of the fixed parent. $U=1-|C|/N$ bounds new-item Recall@20. Updated-model recall uses the indicated candidate mapping and the same cohort; all six Tools comparisons and Beauty seed 17 remain below the bound. Beauty uses newly trained models for the chronological comparison. Percent units.}
\label{tab:assignment-bounds}
\centering\small
\setlength{\tabcolsep}{3.3pt}
\begin{tabular*}{\linewidth}{@{\extracolsep{\fill}}llrrrrc@{}}
\toprule
Data/seed & Candidates & $|C|/N$ & $U$ & Updated & Gap (pp) & $R>U$ \\
\midrule
VK/42 & Hungarian & 6,253/6,398 & 2.27 & 8.82 & +6.55 & Yes \\
Tools/17 & Frozen & 172/368 & 53.26 & 13.59 & -39.67 & No \\
Tools/17 & Refreshed & 172/368 & 53.26 & 9.51 & -43.75 & No \\
Tools/29 & Frozen & 224/368 & 39.13 & 14.67 & -24.46 & No \\
Tools/29 & Refreshed & 224/368 & 39.13 & 11.96 & -27.17 & No \\
Tools/43 & Frozen & 173/368 & 52.99 & 12.50 & -40.49 & No \\
Tools/43 & Refreshed & 173/368 & 52.99 & 12.23 & -40.76 & No \\
\midrule
Beauty/17 & Fixed content & 712/734 & 3.00 & 2.72 & -0.27 & No \\
Beauty/29 & Fixed content & 709/734 & 3.41 & 5.04 & +1.63 & Yes \\
Beauty/43 & Fixed content & 717/734 & 2.32 & 4.63 & +2.32 & Yes \\
\bottomrule
\end{tabular*}
\end{table}

Tools shows the complementary outcome. All six updates remain below their
39.13--53.26\% bounds, leaving allocation-only repair unresolved despite
updated-model recovery on parent-certified subsets. Adding mixed-input queries to the VK and Beauty
new-target cohorts also makes their bounds inconclusive. On VK, the denominator
increases to 13,676 while 6,253 queries remain certified, giving a 54.28\%
bound versus 14.21\% updated Recall@20. On Beauty, the denominator increases
from 734 to 1,486. The certified counts remain 712/709/717 for seeds 17/29/43,
giving bounds of 52.09/52.29/51.75\% versus updated recall of
3.23/5.65/5.59\%. Mixed histories contribute the maximum possible value one
to these bounds because their inputs can change with the assignment.
These comparisons measure the coverage of the certificate as well as the
size of the established separation.

\subsection{How much does branch analysis improve coverage?}
\label{sec:branch-results}

Increasing the beam width exposes a limitation of the original certificate.
At width 40, its Beauty bounds rise to 7.49--10.22\%, exceeding all three updated
Recall@20 values. At width 80, all first-layer checks fail because 64 old root
symbols cannot fill 80 strictly ordered old positions. The resulting bounds
leave every width-80 comparison unresolved
(Appendix~\ref{app:width-protocol}).

Final-layer certification and penultimate branching recover useful width-40
bounds. Final-layer certification permits changes in the old ordering while
checking that Top-20 contains only old items. Branch analysis additionally
enumerates penultimate alternatives after the first two layers are certified.
It adds 31/29/23 certified query--seed configurations to the final-layer test,
raising the combined counts to 708/705/714 out of 734 queries.
The corresponding bounds fall to 3.54/3.95/2.72\%, and two updates again exceed
them (Table~\ref{tab:branch-bounds}). The remaining 26/29/20 queries contribute
one each to the bounds.
\begin{table}[!htbp]
\caption{Stronger certificates on the same Beauty cohort ($W=40$, Recall@20;
percent units). All-layer, final-layer, and branch certificates use the same
parent checkpoints. Combined counts are over all 734 queries. Updated-model
outcomes are reused. Branch analysis is a post-hoc extension of the original width experiment.}
\label{tab:branch-bounds}
\centering\small
\begin{tabular*}{\linewidth}{@{\extracolsep{\fill}}lrrrrrrc@{}}
\toprule
Seed & \shortstack{All-layer\\$U$} & \shortstack{Final-layer\\$U$} & Added & \shortstack{Combined\\$|C|$} & New $U$ & Updated & $R>U$\\
\midrule
17 & 9.81 & 7.77 & 31 & 708 & 3.54 & 3.00 & No\\
29 & 10.22 & 7.90 & 29 & 705 & 3.95 & 5.18 & Yes\\
43 & 7.49 & 5.86 & 23 & 714 & 2.72 & 4.36 & Yes\\
\bottomrule
\end{tabular*}
\end{table}

The added configurations comprise 69 distinct test rows across three seeds.
Their critical sets contain $q=1$--6 optional branches: 324 eligibility subsets
merge into 317 ordered states. Of these, 83 reference states are reused and
234 require terminal batch calls. This post-hoc extension of the Beauty
width experiment uses the existing parents and updates. Tools' fixed-condition
applicability analysis is given in Appendix~\ref{app:tools-transfer};
computation for the broader constructive comparison is reported in
Appendix~\ref{app:construction-comparison}.

\subsection{Beam-Width Sensitivity}
\label{app:width-protocol}

Widths 20, 40 and 80 were fixed for all three Beauty parent/update
pairs before observing the chronological-update test outcomes. All widths use
the same checkpoints, full test cohort and seeds. Recall@20 uses the first
20 returned items.

Width 20 reuses the baseline predictions; widths 40 and 80 run new
current-catalog/full-codebook parent and current-catalog update passes. Each width is
certified under its own scoring and state transitions. Independent
reconstruction reproduces the certificate decisions and metrics for the six
additional width/seed conditions. Repeated checkpoint inference also checks
predictions and scores at both wider widths for seed 17.

The width-80 failure follows from the old codebook's at most
64 distinct first-level symbols, fewer than the 80 strict old extensions required by the
condition. This setting measures the resulting structural boundary of the certificate.

% Generated from the full AL02/AL03 independent verification.
\begin{table}[!htbp]
\caption{Complete prespecified beam-width comparison. Recall@20 always uses the first 20 returned items on the same 734 Primary queries. Width 20 uses the baseline predictions. A different width requires a new certificate. All values except counts/widths are percentages or percentage-point gaps.}
\label{tab:al03-width}
\centering\small
\begin{tabular*}{\linewidth}{@{\extracolsep{\fill}}rrrrrr@{}}
\toprule
Seed & $W$ & $|C|/N$ & Bound & Updated R@20 & Gap (pp) \\
\midrule
17 & 20 & 712/734 & 3.00 & 2.72 & -0.27 \\
17 & 40 & 662/734 & 9.81 & 3.00 & -6.81 \\
17 & 80 & 0/734 & 100.00 & 3.00 & -97.00 \\
29 & 20 & 709/734 & 3.41 & 5.04 & +1.63 \\
29 & 40 & 659/734 & 10.22 & 5.18 & -5.04 \\
29 & 80 & 0/734 & 100.00 & 5.18 & -94.82 \\
43 & 20 & 717/734 & 2.32 & 4.63 & +2.32 \\
43 & 40 & 679/734 & 7.49 & 4.36 & -3.13 \\
43 & 80 & 0/734 & 100.00 & 4.36 & -95.64 \\
\bottomrule
\end{tabular*}
\end{table}

\subsection{Fixed-count and branch applicability}
\label{app:tools-transfer}
The Beauty branch study reuses 83 query--seed score arrays containing
317 states and 3,246,080 terminal candidates. At $m=595$, every repair
interval is empty. Capacities range from 4,294,954,577 to 4,294,956,113
and support-group counts from zero to six; the count endpoints are not
tight. The exact count criterion adds no certified queries on these
states. Its endpoint behavior is tested in the exhaustive finite designs.
In the predeclared Tools condition, 172 of 368 queries are certified,
121 fail the common-prefix premise and 75 have a new leaf in the
full-reference Top-20. No query enters additional branch enumeration.
The bound remains 53.26\%, above the update's 13.59\%, with no exclusion.
The complete coverage partition and original score arrays are retained.
\subsection{Constructive repair comparisons}
\label{app:construction-comparison}

\subsubsection{Methods and common evaluation rules}

The comparison concerns per-query existence: a method receives the designated
target identity and may assign it any free leaf while retaining every old item
and admitting exactly $m$ new identities. The output is one complete assignment
for that query. This task measures constructive reachability for a supplied
target identity. The comparisons reuse previously evaluated checkpoints.
Method choices, finite designs, query selection rules and computation limits
were fixed before the comparative executions.

The initial finite allocation uses the first $m$ unoccupied leaves in numerical
order and assigns the target to the first one. Beauty uses its reference
assignment. Random search draws four legal $m$-leaf catalogs, independently
assigns the target to a selected new leaf, and reports whether any trial
retrieves it. All four trials are evaluated, including after success.

High-score allocation selects the highest-scoring new leaf in the full-codebook
reference expansion and fills the remaining slots in numerical order within
the reference-compatible universe: optional prefixes above its beam cutoff
that are absent from that state remain disabled. It does not enforce support
for selected optional prefixes, so decoding the catalog may change the beam. Reference-state
exact construction applies the complete support and count criterion while
retaining only the full-reference penultimate state. Full construction examines
that state first and then the remaining states in deterministic order, returning
the first verified witness or certifying failure after complete enumeration.
The lexical-padding ablation retains exact state and target selection and support
choices but fills remaining slots in numerical order instead of preferring
non-outrankers. Every returned mapping is evaluated by the complete decoder.

\subsubsection{Finite end-to-end comparisons and condition ablations}

\begin{table}[!htbp]
\caption{Constructive comparisons. Finite columns report successful target retrieval
among all feasible conditions; Beauty columns report repairs among all remaining
common-prefix queries (40/44/27), each initially missed. The initial allocation is
lexical append in finite problems and the reference mapping on Beauty. Random
search reports success in any of four trials. Each repair uses its own assignment.}
\label{tab:construction-comparison}
\centering\small
\setlength{\tabcolsep}{3.3pt}
\begin{tabular*}{\linewidth}{@{\extracolsep{\fill}}lrrrrr@{}}
\toprule
& \multicolumn{2}{c}{Finite scoring family} & \multicolumn{3}{c}{Beauty seed}\\
Method & Prefix-local & Complete-state & 17 & 29 & 43\\
\midrule
Initial allocation & 170/882 & 107/934 & 0/40 & 0/44 & 0/27\\
Random, four trials & 488/882 & 494/934 & 0/40 & 0/44 & 0/27\\
High-score allocation & 882/882 & 733/934 & 9/40 & 15/44 & 3/27\\
Reference-state exact & 850/882 & 858/934 & 9/40 & 15/44 & 3/27\\
Full construction & 882/882 & 934/934 & 9/40 & 15/44 & 3/27\\
Lexical padding & 882/882 & 934/934 & 9/40 & 15/44 & 3/27\\
\bottomrule
\end{tabular*}
\end{table}

\begin{figure}[t]
\centering
\includegraphics[width=\linewidth]{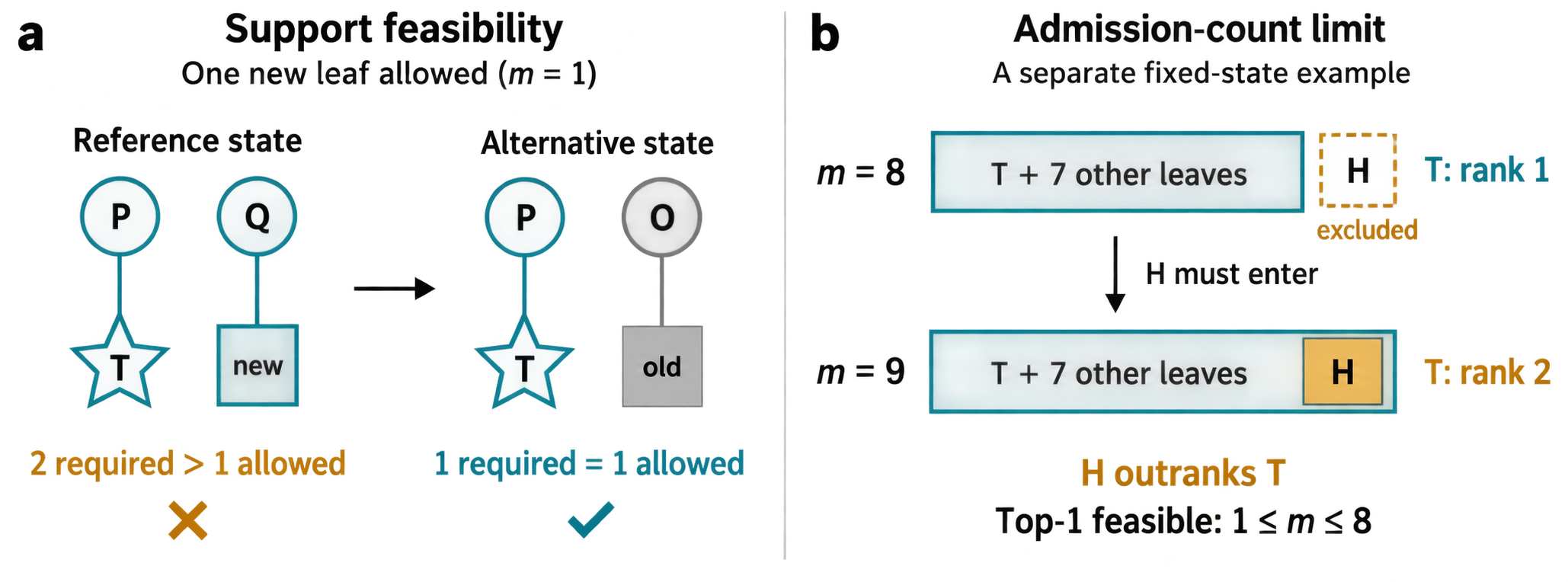}
\caption{Support and admission-count mechanisms in two exact finite examples.
\textbf{a}, At $m=1$, reference prefixes $P=(1,1)$ and $Q=(2,1)$ each require
a new supporting leaf. The alternative replaces $Q$ by the old-supported
prefix $O=(1,0)$, allowing target leaf $T=6$ to be retrieved with one admission.
The left drawing shows required supports, not a feasible one-admission catalog.
\textbf{b}, In a separate fixed state, the nine optional leaves include
target $T=11$ and its higher-scoring competitor $H=10$.
At $m=8$, excluding $H$ gives target rank one; admitting all nine leaves
forces $H$ into the catalog and lowers the target to rank two.
Exhaustive decoding verifies Top-1 feasibility exactly for $1\le m\le8$;
at $m=0$ the target is absent. Mandatory old leaves remain fixed and are
omitted in (b).}
\label{fig:construction-mechanisms}
\end{figure}

We use three-layer alphabets $(3,2,2)$ and $(4,2,2)$ with widths two and three.
Every root has one mandatory old leaf at suffix $(0,0)$, so the root beam is
common to every containing catalog. Each shape has 24 deterministic scoring
seeds for each of two families: prefix-local scoring and scoring that depends
on the complete ordered penultimate state. Hash-derived positive integer
weights are normalized to exact rational probabilities; ties use lexical path
order. There are 96 scorers, with 9 or 12 optional leaves per scorer.

Independent full beam decoding exhausts all 221,184 containing catalogs.
For every positive admission count and $K\in\{1,2\}$, it determines whether
any assignment can retrieve the target. This yields 2,016 conditions, 1,816
feasible. The finite columns of Table~\ref{tab:construction-comparison} use
these feasible denominators, 882 for prefix-local and 934 for complete-state
scoring. Full construction also correctly rejects all 200 infeasible conditions.
High-score allocation solves all feasible prefix-local conditions in this
finite design; its 201 misses all occur under complete-state scoring.
Reference-state exact construction misses 32 prefix-local and 76 complete-state
conditions because another beam is needed. Lexical padding has no observed
success penalty in this comparison.

Condition ablations test a different endpoint: feasibility of a specified leaf
at a specified state and count. Removing support groups also removes their
contribution to required counts and forced outrankers. The other two ablations
remove only the lower or upper count endpoint, retaining the compatible-state
capacity for the latter. Across 24,630 decisions, these removals produce
1,586, 659, and 802 false positives, respectively, with no false negatives.
The full criterion has neither error. These deliberately incomplete predicates
serve as diagnostic checks of the necessary conditions.
An independent local ranking experiment checks 193,408 decisions and
24,252 constructed witnesses. It enumerates 67,392 catalogs across 2,400
ranking instances with local universes of six or eight leaves, disjoint two-leaf
support groups, every old mask and eligible group subset, four ranking
seeds, and $K\in\{1,2,3,4\}$. Two leaves are outside the terminal expansion.
It tests forced support, insufficient count and off-beam padding with no
false-positive or false-negative decisions.

Figure~\ref{fig:construction-mechanisms} uses two deterministically selected
examples from this enumeration. In panel a, $m=1$ cannot support both optional
prefixes in the reference state; replacing one with an old-supported prefix
allows target leaf 6 to be retrieved. In panel b, the specified target is leaf
11, $K=1$, and the required group is $\{10,11\}$. Its predicted interval is
$[1,8]$. At $m=9$, leaf 10 must be included and outranks the target. At $m=0$,
the target cannot be admitted, so no rank is plotted.

\begin{figure}[t]
\centering
\includegraphics[width=\linewidth]{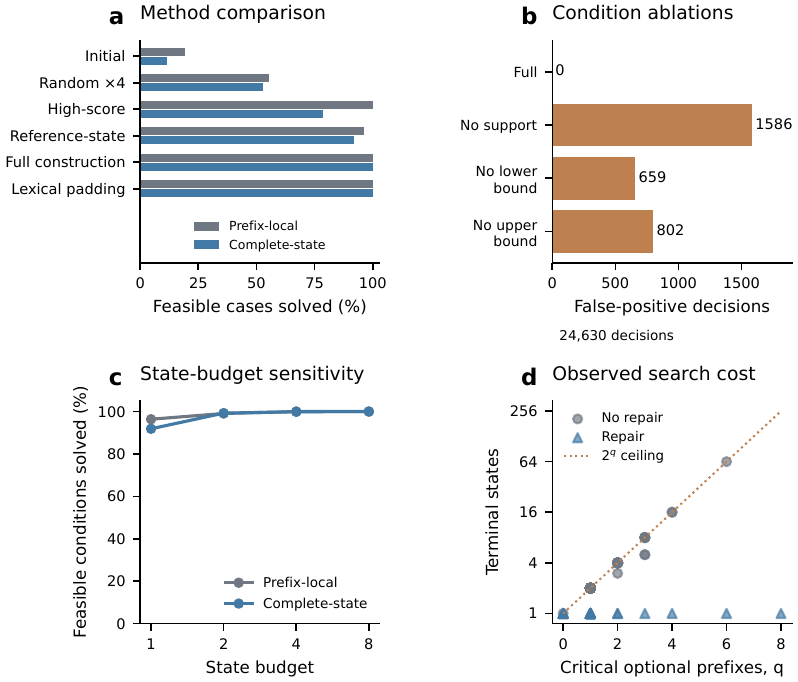}
\caption{Method comparisons, component removals, and search cost.
\textbf{a}, Success among feasible conditions in each finite scoring family
(882 and 934); bars are exact descriptive proportions.
\textbf{b}, False feasibility decisions among 24,630 state--target--count
conditions after removing support or count endpoints.
\textbf{c}, Feasible conditions solved with a fixed terminal-state budget;
unfinished enumeration leaves unresolved cases rather than certifying failure.
\textbf{d}, States evaluated for all 111 trained-model comparisons, with the
$2^q$ ceiling. Successful construction may stop at the first state, whereas
failure requires enumeration. Repeated conditions and query--seed pairs are
within-design measurements.}
\label{fig:construction-results}
\end{figure}

\subsubsection{Paired evaluation with trained Beauty models}

We reuse the three Beauty parents at width 40 and $K=20$, with 11,183 old
items and 595 new items. Each seed retains all 734 Primary queries in the
accounting. Before new comparisons, final-layer certification resolves
677/676/691 queries, while 17/14/16 fail the common first-two-layer condition.
All remaining 40/44/27 queries enter the comparison, including cases with
and without a new leaf in the reference Top-20. The selected query set is
retained after successful repairs.
Every selected query is initially missed under its reference assignment.

Each parent runs on an NVIDIA RTX 3090 with PyTorch 2.6.0+cu124 and
Transformers 5.9.0, deterministic float32 scoring, disabled TF32, and the
original 64-query batches and within-batch slots. CuBLAS uses the deterministic
workspace configuration \texttt{:4096:8}. The complete pre-branch observations
and reference terminal scores agree exactly with the corresponding recorded
GPU evaluations. Branch continuations retain the actual caches. Enumeration
allows at most eight critical optional prefixes and 256 states per query;
no selected query exceeds these limits. Every method uses the same model,
input, catalog size, old identity mapping, and numerical execution shape.

Independent checks verify 636 complete legal assignments across all methods
and random trials. For every full-construction witness, direct decoding
reproduces the constructed state, terminal scores and predicted target rank.

Full construction repairs 9/15/3 queries, matching both high-score allocation
and reference-state exact construction. With four trials per query, random
search retrieves none of the 111 targets. All 27 successful repairs occur
in the reference state. Five need no new-prefix support; the other 22
require one to seven groups (the counts for 1/2/3/4/6/7 groups are
16/2/1/1/1/1).

The remaining 84 compared queries have no repair in the enumerated family.
All 318 of their terminal states have at least 20 strictly higher-scoring
mandatory old leaves than the best new leaf, so internal old-score ties do not
change these negative decisions. The rank-based comparison includes one
seed-43 query outside the certificate coverage in Table~\ref{tab:branch-bounds}:
its old scores contain an internal tie, but 49 old leaves strictly outrank
its best new leaf.

\begin{figure}[t]
\centering
\includegraphics[width=\linewidth]{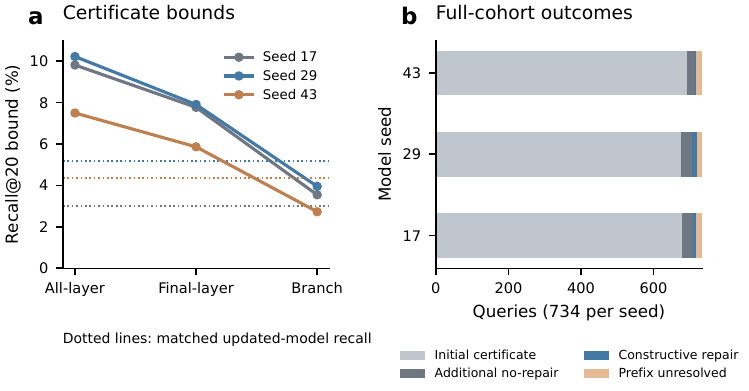}
\caption{Certificate strength and complete-cohort accounting.
\textbf{a}, Original all-layer, final-layer, and branch bounds at Beauty width 40;
dotted lines are the matched updated-model recalls.
\textbf{b}, Every Primary query is retained after constructive comparison:
initial certification, additional no-repair decisions, query-specific successful
constructions, or unresolved common prefixes. The seed-43 partition includes
the additional tied-old-score case described in the text. A different assignment
may realize each repair; the blue portion aggregates per-query repairs.}
\label{fig:certification-scope}
\end{figure}

\subsubsection{Budgets, computation, and statistical units}

The finite budget comparison limits the deterministic state order to
1, 2, 4, or 8 terminal evaluations and retains unknown outcomes when the
budget does not complete a negative decision. The trained-model comparison
evaluates 345 states across 111 queries, including 111 reference states and
234 additional terminal continuations. Critical-set sizes range from zero
to eight. Positive cases stop at their verified first-state construction;
negative cases require all reachable states. Figure~\ref{fig:construction-results}
shows these search outcomes; the sufficient certificate's single-pass cost is
measured separately.

The three complete paired studies take 42.88/52.66/28.58 seconds, including
model loading, all comparison methods, validation, and output recording, but
excluding interpreter and library startup. Peak allocated GPU memory is
5.67 GiB per study. The shared reference reconstruction uses four decoder
steps per compared query. Together, the methods use 1,283/1,547/836 decoder batch steps in the three studies.
Per-method synchronized execution times, shared-reference times, search and
recording times, and forward counts are supplied with the individual results.
The methods have different computation costs: random search uses four full
decodings, while exact search depends on the state count.

For the finite designs, summaries distinguish the 96 scorers, count/cutoff
conditions and state-specific decisions. Trained-model summaries distinguish
the three parents and their query--seed pairs. Random trials and repeated
states are within-study measurements. All summaries are descriptive;
no significance test is applied.

\subsection{Decoder Implementation and Experimental Details}
\label{app:native-closure}

\subsubsection{Parent/update models and data provenance}

The Beauty models use the chronological split and tokenizer in
Appendix~\ref{app:beauty-protocol}. Each parent is trained from scratch.
AdamW uses learning rate $3\times10^{-4}$, weight decay 0.01,
gradient clipping at one, and batch size 256. All 40 parent and 20 update epochs run;
minimum validation teacher-forced loss selects each checkpoint. The update starts at its
selected parent, uses a new optimizer, and includes eligible old training examples.
Checkpoint selection includes epoch zero for both parent and updated models.
Validation loss selects parent epochs 30/23/24 and update epochs 11/17/10.

The VK and Tools certificate experiments use fixed parent/update model pairs. The VK training universe contains 65,659
IDs, including future IDs used as negatives in dense training; decoding
permits 62,398 current items. Its released validation data reuse
gap-training interactions, so the epoch-28 endpoint is fixed independently
of those data. Certificate metrics use explicit query means and capped
labels. Multi-target old/new slices can overlap. Each Tools parent
supplies frozen- and refreshed-candidate updates; both are evaluated
against that parent's certificate on the same Primary cohort.

\subsubsection{Decoder configuration}

The VK GPT and Tools T5 certificate experiments use four generated
SID tokens, width 20, and return 20 paths. These are distinct from the
width-40 T5/LC-Rec recommendation comparisons. The GPT codebook
blocks have 512 tokens per position with zero bias; the T5 blocks have 256 tokens with bias one.
Log-softmax is computed over the full model vocabulary before the sole catalog mask.
No allowed-child renormalization, sampling, repetition penalty, or variable-length valid EOS
completion is introduced. The effective generation settings have length penalty one.

The first beam has score zero and other initial copies have score $-10^9$. Valid top-20
expansions are checked to remain above inactive sentinels. At layers 1--3, finite selected
expansions do not finish, and the inactive completed pool cannot change the live beam or
trigger stopping. A catalog with few legal backup extensions may fill some top-40
positions with invalid/EOS sentinels. These cannot enter the finite top 20 or its effective
completed pool. All paths terminate at layer four and share the same length normalization.
The subsequently unused running state can differ because of sentinel rounding ties; the
proof concerns effective pre-score states and final outputs, not those dead temporaries.

The T5 evaluation additionally tracks the encoder output and the self-/cross-attention caches.
Every T5 generation call performs one encoder and four decoder forwards. Both models run in
evaluation mode without batch normalization. Fixed-shape query separability and deterministic
scoring are premises; matching repeated cache fingerprints is finite implementation evidence.

\subsubsection{Numerical validation}

The implementation is checked separately from the logical criterion.
An independent exhaustive decoder check covers 12,240 binary three-level cases;
its 2,975 certificates agree across all 41,552 containing catalogs.
Reference/full-codebook comparisons, repeated executions and independently
reconstructed layer decisions check the saved VK and Tools trajectories.

Certificates use scores and effective states from the actual cached decoding
pass. A separate full-forward evaluation can differ beyond the fixed
$2\times10^{-5}$ tolerance, so it is not substituted for the cached score
function.

\subsubsection{Metric and cohort definitions}

VK histories take the last 20 items from time parts $[0,9)$; targets are in $[9,10)$.
Old items occur in $[0,8)$, admitted new items in $[0,9)$ but not the old catalog. The
test set also has targets not yet admitted; these remain in whole-catalog evaluation with zero
contribution. Recall counts unique predicted hits against the label set and divides by the
count of label occurrences capped at 20. New-target cohorts retain only their new-label
component for that metric. Multi-target queries can belong to more than one target-kind cohort.
Tools uses the local 0.7-to-0.8 catalog transition and one target per query.

The primary bound comparison uses every pure-old query with an admitted new target. Parent-certified
groups are secondary subsets defined from the parent before examining updated-model results.
Tools' three parent certificates agree across the frozen and refreshed mappings, as implied by
reference independence for those inputs. They certify 4,274/5,091, 4,459/5,091, and 3,951/5,091
pure-old queries, respectively. None of the six updates exceeds the bound on the complete Primary cohort.

Within each certificate comparison, parent and updated checkpoints use
identical numerical settings.

\subsection{Cost of invariance checking}
\label{app:streaming-cost}
Streaming evaluation reproduces every ordered prediction, layer decision
and margin for 39,219 VK and 5,360 Tools queries without an additional
model forward pass. The first and last batches of each architecture receive
one A/B warmup and a fixed ABBAABBA timing sequence: 32 timed calls,
four repeats per mode and shape. Per-call wall time includes score checks
and synchronization, but excludes input transfer, output copying and file
writing. Both modes retain the same old-prefix index; peak-memory
differences exclude its shared resident cost. Index storage is
1,074,560 bytes on VK and 225,688 bytes on Tools. Table~\ref{tab:streaming-cost}
reports the four shapes; branch enumeration incurs separate terminal
continuations and does not share this same-pass cost guarantee.
\begin{table}[!htbp]
\caption{Same-pass checker cost on two NVIDIA RTX 3090 GPUs. A is uninstrumented decoding and B adds the checker. Each fixed shape has four timed A/B calls after one warmup each; mean wall time is shown. Both modes retain the same old-prefix index. $\Delta$peak is the difference in peak allocated memory in that environment, not total temporary or index memory. Negative timing change is not evidence of speedup.}
\label{tab:streaming-cost}
\centering\small
\setlength{\tabcolsep}{3.3pt}
\begin{tabular*}{\linewidth}{@{\extracolsep{\fill}}lrrrrr@{}}
\toprule
Architecture & Batch & A (ms) & B (ms) & Change & $\Delta$peak (KiB) \\
\midrule
GPT / VK & 256 & 1336.52 & 1337.75 & +0.09\% & 20 \\
GPT / VK & 51 & 271.23 & 271.97 & +0.27\% & 4 \\
T5 / Tools & 128 & 1878.40 & 1882.54 & +0.22\% & 10 \\
T5 / Tools & 112 & 1649.98 & 1643.02 & -0.42\% & 9 \\
\bottomrule
\end{tabular*}
\end{table}

\section{Extended T5 Results and Component Analysis}
\label{app:extended-t5}

This appendix evaluates the final T5 models for seeds 17, 42 and 2027.
Parameter and budget sweeps center on the configurations in
Table~\ref{tab:current-scoring}; matched component controls inherit each
full-method run's adaptation and scoring settings.

\subsection{Sensitivity to scoring parameters and evaluation budget}
\label{app:final-sensitivity}

We hold the nine trained models, assignments and collaborative predictors
fixed and vary scoring parameters around the configurations in
Table~\ref{tab:current-scoring}. For each dataset
and seed, we multiply one of $\lambda$, $\gamma$ or $b$ by
$\{0.5,0.75,1,1.25,1.5\}$ while holding the other two fixed.
Beam width is 40, the additional-item budget is 80, and the completion batch
size is 20. Evaluation uses the same 1,040 Beauty, 3,007 Tools and 774 Toys
validation queries as parameter selection. Figure~\ref{fig:final-scoring-sensitivity}
shows paired NDCG@10 changes from each seed's adopted configuration;
all four ranking metrics are retained in the source data.

The response differs across parameters and datasets. Halving $\lambda$
changes mean NDCG@10 by $-0.129$, $-0.145$ and $-0.013$ percentage points
on Beauty, Tools and Toys. On Beauty, halving $b$ increases it by 0.100
points, whereas multiplying $b$ by 1.5 decreases it by 0.166 points.
Across the tested one-parameter perturbations, mean changes range from
$-0.166$ to $+0.100$ on Beauty, $-0.145$ to $+0.029$ on Tools, and
$-0.055$ to $+0.035$ on Toys. Individual seeds can respond in opposite
directions, as seen for Toys at half the fusion weight.

\begin{figure}[!htbp]
\centering
\includegraphics[width=\linewidth]{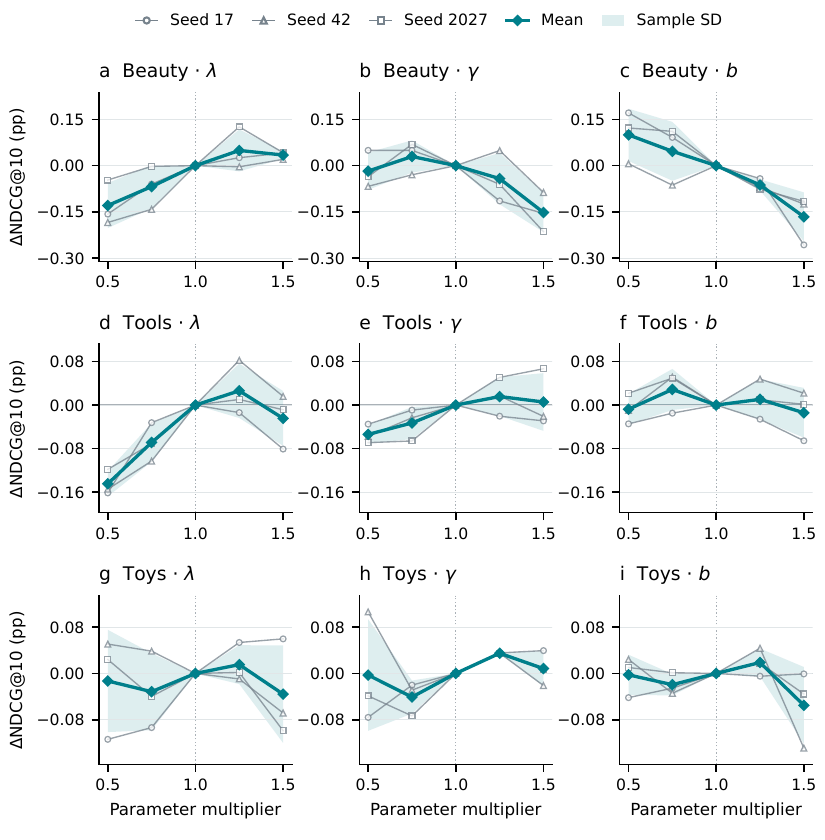}
\caption{Scoring sensitivity around the nine final T5 configurations.
Rows show Beauty, Tools and Toys; columns vary $\lambda$, $\gamma$ and $b$
one at a time. The horizontal coordinate multiplies each seed's own
selected parameter. Vertical coordinates are paired validation NDCG@10
changes in percentage points from multiplier one. Gray paths show seeds
17/42/2027; teal diamonds and bands show their mean and sample standard
deviation. Straight segments join the five measured multipliers. Other
parameters, models, mappings and the 80-item completion budget remain fixed.}
\label{fig:final-scoring-sensitivity}
\end{figure}

\begin{figure}[!tp]
\centering
\includegraphics[width=\linewidth]{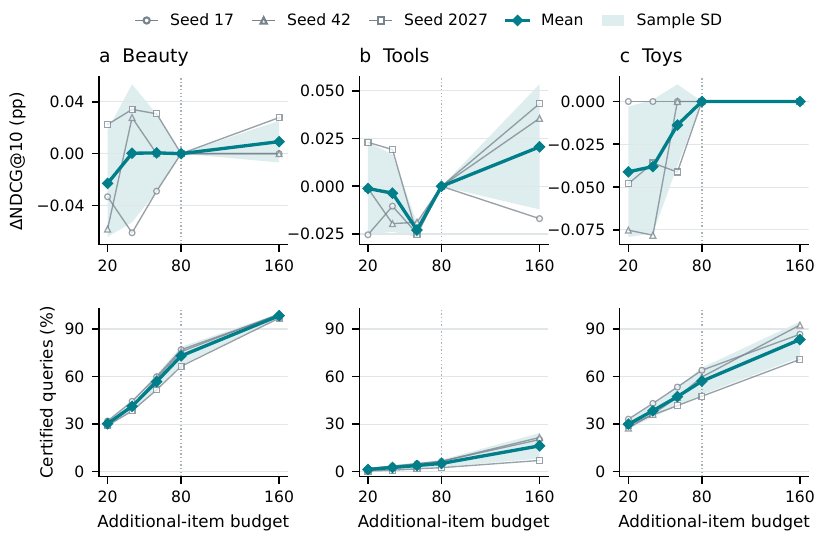}
\caption{Validation sensitivity to the additional-item budget at fixed
final scoring parameters. Top: paired NDCG@10 change from budget 80.
Bottom: the fraction satisfying the numerical Top-20 certificate.
Gray paths show seeds 17/42/2027; teal diamonds and bands give means and
sample standard deviations. Points correspond to budgets 20/40/60/80/160;
straight segments connect observations. The budget caps additional item
evaluations; actual counts and all four ranking metrics appear in
Table~\ref{tab:final-budget-sensitivity}.}
\label{fig:final-budget-sensitivity}
\end{figure}
\begin{table}[!tp]
\centering\small
\setlength{\tabcolsep}{3.5pt}
\caption{Validation ranking and work as the additional-item budget varies.
Means over seeds 17/42/2027 with fixed final models, mappings and scoring
parameters. R/N denote Recall/NDCG (\%); Added counts actual additional
likelihood evaluations, and Cert. is numerical Top-20 certification (\%).
Per-seed variation is shown in Figure~\ref{fig:final-budget-sensitivity}.}
\label{tab:final-budget-sensitivity}
\begin{tabular}{llrrrrrr}
\toprule
Dataset & Budget & R@10 & N@10 & R@20 & N@20 & Added & Cert. \\
\midrule
Beauty & 20 & 5.641 & 3.071 & 8.910 & 3.901 & 19.79 & 30.22 \\
 & 40 & 5.705 & 3.095 & 9.103 & 3.956 & 33.74 & 41.19 \\
 & 60 & 5.705 & 3.095 & 9.167 & 3.972 & 45.51 & 56.76 \\
 & 80 & 5.705 & 3.094 & 9.167 & 3.971 & 54.15 & 73.01 \\
 & 160 & 5.737 & 3.104 & 9.199 & 3.981 & 64.71 & 98.40 \\
\addlinespace[2pt]
Tools & 20 & 5.421 & 3.082 & 7.205 & 3.530 & 19.99 & 1.33 \\
 & 40 & 5.399 & 3.079 & 7.250 & 3.547 & 39.72 & 2.55 \\
 & 60 & 5.332 & 3.060 & 7.250 & 3.542 & 59.21 & 3.81 \\
 & 80 & 5.410 & 3.083 & 7.383 & 3.576 & 78.45 & 5.14 \\
 & 160 & 5.476 & 3.104 & 7.394 & 3.584 & 152.46 & 16.24 \\
\addlinespace[2pt]
Toys & 20 & 4.651 & 2.376 & 7.278 & 3.037 & 18.11 & 29.89 \\
 & 40 & 4.651 & 2.379 & 7.364 & 3.063 & 32.13 & 38.37 \\
 & 60 & 4.737 & 2.404 & 7.450 & 3.084 & 44.45 & 47.29 \\
 & 80 & 4.780 & 2.417 & 7.407 & 3.076 & 55.00 & 57.06 \\
 & 160 & 4.780 & 2.417 & 7.321 & 3.058 & 79.84 & 83.33 \\
\bottomrule
\end{tabular}
\end{table}

% Keep the budget discussion together around the paired figure/table page.
\needspace{15\baselineskip}
\paragraph{Budget and certification.}
At the adopted scoring parameters, we vary only the maximum number of
additional likelihood evaluations over $\{20,40,60,80,160\}$
(Figure~\ref{fig:final-budget-sensitivity},
Table~\ref{tab:final-budget-sensitivity}). Raising the budget from 80 to
160 increases mean numerical Top-20 certification from
73.01/5.14/57.06\% to 98.40/16.24/83.33\% on Beauty/Tools/Toys.
Mean additional evaluations rise from 54.15/78.45/55.00 to
64.71/152.46/79.84. The corresponding NDCG@10 changes are
$+0.009$, $+0.021$ and zero percentage points. Toys Recall@20/NDCG@20
decrease from 7.407/3.076\% to 7.321/3.058\%.
Thus the larger budget
certifies more rankings, while its recommendation effect remains
dataset- and seed-dependent. All models and the main-result configurations
remain fixed throughout these measurements.
\FloatBarrier

\subsection{Calibration tradeoffs and candidate effort}
Figure~\ref{fig:calibration-test-tradeoffs} shows the
old/new-item tradeoffs between global correction and calibrated
scoring across all datasets and three retriever seeds.
Figure~\ref{fig:completion-efficiency} reports their candidate effort and
measured GPU inference costs.

\begin{figure}[!htbp]
\centering
\includegraphics[width=\linewidth]{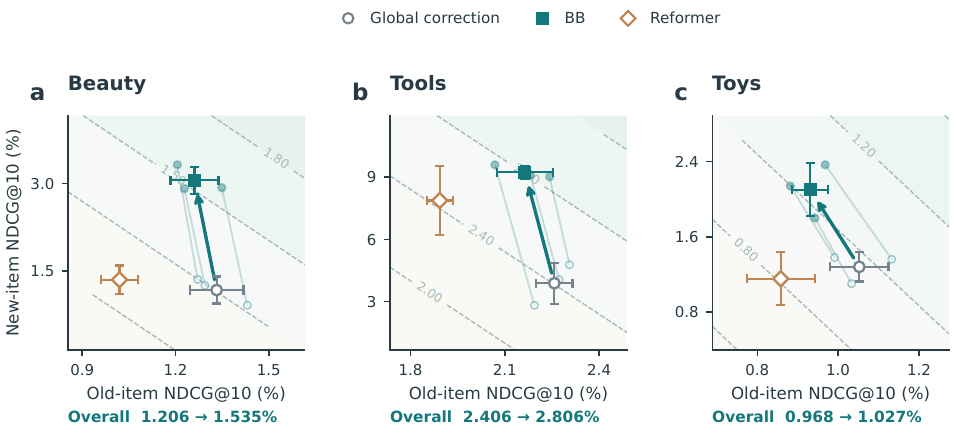}
\caption{Calibration redistributes old- and new-item ranking quality under
the main-result configurations with seeds 17/42/2027.
Thin paths connect global correction to BB for each of three retriever seeds;
large markers and error bars show means and sample standard deviations.
Dashed contours give overall test NDCG@10 (\%) using the complete query
population. Arrows connect matched global/calibrated settings; all methods
use seeds 17/42/2027.}
\label{fig:calibration-test-tradeoffs}
\end{figure}
\begin{figure}[!htbp]
\centering
\includegraphics[width=\linewidth]{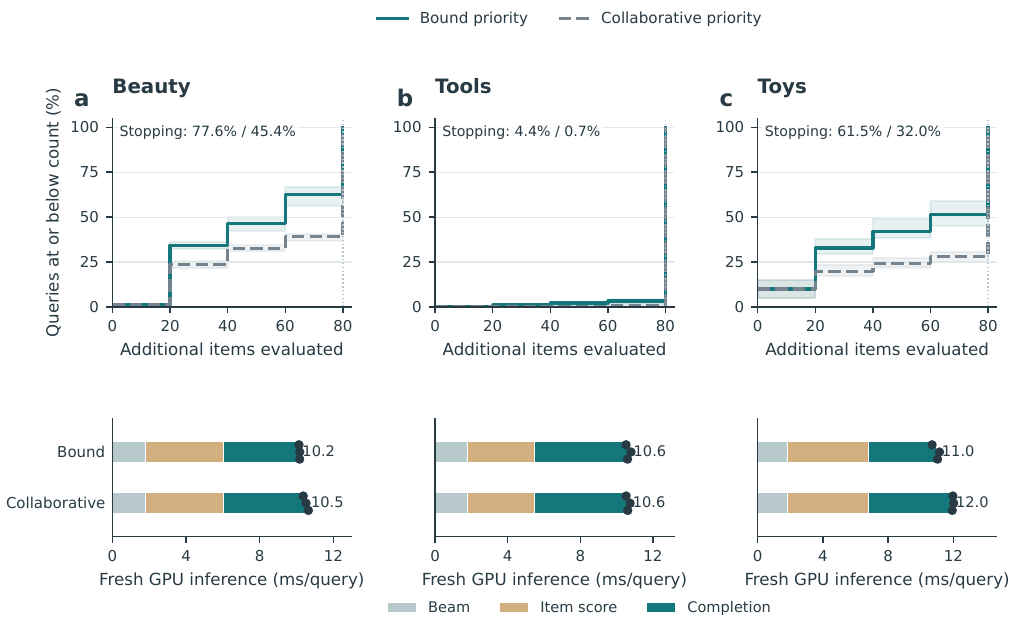}
\caption{Candidate effort and measured inference cost for seeds
17/42/2027 under the same calibrated
score and 80-item maximum. Top: cumulative evaluated-item counts over complete
test sets, with seed means and three-seed ranges. Both curves reach 100\% at
the budget; this includes budget exhaustion. Annotations give the separate
numerical Top-20 certification rates in bound/collaborative order. Bottom: fresh GPU
inference wall time on the same 128 histories per domain and seed, decomposed
into beam, item scoring and completion costs. Bars show seed means and dots
show per-seed totals after averaging three repetitions. Neural inference
uses an RTX 3090; item correction and bookkeeping use two CPU threads.}
\label{fig:completion-efficiency}
\end{figure}
\begin{table}[!htbp]
\centering\footnotesize
\setlength{\tabcolsep}{3pt}
\caption{Candidate accounting and measured GPU inference cost, averaged across seeds 17/42/2027. Initial/added leaves and numerical Top-20 certification rates use complete test sets. Wall time per query uses the same first 128 histories per domain and seed, three repetitions, RTX 3090 neural inference and two CPU threads; it includes fresh beam, item scoring, encoder and completion. Timing excludes loading and warmup (Appendix~\ref{app:completion}).}
\label{tab:completion-cost}
\begin{tabular}{llrrrr}
\toprule
Domain & Priority & Initial & Added & Certified (\%) & ms/query \\
\midrule
Beauty & BB (Ours) & 126.37 & 51.08 & 77.64 & 10.16 \\
Beauty & Collaborative priority & 126.37 & 60.69 & 45.40 & 10.50 \\
Tools & BB (Ours) & 40.17 & 78.70 & 4.45 & 10.64 \\
Tools & Collaborative priority & 40.17 & 79.77 & 0.71 & 10.63 \\
Toys & BB (Ours) & 96.08 & 52.75 & 61.53 & 10.96 \\
Toys & Collaborative priority & 96.08 & 63.65 & 31.95 & 11.98 \\
\bottomrule
\end{tabular}
\end{table}

\paragraph{Runtime and candidate counts.}
On Beauty, mean total wall time is 10.16 versus 10.50 ms/query
for bound and collaborative priority; the completion stage alone is
4.10 versus 4.45 ms/query.
On Toys, total time is 10.96 versus 11.98 ms/query, while Tools is
nearly unchanged (10.64 versus 10.63 ms/query).
Total latency also reflects beam decoding, item scoring, batch utilization and
host-side work. The reported times are batched means on the fixed timing subsets.

\paragraph{Matched component controls.}
Tables~\ref{tab:completion-all-at10}, \ref{tab:completion-extended}
and~\ref{tab:completion-subgroups}
use seeds 17/42/2027 and the corresponding settings in
Table~\ref{tab:current-scoring}. These matched controls share the BB runs in the main comparison,
as do the GPU timings above.
\begin{table}[!htbp]
\centering\scriptsize
\setlength{\tabcolsep}{3pt}
\caption{Matched T5 component controls at cutoff 20: Recall and NDCG (\%),
mean $\pm$ sample s.d. over seeds 17/42/2027. Each removal inherits its
corresponding full-method settings; BB uses the runs in Table~\ref{tab:trained-comparison}.
Initial-pool reranking uses $E_0$ without additional evaluations,
whereas without completion uses only the returned beam.}
\label{tab:completion-extended}
\begin{tabular}{lrrrrrr}
\toprule
& \multicolumn{2}{c}{Beauty} & \multicolumn{2}{c}{Tools} & \multicolumn{2}{c}{Toys} \\
Variant & R@20 & N@20 & R@20 & N@20 & R@20 & N@20 \\
\midrule
Generator only & $4.045\!\pm\!0.089$ & $1.547\!\pm\!0.073$ & $4.932\!\pm\!0.163$ & $2.236\!\pm\!0.123$ & $2.979\!\pm\!0.148$ & $1.124\!\pm\!0.086$ \\
Without collaborative score & $4.040\!\pm\!0.082$ & $1.546\!\pm\!0.073$ & $4.956\!\pm\!0.173$ & $2.242\!\pm\!0.124$ & $2.979\!\pm\!0.148$ & $1.123\!\pm\!0.085$ \\
Without completion & $4.663\!\pm\!0.250$ & $1.858\!\pm\!0.091$ & $5.585\!\pm\!0.235$ & $2.883\!\pm\!0.050$ & $3.240\!\pm\!0.115$ & $1.282\!\pm\!0.060$ \\
Initial-pool reranking & $4.663\!\pm\!0.263$ & $1.867\!\pm\!0.100$ & $5.585\!\pm\!0.235$ & $2.883\!\pm\!0.050$ & $3.383\!\pm\!0.042$ & $1.329\!\pm\!0.043$ \\
Without construction & $5.135\!\pm\!0.322$ & $1.999\!\pm\!0.116$ & $6.430\!\pm\!0.121$ & $3.180\!\pm\!0.010$ & $3.394\!\pm\!0.101$ & $1.348\!\pm\!0.046$ \\
Item correction only & $4.372\!\pm\!0.034$ & $1.642\!\pm\!0.025$ & $4.558\!\pm\!0.814$ & $2.130\!\pm\!0.410$ & $2.043\!\pm\!0.028$ & $0.831\!\pm\!0.007$ \\
Collaborative priority & $5.153\!\pm\!0.157$ & $2.022\!\pm\!0.072$ & $6.437\!\pm\!0.049$ & $3.168\!\pm\!0.039$ & $3.426\!\pm\!0.051$ & $1.349\!\pm\!0.036$ \\
Without calibration & $4.158\!\pm\!0.180$ & $1.606\!\pm\!0.061$ & $5.808\!\pm\!0.138$ & $2.752\!\pm\!0.109$ & $3.346\!\pm\!0.064$ & $1.296\!\pm\!0.060$ \\
BB (Ours) & $5.208\!\pm\!0.189$ & $2.035\!\pm\!0.079$ & $6.437\!\pm\!0.112$ & $3.179\!\pm\!0.009$ & $3.458\!\pm\!0.109$ & $1.352\!\pm\!0.030$ \\
\bottomrule
\end{tabular}
\end{table}

\begin{table}[!htbp]
\centering\footnotesize
\setlength{\tabcolsep}{3pt}
\caption{Subgroup NDCG@10 (\%), averaged over seeds 17/42/2027 using the full-method and matched-control settings in Table~\ref{tab:current-scoring}. New denotes current-catalog new targets; Primary additionally requires old-only histories. Current combines old and new targets. Future-catalog targets remain misses in the main-table denominator.}
\label{tab:completion-subgroups}
\begin{tabular}{llrrrr}
\toprule
Domain & Variant & Old & New & Current & Primary \\
\midrule
Beauty & Generator only & 1.271 & 0.943 & 1.200 & 0.674 \\
Beauty & Without construction & 1.269 & 2.940 & 1.633 & 2.010 \\
Beauty & Item correction only & 0.458 & 4.063 & 1.243 & 3.463 \\
Beauty & Without calibration & 1.332 & 1.171 & 1.297 & 0.383 \\
Beauty & BB (Ours) & 1.261 & 3.052 & 1.651 & 2.290 \\
Tools & Generator only & 1.598 & 4.367 & 1.850 & 4.457 \\
Tools & Without construction & 2.164 & 9.277 & 2.810 & 9.372 \\
Tools & Item correction only & 1.113 & 8.566 & 1.790 & 8.776 \\
Tools & Without calibration & 2.258 & 3.888 & 2.406 & 4.221 \\
Tools & BB (Ours) & 2.164 & 9.234 & 2.806 & 9.355 \\
Toys & Generator only & 0.848 & 1.089 & 0.896 & 1.184 \\
Toys & Without construction & 0.971 & 1.919 & 1.159 & 1.596 \\
Toys & Item correction only & 0.080 & 3.321 & 0.725 & 3.279 \\
Toys & Without calibration & 1.052 & 1.281 & 1.098 & 1.092 \\
Toys & BB (Ours) & 0.931 & 2.103 & 1.164 & 1.907 \\
\bottomrule
\end{tabular}
\end{table}

\FloatBarrier
\subsection{New-target ranking, certification and query effects of construction}
\label{app:final-construction-effects}

This comparison evaluates the selected shared assignment and its subsequent
adaptation against the paired static branch. Both use the adaptation setting,
selected checkpoints and scoring parameters in
Table~\ref{tab:current-scoring}. Queries are paired by history and target;
the Beauty/Tools/Toys denominators are 7,335/5,360/6,266, including
unavailable future targets.

\paragraph{New-target ranking.}
On Beauty and Toys, new-target NDCG@10 increases by 3.82\% and 9.56\%,
respectively (Table~\ref{tab:completion-subgroups}). The gains are larger
for Primary queries, whose new targets follow entirely old-item histories:
2.010\% to 2.290\% on Beauty and 1.596\% to 1.907\% on Toys, or
13.95\% and 19.42\% relative improvements.
These queries account for 734/7,335 and 571/6,266 of the full test
populations. Old-target NDCG@10 decreases by 0.60\% and 4.14\%, leaving
overall improvements of 1.13\% and 0.37\%.
For Beauty, the new- and old-target contributions to the overall NDCG@10
difference are $+0.02273$ and $-0.00557$ percentage points; for Toys they
are $+0.03216$ and $-0.02838$. These contributions use the full-query
denominator and explain how the group effects combine.

\paragraph{Certification and evaluation effort.}
On Beauty, the constructed/adapted branch increases numerical Top-20
certification from 74.10\% to 77.64\%, while reducing additional
evaluations from 52.59 to 51.08 per query (2.86\%). On Toys,
certification increases from 51.66\% to 61.53\%, with evaluations falling
from 57.28 to 52.75 (7.90\%). Tools averages 78.70 additional evaluations
in both branches, with certification rates of 4.43\% and 4.45\%.
These full-test measurements use the same 80-item maximum and numerical
stopping rule. Figure~\ref{fig:main-ablation-profiles}b shows the Beauty
curves and their paired certification differences.

\paragraph{Training support and selected assignments.}
Tools supplies 133 new-target training contexts covering 74 of 158 new
items, compared with 2,993 covering 532 of 595 on Beauty and 4,950
covering 944 of 1,025 on Toys (Appendix~\ref{app:protocol}).
Its selected assignment remains unchanged in seeds 17 and 2027 and
changes two identities in seed 42. In the first two seeds, constructed
and static predictions coincide; seed 42 loses one Top-10 hit. This combination of sparse
new-target supervision and little assignment change distinguishes Tools
from the two datasets with larger cohort and certification gains.

\paragraph{Paired query gains and losses.}
Table~\ref{tab:final-construction-query-effects} resolves the aggregate
effects by seed. Beauty seed 42 recovers 43 Top-10 hits and loses 34,
and its NDCG@10 gain outweighs the small negative overall differences
in seeds 17 and 2027. Seed 17 recovers 34 hits and loses 42, with
positive and negative NDCG@10 contributions of $+0.3591$ and $-0.3643$
percentage points. Toys seed 17 recovers 51 hits and loses 46; its
positive NDCG change exceeds the negative changes in seeds 42 and 2027.
Seed 42 has two recovered and two lost hits, with a negative NDCG change.
Primary effects also vary: the three paired NDCG@10 differences are
$+0.24070$, $+0.61568$ and $-0.01528$ percentage points on Beauty,
and $+0.83901$, $0$ and $+0.09132$ on Toys.

The selected maps directly recode a history or target in 61/32/88 Beauty
queries and 24/5/29 Toys queries. Ranking changes extend beyond these
sets: Beauty seed 17 recovers 33 and loses 40 hits among queries with
unchanged history and target codes, whose competing candidate assignments
and adapted scores can still differ.

\begin{table}[!htbp]
\centering\small
\setlength{\tabcolsep}{3.5pt}
\caption{Paired query effects of final construction and adaptation.
Recovered/lost count Top-10 misses becoming hits and hits becoming misses,
relative to without construction. $\Delta N^+$ and $\Delta N^-$ sum positive
and negative query NDCG@10 changes and divide by all test queries, in
percentage points; their sum is Net. Recoded counts queries whose history
or target contains an item with a changed code.}
\label{tab:final-construction-query-effects}
\begin{tabular}{llrrrrrr}
\toprule
Dataset & Seed & Recoded & Recovered & Lost & $\Delta N^+$ & $\Delta N^-$ & Net \\
\midrule
Beauty & 17 & 61 & 34 & 42 & +0.3591 & -0.3643 & -0.0052 \\
 & 42 & 32 & 43 & 34 & +0.3398 & -0.2817 & +0.0581 \\
 & 2027 & 88 & 4 & 8 & +0.0447 & -0.0461 & -0.0014 \\
\addlinespace[2pt]
Tools & 17 & 0 & 0 & 0 & 0.0000 & 0.0000 & 0.0000 \\
 & 42 & 33 & 0 & 1 & +0.0113 & -0.0232 & -0.0118 \\
 & 2027 & 0 & 0 & 0 & 0.0000 & 0.0000 & 0.0000 \\
\addlinespace[2pt]
Toys & 17 & 24 & 51 & 46 & +0.4449 & -0.4142 & +0.0307 \\
 & 42 & 5 & 2 & 2 & +0.0133 & -0.0319 & -0.0186 \\
 & 2027 & 29 & 27 & 21 & +0.2142 & -0.2149 & -0.0007 \\
\bottomrule
\end{tabular}
\end{table}

\subsection{Query recovery under the combined score}
\label{app:access-decomposition}

\paragraph{Reuse and additional scoring.}
Initial-pool reranking orders all fully scored leaves $E_0$ using the
same model, map and combined score as BB, with no additional likelihood
evaluations. Tables~\ref{tab:completion-all-at10}
and~\ref{tab:completion-extended} place this control between returned-beam
reranking and full completion. Table~\ref{tab:completion-access-decomposition}
reports the two paired increments. Reuse raises mean NDCG@10 by
0.009/0.000/0.028 percentage points on Beauty/Tools/Toys; additional
scoring contributes 0.119/0.249/0.023 percentage points.
Tools has identical four-metric results under returned-beam and initial-pool
reranking for each seed. Its mean initial pool contains 40.17 items, close
to the 40 returned candidates; Beauty and Toys contain 126.37 and 96.08.
The contribution of reuse is largest on Toys, where both stages improve
the dataset means.

Full completion exceeds initial-pool reranking on all four dataset-mean
metrics and on Recall@10 and NDCG@10 in all nine runs.
The paired NDCG@10 increases on Toys seeds 17 and 2027 are small,
at 0.006913 and 0.003187 percentage points. At cutoff 20, Toys seed 42
has four fewer hits (213 to 209; $-0.063837$ percentage points), while
its NDCG@20 increases by 0.004283 percentage points. These paired
contrasts capture both recovered and displaced hits under the fixed score.

\begin{table}[!htbp]
\centering\small
\setlength{\tabcolsep}{3pt}
\caption{Paired decomposition of candidate-access gains, in percentage points:
mean $\pm$ sample s.d. across seeds 17/42/2027.
Reuse is initial-pool reranking minus returned-beam reranking;
additional scoring is full completion minus initial-pool reranking.
The model, map, combined score and queries are fixed within each seed.}
\label{tab:completion-access-decomposition}
\begin{tabular}{llrrrr}
\toprule
Dataset & Contrast & $\Delta$R@10 & $\Delta$N@10 & $\Delta$R@20 & $\Delta$N@20 \\
\midrule
Beauty & Reuse & $0.009\!\pm\!0.021$ & $0.009\!\pm\!0.012$ & $0.000\!\pm\!0.047$ & $0.009\!\pm\!0.010$ \\
Beauty & Additional scoring & $0.336\!\pm\!0.119$ & $0.119\!\pm\!0.050$ & $0.545\!\pm\!0.076$ & $0.169\!\pm\!0.027$ \\
\addlinespace[2pt]
Tools & Reuse & $0.000\!\pm\!0.000$ & $0.000\!\pm\!0.000$ & $0.000\!\pm\!0.000$ & $0.000\!\pm\!0.000$ \\
Tools & Additional scoring & $0.665\!\pm\!0.088$ & $0.249\!\pm\!0.037$ & $0.852\!\pm\!0.191$ & $0.296\!\pm\!0.045$ \\
\addlinespace[2pt]
Toys & Reuse & $0.064\!\pm\!0.058$ & $0.028\!\pm\!0.018$ & $0.144\!\pm\!0.073$ & $0.047\!\pm\!0.019$ \\
Toys & Additional scoring & $0.074\!\pm\!0.060$ & $0.023\!\pm\!0.032$ & $0.074\!\pm\!0.121$ & $0.022\!\pm\!0.016$ \\
\bottomrule
\end{tabular}
\end{table}

\paragraph{Candidate access across cutoffs.}
Figure~\ref{fig:candidate-recovery} compares the generator, returned-beam
reranking, initial-pool reranking and full completion at every cutoff from
1 to 20. Its lower panels separate hit turnover from initial-pool reuse
and additional scoring. At cutoff 10, reuse adds a net
0.09/0.00/0.64 hits per 1,000 queries on Beauty/Tools/Toys; additional
scoring contributes 3.36/6.65/0.74. Each increment includes recovered and
displaced hits under the same model, mapping and score.

\begin{figure}[!htbp]
\centering
\includegraphics[width=\linewidth]{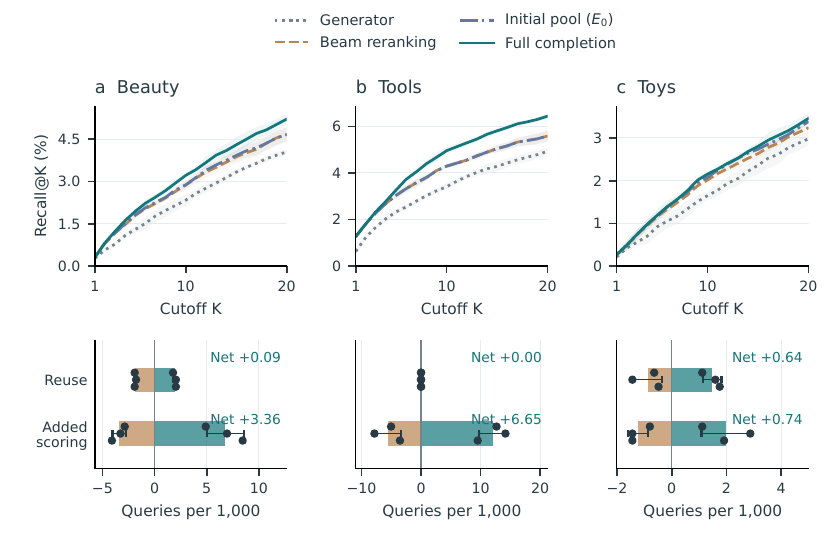}
\caption{Candidate access under the final T5 configurations, seeds
17/42/2027. Top: Recall at all integer cutoffs 1--20; lines and bands show
means and sample standard deviations. Bottom: Top-10 turnover for
returned beam $\to E_0$ (reuse) and $E_0\to$ full completion (added scoring).
Teal bars count recovered hits and ochre bars displaced hits per 1,000
queries. Bars and error bars give means and sample standard deviations;
dots show each seed and annotations give the mean net change.}
\label{fig:candidate-recovery}
\end{figure}

\paragraph{A fixed-score recovery case.}
\label{app:completion-case}
Figure~\ref{fig:completion-score-plane} shows full-range and zoomed views
of one query under the same score. For Beauty seed~17,
we select among the 16 new-target queries whose targets are actually added
by completion and change from a Top-10 miss under beam reranking to a hit.
We choose the final target rank closest to this subset's median (rank~7),
breaking ties by test-query order; the chosen target finishes at rank~8. The case illustrates a successful
recovery; the complete-population gains and losses in
Figure~\ref{fig:candidate-recovery} characterize aggregate effects.

The initial search fully evaluates 82 items but returns only 40 beam
candidates. Completion adds 20 full-likelihood evaluations, yielding 102
scored items. The target is fifth in priority order within that added batch
and finishes at rank~8; 5 added items enter the final Top-10.

The initial-beam control ranks only the returned 40 items, whereas the full
method also reuses scored terminal expansions outside the beam.
Both panels display the same final score coordinates. In particular, scores
obtained during completion are plotted retrospectively in the left panel;
they were not available before those evaluations. No model update or
identifier change occurs between the two panels.

\begin{figure}[!htbp]
\centering
\textbf{Full score range}\par\smallskip
\includegraphics[width=\linewidth]{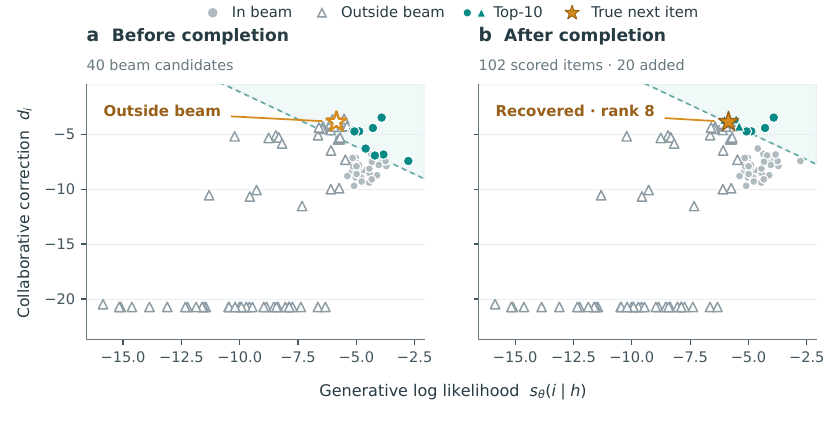}
\par\medskip
\textbf{Detail near the Top-10 cutoffs}\par\smallskip
\includegraphics[width=\linewidth]{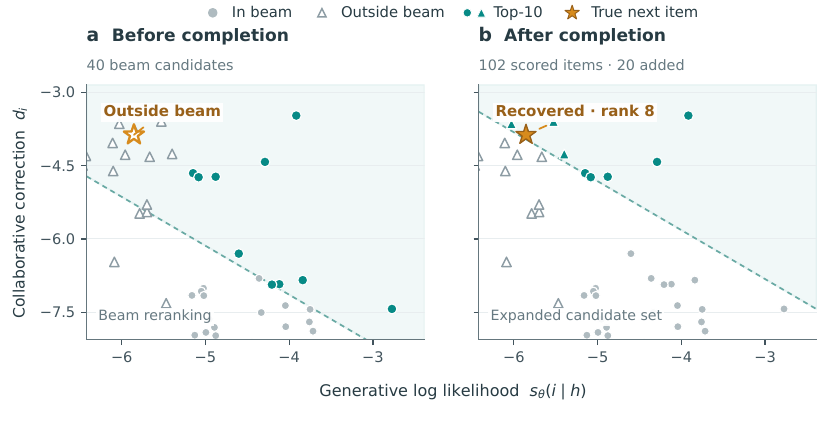}
\caption{Candidate completion for a Beauty seed-17 query under the fixed
score $F_i=s_\theta(i\mid h)+d_i$. Each row compares (a) beam reranking
and (b) ranking after completion. The upper row shows all 102 scored items;
the lower row enlarges the 41 items near the decision boundaries.
Circles/triangles denote items inside/outside the returned beam; teal marks
returned Top-10 items and the gold star marks the target. Dashed lines
mark stage-specific Top-10 cutoffs. Initial search scores 82 items and
returns 40; completion scores 20 more and places the target at rank 8.
Coordinates stay fixed within each row; later-computed scores are shown
retrospectively in (a). The case-selection rule is stated above.}
\label{fig:completion-score-plane}
\label{fig:completion-score-plane-full}
\end{figure}

\needspace{7\baselineskip}
\subsection{Where the Tools ranking difference arises}
\label{app:tools-ranking}

The Tools test set contains 4,873 old-target and 487 new-target queries.
BB achieves Recall@10 of 4.956\%, compared with
Reformer's 4.210\%, and NDCG@10 of 2.806\%, compared with
2.436\%. Old targets contribute $+0.24652$ percentage points to
BB minus Reformer NDCG@10 and new targets contribute $+0.12336$ points.
Their sum is the aggregate difference; the subgroup orderings are
distinct (Table~\ref{tab:continual-groups}).

\paragraph{Calibration and the generator.}
Generator-only new-target NDCG@10 is 4.367\%.
Applying the learned correction to every item yields
3.888\%; the calibrated method reaches
9.234\%, while old-target NDCG changes from
2.258\% to 2.164\%.
New-target Recall@10 changes from 6.982\%
to 13.689\%. This isolates a scoring effect
without changing the generator or its identifiers. Reformer's new-target
NDCG@10 is 7.876\%; the complete-system difference also
includes representation and adaptation choices.

\paragraph{Training support and concentration.}
New denotes addition since the old catalog, not absence of training
interactions. Of the 487 new-target test queries, 308 have target items
observed in update training and 179 have items absent from the available
training interactions. Neither the calibrated BB nor Reformer hits any
of the latter targets in Top-10. Reformer's
new-target gain is concentrated: one item, appearing as a training target
12 times and as a test target 40 times, contributes 75.7\% of its new-target
NDCG. The same item contributes 76.8\% of BB's new-target NDCG.
BB hits eight, seven and eight distinct new items across seeds 17/42/2027;
each Reformer seed hits four. These are query-weighted
recommendation results rather than uniform averages over item identities.

\subsection{Update strategies and conventional recommenders}
Tables~\ref{tab:continual-at20} and~\ref{tab:continual-groups} extend the
standard update and Reformer comparisons to cutoff 20 and target cohorts.
The conventional references use the same test histories and catalog;
their training and scoring protocols are in Appendix~\ref{app:default-baselines}.
With validation-tuned settings, BB has the highest column means in the
conventional comparison at cutoff 10.
\begin{table}[!htbp]
\centering\scriptsize
\setlength{\tabcolsep}{2pt}
\caption{Cutoff-20 results for standard update strategies, Reformer and BB (\%). Mean $\pm$ sample s.d. over seeds 17/42/2027 for all methods.}
\label{tab:continual-at20}
\begin{tabular}{lrrrrrr}
\toprule
& \multicolumn{2}{c}{Beauty} & \multicolumn{2}{c}{Tools} & \multicolumn{2}{c}{Toys} \\
Method & R@20 & N@20 & R@20 & N@20 & R@20 & N@20 \\
\midrule
Frozen/Frozen & $2.259\!\pm\!0.052$ & $0.937\!\pm\!0.043$ & $3.601\!\pm\!0.263$ & $1.674\!\pm\!0.078$ & $1.793\!\pm\!0.189$ & $0.720\!\pm\!0.071$ \\
FT/Frozen & $0.673\!\pm\!0.052$ & $0.239\!\pm\!0.014$ & $0.609\!\pm\!0.103$ & $0.233\!\pm\!0.044$ & $1.005\!\pm\!0.138$ & $0.454\!\pm\!0.049$ \\
Frozen/FT & $3.149\!\pm\!0.121$ & $1.247\!\pm\!0.035$ & $5.896\!\pm\!0.233$ & $2.728\!\pm\!0.130$ & $2.474\!\pm\!0.169$ & $0.995\!\pm\!0.077$ \\
FT/FT & $3.536\!\pm\!0.057$ & $1.363\!\pm\!0.034$ & $5.721\!\pm\!0.106$ & $2.577\!\pm\!0.077$ & $2.522\!\pm\!0.152$ & $0.981\!\pm\!0.027$ \\
FT/RT & $3.913\!\pm\!0.125$ & $1.489\!\pm\!0.048$ & $4.049\!\pm\!0.153$ & $1.796\!\pm\!0.100$ & $2.676\!\pm\!0.088$ & $1.016\!\pm\!0.046$ \\
Reformer~\citeyearpar{shi2025incremental} & $3.499\!\pm\!0.311$ & $1.376\!\pm\!0.115$ & $5.808\!\pm\!0.109$ & $2.837\!\pm\!0.142$ & $2.553\!\pm\!0.257$ & $1.044\!\pm\!0.123$ \\
\textbf{BB (Ours)} & $\mathbf{5.208\!\pm\!0.189}$ & $\mathbf{2.035\!\pm\!0.079}$ & $\mathbf{6.437\!\pm\!0.112}$ & $\mathbf{3.179\!\pm\!0.009}$ & $\mathbf{3.458\!\pm\!0.109}$ & $\mathbf{1.352\!\pm\!0.030}$ \\
\bottomrule
\end{tabular}
\end{table}

\begin{table}[!htbp]
\centering\scriptsize
\setlength{\tabcolsep}{3pt}
\caption{Old- and new-target results for the expanded update comparison (\%). Mean $\pm$ sample s.d. over seeds 17/42/2027 for all methods.}
\label{tab:continual-groups}
\begin{tabular}{llrrrr}
\toprule
Dataset & Method & Old R@10 & Old N@10 & New R@10 & New N@10 \\
\midrule
Beauty & Frozen/Frozen & $1.987\!\pm\!0.032$ & $1.009\!\pm\!0.052$ & $0.000\!\pm\!0.000$ & $0.000\!\pm\!0.000$ \\
 & FT/Frozen & $0.462\!\pm\!0.011$ & $0.211\!\pm\!0.012$ & $0.045\!\pm\!0.078$ & $0.014\!\pm\!0.024$ \\
 & Frozen/FT & $1.937\!\pm\!0.066$ & $1.000\!\pm\!0.008$ & $2.041\!\pm\!0.236$ & $0.924\!\pm\!0.135$ \\
 & FT/FT & $2.025\!\pm\!0.037$ & $1.005\!\pm\!0.043$ & $3.028\!\pm\!0.308$ & $1.315\!\pm\!0.178$ \\
 & FT/RT & $2.100\!\pm\!0.172$ & $1.032\!\pm\!0.053$ & $3.320\!\pm\!0.383$ & $1.520\!\pm\!0.166$ \\
 & Reformer & $1.994\!\pm\!0.103$ & $1.020\!\pm\!0.060$ & $3.028\!\pm\!0.583$ & $1.347\!\pm\!0.247$ \\
 & BB (Ours) & $2.618\!\pm\!0.076$ & $1.261\!\pm\!0.077$ & $6.483\!\pm\!0.371$ & $3.052\!\pm\!0.232$ \\
\midrule
Tools & Frozen/Frozen & $2.722\!\pm\!0.117$ & $1.529\!\pm\!0.061$ & $0.000\!\pm\!0.000$ & $0.000\!\pm\!0.000$ \\
 & FT/Frozen & $0.349\!\pm\!0.036$ & $0.175\!\pm\!0.017$ & $0.205\!\pm\!0.356$ & $0.067\!\pm\!0.115$ \\
 & Frozen/FT & $3.879\!\pm\!0.160$ & $2.067\!\pm\!0.069$ & $8.419\!\pm\!0.895$ & $4.902\!\pm\!0.961$ \\
 & FT/FT & $3.147\!\pm\!0.140$ & $1.642\!\pm\!0.086$ & $13.415\!\pm\!0.314$ & $7.375\!\pm\!0.584$ \\
 & FT/RT & $3.003\!\pm\!0.357$ & $1.614\!\pm\!0.139$ & $0.000\!\pm\!0.000$ & $0.000\!\pm\!0.000$ \\
 & Reformer & $3.441\!\pm\!0.031$ & $1.893\!\pm\!0.042$ & $11.910\!\pm\!1.550$ & $7.876\!\pm\!1.668$ \\
 & BB (Ours) & $4.084\!\pm\!0.142$ & $2.164\!\pm\!0.088$ & $13.689\!\pm\!0.314$ & $9.234\!\pm\!0.314$ \\
\midrule
Toys & Frozen/Frozen & $1.566\!\pm\!0.249$ & $0.773\!\pm\!0.097$ & $0.000\!\pm\!0.000$ & $0.000\!\pm\!0.000$ \\
 & FT/Frozen & $0.964\!\pm\!0.116$ & $0.527\!\pm\!0.050$ & $0.000\!\pm\!0.000$ & $0.000\!\pm\!0.000$ \\
 & Frozen/FT & $1.814\!\pm\!0.079$ & $0.926\!\pm\!0.062$ & $1.062\!\pm\!0.379$ & $0.502\!\pm\!0.193$ \\
 & FT/FT & $1.716\!\pm\!0.158$ & $0.857\!\pm\!0.090$ & $1.213\!\pm\!0.139$ & $0.560\!\pm\!0.067$ \\
 & FT/RT & $1.724\!\pm\!0.192$ & $0.840\!\pm\!0.109$ & $1.426\!\pm\!0.344$ & $0.686\!\pm\!0.123$ \\
 & Reformer & $1.716\!\pm\!0.197$ & $0.858\!\pm\!0.084$ & $2.335\!\pm\!0.430$ & $1.154\!\pm\!0.282$ \\
 & BB (Ours) & $1.972\!\pm\!0.138$ & $0.931\!\pm\!0.044$ & $4.337\!\pm\!0.556$ & $2.103\!\pm\!0.285$ \\
\bottomrule
\end{tabular}
\end{table}

\begin{table}[htbp]
\centering\scriptsize
\setlength{\tabcolsep}{1.5pt}
\caption{Additional comparison with conventional recommenders on complete test sets: Recall and NDCG at 10 (\%).
Neural methods, including BB, use seeds 17/42/2027 and report mean $\pm$
sample s.d. ($n=3$); ItemKNN and EASE are deterministic fits.
All methods use validation-tuned settings, the same test queries and the
same catalog. Training protocols are given in Appendix~\ref{app:default-baselines}.
Bold marks the highest column means.}
\label{tab:conventional-comparison}
\begin{tabular}{lrrrrrr}
\toprule
& \multicolumn{2}{c}{Beauty} & \multicolumn{2}{c}{Tools} & \multicolumn{2}{c}{Toys} \\
Method & R@10 & N@10 & R@10 & N@10 & R@10 & N@10 \\
\midrule
Frozen model & $1.636\!\pm\!0.098$ & $0.775\!\pm\!0.084$ & $2.077\!\pm\!0.075$ & $1.074\!\pm\!0.043$ & $1.059\!\pm\!0.066$ & $0.532\!\pm\!0.003$ \\
Static assignment & $2.122\!\pm\!0.127$ & $1.006\!\pm\!0.105$ & $3.464\!\pm\!0.127$ & $1.913\!\pm\!0.065$ & $1.612\!\pm\!0.188$ & $0.747\!\pm\!0.090$ \\
\midrule
ItemKNN~\citeyearpar{aiolli2013efficient} & $1.391$ & $0.656$ & $1.325$ & $0.738$ & $1.404$ & $0.719$ \\
EASE~\citeyearpar{steck2019ease} & $1.963$ & $1.010$ & $3.265$ & $1.936$ & $1.835$ & $0.968$ \\
SASRec~\citeyearpar{kang2018sasrec} & $2.677\!\pm\!0.088$ & $1.176\!\pm\!0.030$ & $2.935\!\pm\!0.188$ & $1.519\!\pm\!0.032$ & $1.654\!\pm\!0.124$ & $0.756\!\pm\!0.030$ \\
JTM matching (adapted) & $2.059\!\pm\!0.165$ & $0.984\!\pm\!0.116$ & $3.302\!\pm\!0.202$ & $1.799\!\pm\!0.190$ & $1.623\!\pm\!0.060$ & $0.735\!\pm\!0.024$ \\
DREAM voting (adapted) & $2.154\!\pm\!0.165$ & $1.022\!\pm\!0.122$ & $3.470\!\pm\!0.135$ & $1.915\!\pm\!0.067$ & $1.325\!\pm\!0.032$ & $0.614\!\pm\!0.012$ \\
DACT~\citeyearpar{feng2026dact} & $2.095\!\pm\!0.200$ & $1.063\!\pm\!0.068$ & $4.142\!\pm\!0.391$ & $2.260\!\pm\!0.214$ & $1.505\!\pm\!0.066$ & $0.762\!\pm\!0.047$ \\
\midrule
\textbf{BB (Ours)} & $\mathbf{3.217\!\pm\!0.072}$ & $\mathbf{1.535\!\pm\!0.047}$ & $\mathbf{4.956\!\pm\!0.138}$ & $\mathbf{2.806\!\pm\!0.053}$ & $\mathbf{2.154\!\pm\!0.064}$ & $\mathbf{1.027\!\pm\!0.063}$ \\
\bottomrule
\end{tabular}
\end{table}

\begin{table}[htbp]\centering\scriptsize
\caption{Validation-tuned conventional baselines at cutoff 20 (\%). SASRec reports mean $\pm$ sample s.d. over seeds 17/42/2027; the other methods are deterministic.}
\label{tab:default-baselines-at20}
\begin{tabular}{lrrrrrr}\toprule
& \multicolumn{2}{c}{Beauty} & \multicolumn{2}{c}{Tools} & \multicolumn{2}{c}{Toys} \\
Method & R@20 & N@20 & R@20 & N@20 & R@20 & N@20 \\ \midrule
ItemKNN & $2.290$ & $0.883$ & $2.220$ & $0.964$ & $2.474$ & $0.991$ \\
EASE & $3.367$ & $1.363$ & $4.795$ & $2.321$ & $2.777$ & $1.207$ \\
SASRec & $4.394\!\pm\!0.028$ & $1.607\!\pm\!0.023$ & $4.683\!\pm\!0.275$ & $1.963\!\pm\!0.073$ & $2.718\!\pm\!0.037$ & $1.024\!\pm\!0.011$ \\
\bottomrule\end{tabular}\end{table}

\begin{table}[htbp]\centering\scriptsize
\caption{Old- and new-target NDCG@10 for validation-tuned conventional baselines (\%). SASRec reports mean $\pm$ sample s.d. over seeds 17/42/2027; the other methods are deterministic.}
\label{tab:default-baselines-groups}
\begin{tabular}{lrrrrrr}\toprule
& \multicolumn{2}{c}{Beauty} & \multicolumn{2}{c}{Tools} & \multicolumn{2}{c}{Toys} \\
Method & Old & New & Old & New & Old & New \\ \midrule
ItemKNN & $0.758$ & $0.516$ & $0.792$ & $0.198$ & $0.758$ & $1.043$ \\
EASE & $1.315$ & $0.262$ & $2.103$ & $0.267$ & $1.222$ & $0.596$ \\
SASRec & $0.861\!\pm\!0.032$ & $2.713\!\pm\!0.187$ & $1.588\!\pm\!0.034$ & $0.827\!\pm\!0.394$ & $0.761\!\pm\!0.028$ & $1.247\!\pm\!0.246$ \\
\bottomrule\end{tabular}\end{table}

\subsection{Old/new utility across recommendation cutoffs}
\label{app:cohort-cutoff}

Figure~\ref{fig:cohort-cutoff-t5} compares old- and new-target Recall at
$K\in\{5,10,15,20\}$ while keeping each model, map and scoring
configuration fixed. The three datasets retain the same query populations
at every cutoff. Their old/new/all query counts are
5,334/1,486/7,335 (Beauty), 4,873/487/5,360 (Tools) and
4,428/1,099/6,266 (Toys). Future-catalog targets contribute zero to the
all-query metric, which satisfies
\[
 R_{\mathrm{all}}(K)
 = \frac{n_{\mathrm{old}}R_{\mathrm{old}}(K)
          +n_{\mathrm{new}}R_{\mathrm{new}}(K)}{n_{\mathrm{all}}}.
\]

On Tools at cutoff 10, BB raises old-target Recall from Reformer's
3.441\% to 4.084\%, and new-target Recall from 11.910\% to 13.689\%.
BB also exceeds Reformer on both cohorts at cutoff 20.
Newly admitted items can have update-training interactions; these cohorts
distinguish catalog age rather than interaction-free cold start.
The paired component effects are reported in
Table~\ref{tab:completion-all-at10}.

\begin{figure}[!htbp]
\centering
\includegraphics[width=\linewidth]{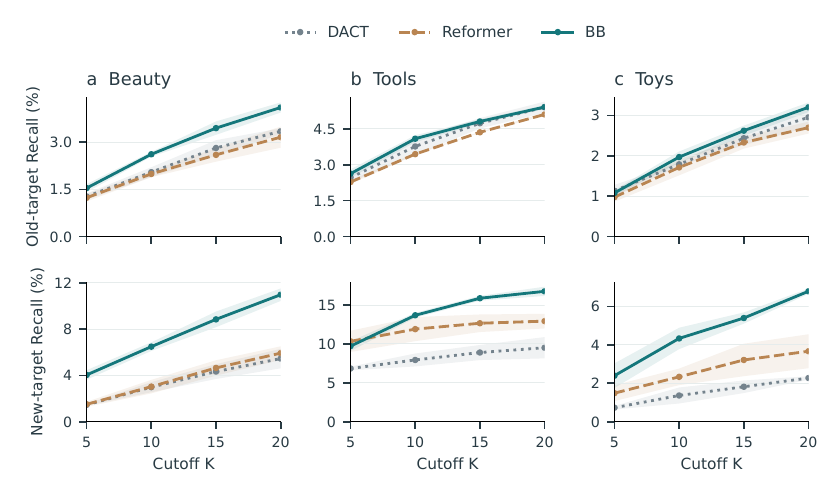}
\caption{T5 cohort utility across cutoffs on Beauty, Tools and Toys.
Top: old-target Recall. Bottom: new-target Recall. All methods use
seeds 17/42/2027; curves and bands show means and sample standard
deviations. Markers give the four evaluated cutoffs, joined by straight
segments. Models, maps, scores and cohort membership are held fixed
as the cutoff changes.}
\label{fig:cohort-cutoff-t5}
\label{fig:cohort-cutoff-t5-beauty}
\label{fig:cohort-cutoff-tools}
\label{fig:cohort-cutoff-t5-toys}
\end{figure}
\FloatBarrier

\section{LC-Rec Backbone: Comparison, Components and Cost}
\label{app:lcrec}
\label{sec:lcrec-results}

This extension examines whether the scoring and completion procedure also
works with a decoder-only recommendation backbone. The full comparison,
component profiles and computational costs are collected here. The study uses seed 17 and
measures scoring and completion on a second backbone; the selected maps
remain unchanged in all three domains.

\subsection{Complete recommendation comparison}
Table~\ref{tab:lcrec-comparison} compares seven methods under the common
catalog update. BB leads the six Recall/NDCG@10 cells, with gains of
15.7--32.9\% and 10.5--35.6\% over the strongest baseline per dataset and
metric. All values use seed 17. The component comparisons below attribute
these gains to scoring and completion; the selected construction maps are
unchanged in all three domains.
\begin{table}[!htbp]
\centering\small
\setlength{\tabcolsep}{5pt}
\caption{Recommendation with the LC-Rec backbone (\%). Single seed 17; the same test queries and available catalogs as the T5 comparison. Bold marks the highest value in each column.}
\label{tab:lcrec-comparison}
\begin{tabular}{lrrrrrr}
\toprule
& \multicolumn{2}{c}{Beauty} & \multicolumn{2}{c}{Tools} & \multicolumn{2}{c}{Toys} \\
Method & R@10 & N@10 & R@10 & N@10 & R@10 & N@10 \\
\midrule
Frozen & 1.159 & 0.631 & 1.474 & 0.812 & 0.958 & 0.414 \\
Fine-tuning & 1.704 & 0.872 & 3.284 & 1.737 & 1.133 & 0.537 \\
LSAT & 1.377 & 0.723 & 1.847 & 1.085 & 1.037 & 0.442 \\
PESO & 1.745 & 0.916 & 3.396 & 1.772 & 1.357 & 0.651 \\
DACT & 2.154 & 1.022 & 3.563 & 1.925 & 1.468 & 0.741 \\
Reformer & 1.909 & 0.926 & 3.526 & 2.049 & 1.053 & 0.501 \\
\midrule
\textbf{BB (Ours)} & \textbf{2.863} & \textbf{1.386} & \textbf{4.123} & \textbf{2.264} & \textbf{1.899} & \textbf{0.858} \\
\bottomrule
\end{tabular}
\end{table}

\subsection{Shared model and update protocol}

The decoder-only comparison uses Qwen2.5-1.5B-Instruct with semantic item
identifiers, following the LC-Rec backbone configuration of
\citet{feng2026dact}. It uses the same Beauty, Tools and Toys examples,
maximum history length and available catalogs as Appendix~\ref{app:protocol}.
All methods start from a shared initial model for each domain and evaluate
the catalog update specified there. The three available catalogs contain
11,778, 10,043 and 11,639 items. Every reported run uses seed 17.

Initial models use rank-8 LoRA with scale 32 and dropout 0.05, applied to
query, value, attention-output and feed-forward up/down projections; input
embeddings and the output head are also trainable. Training uses an 8-bit
base model with bfloat16 computation and full-precision trainable parameters.
AdamW uses learning rate $2\times10^{-5}$, effective batch size 128,
zero weight decay, 200 warmup steps and a cosine schedule. The maximum is
50 epochs. Full validation response cross-entropy is checked approximately
five times per epoch after the first epoch; ten checks without improvement
stop training, and the earliest minimum selects the model. These settings
follow the public DACT recommender implementation; the same optimization
recipe is used to fit the initial models. The initial fits stop after
18.40, 19.00 and 21.81 epochs on Beauty, Tools and Toys, respectively.

Frozen retains the initial model and identifiers. Fine-tuning adapts the
recommender with the shared update mapping. LSAT~\citep{shi2024lsat}
interpolates the historical adapter and the adapted adapter, with weight 0.5.
PESO~\citep{yoo2026peso} uses the released latest-adapter proximal objective
with coefficient 2. DACT uses its updated tokenizer mapping, including
changes to 1,363/1,829/1,110 old-item identifiers. Reformer adds 16 entries
per semantic codebook while retaining existing identifiers. It uses rank-64
LoRA with scale 128 and an initialization preserving the shared parent's
adaptation; DACT retains rank 8 and scale 32. Both reuse the tokenizer
preparation described in Appendices~\ref{app:dact-comparison}
and~\ref{app:continual-updates}. Recommender settings are tuned on validation
from the public implementations; tokenizer preparation follows the protocols above.

BB first adapts the shared parent for five fixed epochs, then evaluates
construction candidates on training examples and continues adaptation.
The paired without-construction model uses the same warmup weights and update
recipe. The selected maps remain unchanged in all three domains, so the
paired models and predictions coincide. Continued adaptation selects its
best validation loss at epochs 5.54/3.00/5.95 and stops at
7.52/4.79/7.93, measured after the five-epoch warm-up.

\subsection{Scoring selection and evaluation}

BB reuses the training-only linear predictor and the correction in
\eqref{eq:calibrated-correction}. Parameters are selected separately for
this backbone. The joint grid uses $\lambda\in\{2,4,6,8\}$ on Beauty/Toys
and $\lambda\in\{0.25,0.5,0.75,1\}$ on Tools, with
$\gamma\in\{0.5,1,1.5\}$ and $b\in\{0,0.5,1,1.5,2\}$.
Each domain evaluates 60 settings with an 80-candidate completion budget,
maximizing validation NDCG@10 followed by Recall@10. Exact ties retain the
previous configuration when available, then minimize parameter distance
from it, followed by smaller parameters. Holding the selected score fixed,
validation compares budgets 80, 160 and 320, preferring the smaller budget
on ties. The selected $(\lambda,\gamma,b,\text{budget})$ are
$(4,1.5,0,80)$, $(0.5,1.5,2,160)$ and $(6,1.5,0,80)$ for
Beauty, Tools and Toys.

Scoring selection uses 1,040 Beauty and 774 Toys cross-day validation
queries and all 3,007 Tools validation queries; model selection by response
loss uses the full validation splits. All three scoring configurations are
fixed before their test evaluation. Inference uses the original full-precision base weights and learned
adapters, a width-40 beam and full-vocabulary log probabilities summed over
four identifier tokens, excluding EOS. Each method re-encodes the same
histories with its own mapping; repeated items remain eligible. The combined
score is used consistently for ranking, candidate bounds and stopping.
As in \eqref{eq:completion-update}, $E_0$ contains all terminal leaves fully
scored by that query's initial search. Completion selects only items outside
the evaluated set, and each selected item counts once toward the additional
budget. Offline likelihood reuse across parameter configurations leaves this
per-query accounting unchanged. The selected budgets are 80/160/80 on
Beauty/Tools/Toys.

\subsection{Component effects and additional metrics}

Table~\ref{tab:lcrec-ablation} reports every paired control.
Collaborative scoring, new-item calibration and candidate completion each
improve NDCG@10 over their corresponding removal in all three domains.
These component gains come from scoring and completion, as the
without-construction control uses the same maps and models.
Section~\ref{sec:count-validation} separately evaluates
constructive feasibility. Collaborative priority slightly exceeds bound
priority on Tools (2.270\% versus 2.264\% NDCG@10) and ties it on the other
two domains. Both use the same final score and candidate budget.

\begin{table}[!htbp]
\centering\scriptsize
\setlength{\tabcolsep}{5pt}
\caption{LC-Rec component controls: Recall and NDCG at 10 (\%), seed 17. Each removal retains the selected parameters and its domain's completion budget; the without-construction control uses the paired map and model.}
\label{tab:lcrec-ablation}
\begin{tabular}{lrrrrrr}
\toprule
& \multicolumn{2}{c}{Beauty} & \multicolumn{2}{c}{Tools} & \multicolumn{2}{c}{Toys} \\
Variant & R@10 & N@10 & R@10 & N@10 & R@10 & N@10 \\
\midrule
BB (Ours) & 2.863 & 1.386 & 4.123 & 2.264 & 1.899 & 0.858 \\
Without construction & 2.863 & 1.386 & 4.123 & 2.264 & 1.899 & 0.858 \\
Without calibration & 2.468 & 1.201 & 3.731 & 1.936 & 1.835 & 0.793 \\
Without collaborative score & 1.922 & 1.010 & 2.761 & 1.465 & 1.293 & 0.618 \\
Without completion & 2.399 & 1.184 & 3.489 & 1.910 & 1.580 & 0.774 \\
Item correction only & 2.590 & 1.267 & 3.563 & 1.949 & 1.787 & 0.804 \\
Collaborative priority & 2.863 & 1.386 & 4.142 & 2.270 & 1.899 & 0.858 \\
Generator only & 1.922 & 1.010 & 2.761 & 1.465 & 1.293 & 0.618 \\
\bottomrule
\end{tabular}
\end{table}

The generator-only Toys NDCG@10 is 0.618\%; a diagnostic ranking by current
update-training target frequency obtains 0.620\% on the same test queries.
The complete procedure reaches 0.858\%, above the generator-only score
and frequency diagnostic on these Toys queries.
Tables~\ref{tab:lcrec-cutoff20} and~\ref{tab:lcrec-groups} give the full
cutoff-20 comparison and BB target-cohort results for seed 17.

\begin{table}[!htbp]
\centering\small
\setlength{\tabcolsep}{5pt}
\caption{LC-Rec Recall and NDCG at 20 (\%), seed 17.}
\label{tab:lcrec-cutoff20}
\begin{tabular}{lrrrrrr}
\toprule
& \multicolumn{2}{c}{Beauty} & \multicolumn{2}{c}{Tools} & \multicolumn{2}{c}{Toys} \\
Method & R@20 & N@20 & R@20 & N@20 & R@20 & N@20 \\
\midrule
Frozen & 1.977 & 0.838 & 2.295 & 1.015 & 1.548 & 0.563 \\
Fine-tuning & 3.095 & 1.217 & 4.776 & 2.115 & 2.139 & 0.787 \\
LSAT & 2.209 & 0.929 & 3.396 & 1.475 & 1.756 & 0.620 \\
PESO & 3.067 & 1.248 & 4.925 & 2.154 & 2.426 & 0.916 \\
DACT & 3.517 & 1.365 & 4.981 & 2.283 & 2.218 & 0.930 \\
Reformer & 3.599 & 1.355 & 4.944 & 2.409 & 2.043 & 0.749 \\
BB (Ours) & 4.513 & 1.802 & 5.690 & 2.662 & 2.968 & 1.128 \\
\bottomrule
\end{tabular}
\end{table}

\begin{table}[!htbp]
\centering\small
\setlength{\tabcolsep}{5pt}
\caption{BB with the LC-Rec backbone by target cohort (\%), seed 17. Future targets lie outside the available catalog and remain misses in the complete-population metrics.}
\label{tab:lcrec-groups}
\begin{tabular}{llrrr}
\toprule
Dataset & Target cohort & Queries & R@10 & N@10 \\
\midrule
Beauty & Old & 5334 & 2.381 & 1.136 \\
Beauty & New & 1486 & 5.585 & 2.762 \\
Beauty & Future & 515 & 0.000 & 0.000 \\
Tools & Old & 4873 & 3.817 & 2.018 \\
Tools & New & 487 & 7.187 & 4.733 \\
Tools & Future & 0 & -- & -- \\
Toys & Old & 4428 & 1.649 & 0.669 \\
Toys & New & 1099 & 4.186 & 2.197 \\
Toys & Future & 739 & 0.000 & 0.000 \\
\bottomrule
\end{tabular}
\end{table}

\subsection{Inference cost}

Table~\ref{tab:lcrec-cost} separates full-test candidate counts from fresh
inference latency on validation queries. The stopping fraction is the
fraction satisfying the score-bound stopping test at the prescribed
numerical tolerance of $10^{-3}$. Queries exhausting their candidate budget
without that condition remain empirical rankings. Latency uses fresh inference
on validation queries; the grid reuses likelihoods across configurations.

\begin{table}[!htbp]
\centering\scriptsize
\setlength{\tabcolsep}{5pt}
\caption{Inference costs of BB with the LC-Rec backbone. Additional candidates and stopping fractions use the full test sets. Latency uses fresh beam decoding and completion for the first 32 validation queries on one RTX 3090, excluding model loading and common input/predictor preparation.}
\label{tab:lcrec-cost}
\begin{tabular}{lrrr}
\toprule
Dataset & Mean additional candidates & Stopping (\%) & Seconds/query \\
\midrule
Beauty & 45.5 & 96.4 & 0.260 \\
Tools & 108.7 & 76.1 & 0.638 \\
Toys & 44.9 & 98.7 & 0.249 \\
\bottomrule
\end{tabular}
\end{table}

\subsection{Old/new utility across cutoffs}
Figure~\ref{fig:lcrec-cohorts} separates old- and new-target Recall using
the cohort definitions in Appendix~\ref{app:cohort-cutoff}. On Tools,
BB improves old-target and all-target utility while new-target Recall
at cutoff 10 remains below DACT; the new-target ordering changes at
cutoff 20. On Toys, BB and DACT tie on old-target Recall at cutoff 10,
while BB has higher new-target Recall. All comparisons use seed 17.

\begin{figure}[!htbp]
\centering
\includegraphics[width=\linewidth]{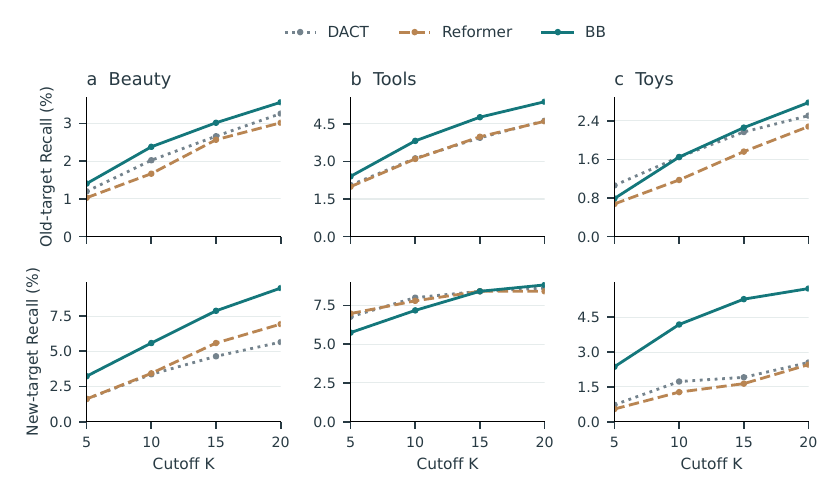}
\caption{LC-Rec cohort utility on Beauty, Tools and Toys, seed 17.
Top: old-target Recall. Bottom: new-target Recall. Models, mappings and
scoring settings stay fixed while the cutoff changes. Markers give the
four evaluated cutoffs, joined by straight segments; no seed variability
interval is estimated. Absolute cutoff-10/20 and component comparisons
are given in Tables~\ref{tab:lcrec-comparison}--\ref{tab:lcrec-groups}.}
\label{fig:lcrec-cohorts}
\label{fig:cohort-cutoff-lcrec-beauty}
\label{fig:cohort-cutoff-lcrec-tools}
\label{fig:cohort-cutoff-lcrec-toys}
\end{figure}
\FloatBarrier

\section{Extended Related Work}
\label{app:extended-related-work}

This section expands Section~\ref{sec:related-work} around the connection
between identifier assignment, candidate access and final ranking that
motivates BB.

\subsection{Identifier updates and constrained repair}

Semantic identifiers make the item mapping part of the retrieval mechanism.
VQ-Rec learns transferable item representations from discrete codes
\citep{hou2023vqrec}. LMIndexer learns semantic identifiers
through document reconstruction and progressive training \citep{jin2024lmindexer}.
TIGER learns residual-quantized identifiers and generates them autoregressively
\citep{rajput2023tiger}. LC-Rec aligns language and collaborative semantics
\citep{zheng2024lcrec}, while LETTER incorporates semantic, collaborative
and diversity objectives into tokenization \citep{wang2024letter}.
These approaches establish the importance of how items are represented
and organized before decoding begins.

Catalog evolution adds the problem of updating this organization while
retaining useful learned structure. Reformer studies incremental
tokenization and recommendation \citep{shi2025incremental}. DACT identifies
collaborative drift and differentiates tokenizer updates, followed by
hierarchical code reassignment \citep{feng2026dact}. SID-Staleness aligns
refreshed codebooks with the previous token space to support warm-start
retriever adaptation \citep{baikalov2026staleness}. GenRecEdit adapts the
generator through token-level editing for cold-start items, with
position-dependent triggering \citep{shen2026genrecedit}.
These different interventions motivate
separating the effects of identifier assignment from changes to the generator.

Score-based assignment provides a particularly close connection.
JTM jointly learns a tree index and preference model, formulating
item-to-leaf assignment through weighted matching \citep{zhu2019joint}.
DREAM constructs collaborative candidate identifiers, uses confidence-weighted
multi-context evidence for commitment, and preserves alternative paths
for inference \citep{guan2026dream}. Such objectives provide evidence for
choosing identifiers. In a finite beam, however, a target's score interacts
with the occupied paths that support or compete with it.

BB makes this interaction explicit through support, terminal-rank and
catalog-count constraints. After a common effective prefix, these constraints
characterize the feasible catalog family for a fixed generator and encoded
history. Minimum-replacement repair then retains the largest feasible
overlap with current leaf occupancy for a specified state and target.
This local construction supplies shared-map proposals, which are evaluated
across re-encoded histories before generator adaptation
(Sections~\ref{sec:repair-guarantees} and~\ref{sec:method-completion}).

The retained-catalog constraint also connects to representation compatibility.
Backward- and forward-compatible embedding methods support retrieval across
model versions \citep{shen2020bct,ramanujan2022fct}; here, compatibility fixes
old item-to-identifier assignments and preserves their recommendation
eligibility. Possible-worlds ranking studies answers over families of
admissible data instances \citep{feng2023uncertainranking}. Our family varies
catalog occupancy, which changes prefix support and hence the beam-search
trajectory itself.

\subsection{Reachability and item-level scoring}

Several studies identify why a valid identifier may still fail to retrieve
its item. Temporal cold-item analysis relates retrieval to observed token
and prefix support \citep{peng2026cold}. SIDScope distinguishes mapping
properties, generated-path survival and item resolution
\citep{ding2026sidscope}. HCGRec addresses unreachable reward groups during
post-training by supplying target-prefix hints to hard training examples
\citep{zhang2026hcgrec}. Together, these studies motivate examining the
intermediate decisions between identifier assignment and item-level ranking.

Training and architectural interventions address different parts of this
process. Beam-aware tree learning studies calibration under the deployed
search procedure \citep{zhuo2020beam}. BEAR regularizes token ranks to reduce
premature beam pruning \citep{yang2026bear}, while Latte introduces a latent
token to relax structural coupling between generation paths
\citep{hou2026latte}. BB studies how legal assignment changes alter the
support and competition encountered by a fixed decoder, then uses the
resulting repairs in shared-map construction and adaptation.

The inference stage addresses the remaining query-dependent candidate
access problem. Linear collaborative models provide item-level evidence
from interactions \citep{steck2019ease}. BB fits a supervised history-to-next-item
ridge predictor and combines its calibrated correction with identifier
likelihood. The correction is available before full identifier evaluation,
so it can guide both final ranking and the evaluation of omitted items.
This connects the ranking objective to candidate completion: reranking
alone cannot recover an item outside the evaluated pool.

\subsection{Candidate access and search guarantees}

Inference can expand access without changing the item mapping.
SpecGR uses an inductive drafter to propose items, verifies them with
generative likelihoods, and guides subsequent drafts using generated
prefixes \citep{ding2026specgr}. Its adaptive exit is triggered when enough
candidates pass a verification threshold; accepted items are ranked by
verifier scores. BB links candidate priority and stopping to the combined
generator--collaborative score. The remaining question is whether an
unevaluated catalog item can still displace the current $K$th item.

This question connects to established search principles.
\citet{huang2017finish} derive optimal stopping modulo beam size for neural
generation, including bounded length rewards. \citet{meister2020bestfirst}
use score monotonicity for best-first beam search and early pruning.
BB applies the prefix upper-bound principle to fixed catalog identifiers
under the final ranking score. Adding the known item correction to an
evaluated prefix score bounds the complete score because the remaining
log-probability terms are nonpositive.

The bound covers every unevaluated catalog item, including identifiers
whose paths left the initial beam. Additional evaluations are ordered by
these bounds, and the strict stopping condition certifies the global
Top-$K$ under the combined score (Appendix~\ref{app:completion-proof}).
An exhausted evaluation budget can instead return an uncertified ranking.
Thus constructive repair organizes the shared catalog, while completion
uses query-specific evidence to recover promising omitted candidates and
determine when further evaluation cannot change the result.

\end{document}